\documentclass[a4paper,11pt,reqno]{article}
\usepackage{fullpage}
\usepackage[utf8]{inputenc}
\usepackage{mathrsfs}
\usepackage{dsfont}
\usepackage{hyperref}
\usepackage{amsmath}
\usepackage{amssymb}
\usepackage{amsthm}
\usepackage{amsfonts}
\usepackage{amstext}
\usepackage{amsopn}
\usepackage{amsxtra}
\usepackage{mathrsfs}
\usepackage{dsfont}
\usepackage{esint}
\usepackage{enumitem}
\usepackage[capitalize]{cleveref}
\usepackage{braket}
\usepackage{stmaryrd} 
\usepackage{tikz}
\usetikzlibrary{patterns}
\usepackage{float}
\usepackage{mathtools}

\newtheorem{lemma}{Lemma}
\newtheorem{theorem}[lemma]{Theorem}
\newtheorem{proposition}[lemma]{Proposition}

\DeclareRobustCommand{\abbrevcrefs}{%
\crefname{proposition}{Prop.}{Props.}%
}
\DeclareRobustCommand{\cshref}[1]{{\abbrevcrefs\cref{#1}}}

\newcommand{\R}{\mathbb{R}}

\newcommand\1{{\ensuremath {\mathds 1} }}

\newcommand{\cM}{\mathcal{M}}

\newcommand{\cB}{\mathcal{B}}

\newcommand{\cE}{\mathcal{E}}
\newcommand{\cF}{\mathcal{F}}
\newcommand{\cN}{\mathcal{N}}

\newcommand{\cH}{\mathcal{H}}
\newcommand{\cU}{\mathcal{U}}

\newcommand{\tr}{{\rm Tr}\,}

\renewcommand{\geq}{\geqslant}
\renewcommand{\leq}{\leqslant}
\renewcommand{\hat}{\widehat}
\renewcommand{\tilde}{\widetilde}

\newcommand{\eps}{\varepsilon}
\usepackage{color}

\newcommand{\nn}{\nonumber}

\newcommand{\dd}{\mathrm{d}}

\newcommand{\ch}{\mathrm{ch}}
\newcommand{\sh}{\mathrm{sh}}

\newcommand{\ad}{\mathrm{ad}}

\newcommand{\ao}{\mathfrak{a}}
\newcommand{\hc}{\mathrm{h.c.}}
\newcommand{\Ran}{\mathrm{Ran}}
\newcommand{\vphi}{\varphi}
\newcommand{\vep}{\varepsilon}

\newcommand{\T}{\widetilde{T}}
\newcommand{\simS}{\sim_{\raisebox{-0.3ex}{\scriptsize S}}}

\newcommand{\wt}{\widetilde}

\newcommand{\bZ}{\mathbb{Z}} 
\newcommand{\bR}{\mathbb{R}} 
\newcommand{\bN}{\mathbb{N}} 
\newcommand{\ph}{\mathbb{\varphi}} 

\newcommand{\weyl}{W_{\hspace{-.06cm} N_0}}

\numberwithin{equation}{section}

\title{Upper Bound for the Free Energy of \\ Dilute Bose Gases at Low Temperature}

\author{Florian Haberberger\thanks{Department of Mathematics, LMU Munich, Theresienstrasse 39, 80333 Munich, Germany.} , Christian Hainzl$^*$, 
Benjamin Schlein\thanks{Institute of Mathematics, University of Zurich, Winterthurerstrasse 190, 8057 Zurich, Switzerland.
} , 
Arnaud Triay$^*$ }

\begin{document}
\maketitle

\begin{abstract}
We consider a Bose gas at density $\rho > 0$, interacting through a repulsive potential $V \in L^2 (\bR^3)$ with scattering length $\mathfrak{a} > 0$. We prove an upper bound for the free energy of the system, valid at low temperature $T \lesssim \rho \mathfrak{a}$. Combined with the recent lower bound obtained in \cite{HabHaiNamSeiTri-23}, our estimate resolves the free energy per unit volume up to and including the Lee--Huang--Yang order $\mathfrak{a} \rho^2  (\rho \mathfrak{a}^3)^{1/2}$.
\end{abstract} 
 
 \section{Introduction} 

\subsection{Setting and main result}

Consider a gas of $n$ bosons, moving in a periodic box $\Omega$ and interacting through a non-negative, radially symmetric and compactly supported potential $V\in L^2(\R^3)$ with scattering length $\mathfrak{a} > 0$. The Hamiltonian of the system has the form 
\begin{equation}\label{eq:H1} H_n = \sum_{i=1}^n -\Delta_{x_i} + \sum_{1\leq i<j\leq n} V (x_i - x_j) \end{equation} 
and, according to Bose statistics, it acts on the Hilbert space $L^2_s (\Omega^n)$, the subspace of $L^2 (\Omega^n)$ consisting of functions that are symmetric with respect to permutations of the $n$ particles. The canonical free energy of the gas at temperature $T \geq 0$ is defined by 
\begin{equation} \label{eq:FNL} F_\Omega (n) = \inf_\Gamma \; \big[ \tr H_n \Gamma - T S( \Gamma) \big],  \end{equation}
where the infimum is taken over all density matrices $\Gamma$ on $L^2_s (\Omega^n)$, i.e. all non-negative operators $\Gamma$ on $L^2_s (\Omega^n)$ with $\tr \Gamma = 1$. Moreover, $S (\Gamma) = - \tr \Gamma \log \Gamma$ denotes the entropy of $\Gamma$.
We are interested in the thermodynamic free energy per unit volume at density $\rho > 0$ and at temperature $T$, which is defined by the limit 
\begin{equation}\label{eq:frhot} f (\rho , T) = \lim_{\substack{n, |\Omega| \to \infty : \\ n / |\Omega| = \rho}} \frac{F_\Omega (n)}{|\Omega|}, \end{equation} 
where $|\Omega |$ denotes the volume of the box $\Omega$. The existence of the thermodynamic limit and its convexity in $\rho$ are well-known; see \cite{Ruelle}.

At temperature $T= 0$, the free energy (\ref{eq:FNL}) is the ground state energy of the Hamiltonian (\ref{eq:H1}). In the thermodynamic limit, Lee--Huang--Yang predicted in \cite{LHY-57} that the ground state energy per unit volume is given by 
\begin{equation}\label{eq:LHY} e (\rho) = f (\rho, 0) = 4 \pi \mathfrak{a} \rho^2 \left( 1 + \frac{128}{15 \sqrt{\pi}} (\rho \mathfrak{a}^3)^{1/2} + \dots \right) \end{equation} 
up to contributions of lower order in $\rho$ in the dilute limit $\rho \mathfrak{a}^3 \to 0$. To leading order, the validity of (\ref{eq:LHY}) was rigorously established in \cite{Dyson-57} and \cite{LieYng-98}, as an upper and, respectively, lower bound.
A lower bound capturing the second order corrections on the right-hand side of (\ref{eq:LHY}) has been shown in \cite{FouSol-20} (for integrable potentials) and in \cite{FouSol-23} (for potentials including a hard-core). The matching upper bound was proved in \cite{YauYin-09} (for sufficiently regular $V$); recently, a simpler proof of the upper bound (for $V \in L^3 (\bR^3)$)
was obtained in \cite{BCS} (see also \cite{Basti-22}).
For hard-core potentials, the derivation of an upper bound confirming (\ref{eq:LHY}) to second order is still an open question; in this setting, the best available estimate, matching (\ref{eq:LHY}) up to an error  comparable with the predicted second order term, has been recently derived in  \cite{BCGOPS-22}. 

For $T>0$, the free energy (\ref{eq:FNL}) is attained by the Gibbs state $ {\rm z}^{-1} e^{-H_n/T}$ with the normalization constant ${\rm z} = \tr e^{-H_n/T}$. In \cite{LHY-57}, it is predicted that in the dilute limit and at low energy, $H_n$ can be approximated by a linear combination of uncoupled harmonic oscillators, labeled by the momentum $p$, with the dispersion $\eps (p) = \sqrt{|p|^4 + 16 \pi \rho \mathfrak{a} \, p^2}$. Under this assumption, a simple ideal gas computation leads to the expectation that 
\begin{equation}\label{eq:frhot-low} 
\begin{split} 
f (\rho, T) = \; &4 \pi \mathfrak{a} \rho^2 \left( 1 + \frac{128}{15 \sqrt{\pi}} (\rho \mathfrak{a}^3)^{1/2}  \right) + \frac{T^{5/2}}{(2\pi)^3} \int_{\bR^3} \log \left( 1  - e^{- \sqrt{|p|^4 + \frac{16 \pi\rho \mathfrak{a}}{T}  p^2 }} \right) dp  + \ldots
\end{split} 
 \end{equation} 
as $\rho \mathfrak{a}^3 \to 0$. Up to errors that are small with respect to $\rho^2 \mathfrak{a}$, the validity of (\ref{eq:frhot-low}) was established in \cite{Seiringer-08,Yin-10}, for temperatures $T \lesssim T_c(\rho)$, with $T_c (\rho) \sim \rho^{2/3}$ denoting the critical temperature of the free Bose gas. In the present paper, we aim at resolving (\ref{eq:frhot-low}) up to the Lee--Huang--Yang order $\rho^2 \mathfrak{a} (\rho \mathfrak{a}^3)^{1/2}$. To this end, we are going to consider temperatures $T \lesssim \rho \ao$, for which the thermal contribution is, at most, comparable with the Lee--Huang--Yang term. 

As a lower bound, a rigorous derivation of (\ref{eq:frhot-low}) to the order $\rho^2 \mathfrak{a} (\rho \mathfrak{a}^3)^{1/2}$ was recently obtained in \cite{HabHaiNamSeiTri-23} for all $T \lesssim \rho \mathfrak{a}$. The goal of the present work is to show an upper bound matching the lower bound of \cite{HabHaiNamSeiTri-23} and therefore proving the formula  (\ref{eq:frhot-low}) for all $T \lesssim \rho \mathfrak{a}$.

Let us elaborate on the choice of the temperature. When $T \simeq \rho \ao$, the typical momenta of thermal excitations $T^{1/2} \simeq (\rho \ao)^{1/2} = \ell_{\rm GP}^{-1}$ is of the order of the inverse of the healing length, also known as the Gross--Pitaevskii length, which is also the typical order of momenta of the excitations responsible for the Lee-Huang-Yang correction at $T=0$.
In fact, as we will explain in more detail later, the study of the thermodynamic limit (\ref{eq:frhot}) often relies on the analysis of subsystems of fixed size. 
The Gross--Pitaevskii length $\ell_\text{GP}$ is also the shortest length on which the systems exhibits the correct second order behavior. 
The free energy is then given by an expansion similar to (\ref{eq:frhot-low}), as it has already been proven in \cite{DeuSeiYng-19,DeuSei-20,BocDeuSto-23,BocDeuSto-24,CapDeu-23}.  

For the proof of the lower bound of (\ref{eq:frhot-low}) in \cite{HabHaiNamSeiTri-23}, it was sufficient to consider subsystems of size slightly larger than $\ell_\text{GP}$. For the upper bound, on the other hand, we need to work on boxes with length of the order 
$\mathfrak{a} (\rho \mathfrak{a}^3)^{-\gamma}$ for some $\gamma > 1$,
which is much larger than $\ell_\text{GP}$. This is a consequence of the fact that the effect of Dirichlet boundary conditions, which are relevant to show upper bounds, on the free energy is substantially stronger than the effect of Neumann boundary conditions  (relevant for lower bounds). From this point of view the analysis is more involved and consequently new techniques are required.

The following theorem is our main result. 

\begin{theorem}\label{thm:main} 
Let $V \in L^2 (\bR^3)$ be non-negative, spherically symmetric and compactly supported. Let $0\leq T \leq C\rho \mathfrak{a}$ for some $C>0$. Then there exist $c, \epsilon > 0$ such that
\begin{equation} \label{eq:main} 
\begin{split} 
f(\rho,T) \leq \; &4\pi\mathfrak{a}\rho^2\left(1 + \frac{128}{15\sqrt{\pi}} (\rho \mathfrak{a}^3)^{1/2} + c (\rho \mathfrak{a}^3)^{1/2+\epsilon} \right) + \frac{T^{5/2}}{(2\pi)^3} \int_{\R^3} \log \left(1-e^{-\sqrt{|p|^4+\frac{16\pi\rho\mathfrak{a}}{T} p^2} }\right) \dd p. \end{split}
\end{equation} 
\end{theorem}

Similarly as for the result on the lower bound \cite{HabHaiNamSeiTri-23}, our proof also allows to take slightly larger temperatures, that is $T \leq C \rho \ao (\rho\ao^3)^{-\nu}$ for some $\nu>0$ small enough. For the sake of simplicity we restrict our attention to $T \leq C \rho \ao$.

\subsection{Localization into small boxes and grand canonical ensemble}

To show Theorem \ref{thm:main}, it is convenient to consider a system confined in a smaller box $\Lambda_L = [-L/2 ; L/2]^3$, with $L = \rho^{-\gamma}$ and with periodic boundary conditions, in the grand canonical setting, where the number of particles is allowed to fluctuate. If $\gamma > 1$, we will show that an upper bound on the small periodic box $\Lambda_L$ implies an upper bound in the thermodynamic limit by patching small boxes together, see \cref{fig:boxes}. To work in the grand canonical ensemble, we introduce the bosonic Fock space 
\[ \cF (\Lambda_L) = \bigoplus_{k \geq 0} L^2 (\Lambda_L)^{\otimes_s k} \]
and consider the Hamiltonian
\begin{equation} \label{eq:cHL} \cH = \sum_{p \in \frac{2\pi}{L} \bZ^3 } p^2 a_p^* a_p + \frac{1}{2} \sum_{p,q,r \in \frac{2\pi}{L} \bZ^3 } \widehat{V} (r) a_{p+r}^* a_q^* a_{q+r} a_p \end{equation} 
on $\cF (\Lambda_L)$.
Here, for a momentum $p \in \frac{2\pi}{L}\bZ^3 $, $a_p^*$ and $a_p$ are the usual creation and annihilation operators satisfying canonical commutation relations. 

Localization to the small periodic box $\Lambda_L$ is achieved through the following standard proposition, which is proven similarly to \cite[Prop. 1.2]{BCS}, see also \Cref{fig:boxes}. For the sake of completeness, we sketch the proof in Appendix \ref{app:loc}, stressing the changes due to the fact that we work at positive temperature.
\begin{proposition}[Comparison to small periodic boxes]\label{prop:localization}
Let $0 < R < \ell < L$ such that $V(x) = 0$ for all $|x| \geq R$. Let $\Gamma_L$ be a density matrix on the Fock space $\cF (\Lambda_L)$, i.e. a non-negative operator on $\cF (\Lambda_L)$ with $\tr \, \Gamma_L = 1$, satisfying periodic boundary conditions and the bound $\tr \cN \Gamma_L < \infty$. Let
\begin{equation}\label{eq:cond-N} 
\tilde\rho:= \frac{1}{(L+2\ell+R)^3}\tr \, \cN \Gamma_L.
\end{equation}
Then there exists a constant $c>0$ such that 
\begin{equation}\label{eq:localization} f (\tilde \rho, T) \leq \frac{1}{(L+2\ell+R)^3} \Big[ \tr \, \cH \Gamma_L - T S(\Gamma_L) \Big] + \frac{c}{L\ell} \tilde \rho
\end{equation}
for all $T \leq C \rho \mathfrak{a}$.
\end{proposition}

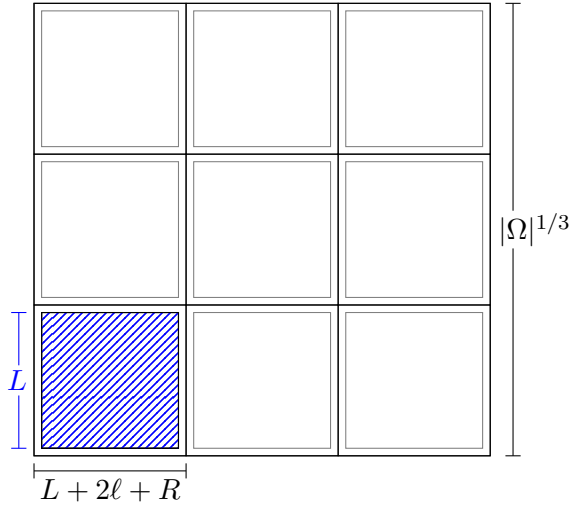
\begin{figure}[H] 
\centering
     \begin{tikzpicture}[scale=1]
        \draw (0,0) rectangle (6,6);
        \foreach \i in {0,2,4} {
            \foreach \j in {0,2,4} {
                \draw (\i,\j) rectangle (\i+2,\j+2);
               \draw[gray,thin] (\i+0.1,\j+0.1) rectangle (\i+1.9,\j+1.9);
            }
          \draw[pattern=north east lines, pattern color=blue] (0.1,0.1) rectangle (1.9,1.9); 
        }
        \draw (0,-0.2) -- (2,-0.2);
        \node at (1,-0.45) {\(L+2\ell+R\)};
        \draw (0,-0.1) -- (0,-0.3);
        \draw (2,-0.1) -- (2,-0.3);
        \draw (6.3,0) -- (6.3,2.75);
        \draw (6.3,3.25) -- (6.3,6);
        \node at (6.2,3) {\(\qquad |\Omega|^{1/3}\)}; 
        \draw (6.2,0) -- (6.4,0);
        \draw (6.2,6) -- (6.4,6);
        \draw[thin,blue] (-0.2,0.1) -- (-0.2,0.8);
        \draw[thin,blue] (-0.2,1.2) -- (-0.2,1.9);
        \node[thin,blue] at (-0.22,1) {\(L\)}; 
        \draw[thin,blue] (-0.1,0.1) -- (-0.3,0.1);
        \draw[thin,blue] (-0.1,1.9) -- (-0.3,1.9);
    \end{tikzpicture}
    
\caption{We obtain a trial state on the thermodynamic box $\Omega$ by first taking the periodic state $\Gamma_L$ on $\Lambda_L$ from \cshref{prop:loc-bound}
and placing it in a slightly bigger box of sidelength $L+2\ell$. On this box, we modify $\Gamma_L$ such that it satisfies Dirichlet boundary conditions, while also not changing its free energy significantly. Finally, we copy this state into boxes of sidelength $L+2\ell+R$, where the slight increase of the box sidelength by $R$ prevents interactions among them.}
\label{fig:boxes}
    \end{figure}

Taking into account \cshref{prop:localization}, our main challenge is the construction of the trial state $\Gamma_L$ for the free energy functional on the small periodic box $\Lambda_L$ with $L = \rho^{-\gamma}$. As discussed above, to make sure that the localization error on the right-hand side of (\ref{eq:localization}) is smaller than the resolution we want to achieve, we will need to choose $\gamma > 1$.
In the next proposition, whose proof occupies the bulk of the paper, we construct the trial state $\Gamma_L$. Its particle density is slightly bigger than the desired density $\rho$ in order to compensate for the extension of the box $\Lambda_L$ that was needed in \cshref{prop:localization} to switch to Dirichlet boundary conditions and to avoid interactions among boxes.
\begin{proposition}[Trial state on small periodic boxes] \label{prop:loc-bound}
For ${\rho} > 0$ and $\gamma > 1$ we set $L = {\rho}^{-\gamma}$. Then, if $\gamma$ is small enough,  there exists $\epsilon > 0$ and a density matrix $\Gamma_L$ on $\cF (\Lambda_L)$, i.e. a non-negative operator on $\cF (\Lambda_L)$ with $\tr \, \Gamma_L = 1$, that satisfies periodic boundary conditions and is such that 
\begin{equation}\label{eq:prop3-NN2} c_1 \rho^{(\gamma+2)/2} \leq  \frac{1}{L^3}\tr \, \cN \Gamma_L - \rho  \leq c_2 \rho^{3/2}  \end{equation} 
for some $c_1,c_2>0$, and
\begin{equation}\label{eq:main-eq} \begin{split} \frac{1}{L^3} \Big[ \tr \, \cH \Gamma_L - T S (\Gamma_L) \Big] \leq \; &4 \pi \mathfrak{a} {\rho}^2 \Big( 1 + \frac{128}{15 \sqrt{\pi}} ({\rho} \mathfrak{a}^3)^{1/2} + c ({\rho} \mathfrak{a}^3)^{1/2 +\epsilon} \Big) \\ &+ \frac{T^{5/2}}{(2\pi)^3} \int_{\R^3} \log \left(1- e^{-\sqrt{|p|^4 + \frac{16 \pi {\rho} \mathfrak{a}}{T}  p^2 }} \right) dp  \end{split} \end{equation}
for some $c>0$ and for all $T \leq C \rho \mathfrak{a}$.
\end{proposition}

Combining Proposition \ref{prop:localization} and \ref{prop:loc-bound} we are ready to prove Theorem \ref{thm:main}.

\subsection{Proof of Theorem \ref{thm:main}}
Let $\gamma>1$ small enough. Putting $L = {\rho}^{-\gamma}$, we conclude from \cshref{prop:loc-bound} that there exists a density matrix $\Gamma_L$ on $\cF (\Gamma_L)$, satisfying periodic boundary conditions and the above estimates (\ref{eq:prop3-NN2}) and (\ref{eq:main-eq}). Setting now $\ell = L^\alpha = \rho^{-\gamma \alpha}$ for some $\alpha \in (0,1)$ and defining $\tilde\rho$ as in \eqref{eq:cond-N}, \cshref{prop:localization} implies that 
\[ \begin{split}  f(\tilde \rho , T) \leq \; &4 \pi \mathfrak{a} {\rho}^2 \Big( 1 + \frac{128}{15 \sqrt{\pi}} ({\rho} \mathfrak{a}^3)^{1/2}  \Big) + \frac{T^{5/2}}{(2\pi)^3} \int_{\R^3} \log \left(1- e^{-\sqrt{|p|^4 + \tfrac{16 \pi \mathfrak{a} {\rho}}{T} p^2}} \right) dp \\ &+ C \big[ \rho^{5/2+\epsilon} +  \rho^{1+\gamma +\alpha \gamma} \big]. \end{split} \]
Moreover, by \eqref{eq:prop3-NN2} there exists $c_1>0$ such that
\begin{align*}
\tilde{\rho} = \frac{\tr \cN \Gamma_L}{(L+2\ell+R)^3} \geq (\rho + c_1 \rho^{(\gamma+2)/2}) (1 - C \rho^{\gamma -\alpha \gamma}) \geq \rho + c_1 \rho^{(\gamma+2)/2} - C \rho^{1+\gamma - \alpha \gamma} \geq \rho
\end{align*}
for $\rho$ small enough, if we assume $\alpha < 1/2$. Here, we used that $RL^{-1} \leq \ell L^{-1} = \rho^{\gamma - \alpha \gamma}$.

Using that $r \mapsto f(r,T)$ is convex (see \cite{Ruelle}), that $f(0,T) = 0$ and that $f(\tilde \rho,T) \geq 0$ for our range of parameters (see \cite{HabHaiNamSeiTri-23}), we have $f(\rho,T) \leq f(\widetilde{\rho},T)$.
Taking into account that $\gamma > 1$ and choosing $\alpha = 1/2 - \eta$ for a sufficiently small $\eta > 0$, this concludes the proof of Theorem~\ref{thm:main}.
\qed
 
\subsection{Strategy of the proof of Proposition \ref{prop:loc-bound}}
 
The remainder of the paper is devoted to the proof of \Cref{prop:loc-bound}. In order to show (\ref{eq:main-eq}), it is convenient to rescale variables $x_j \mapsto x_j / L$. Particles then move in the unit torus $\Lambda$, i.e. the unit box $[-1/2,1/2]^3$ with periodic boundary conditions. For convenience, we introduce $N := {\rho} L^3 = {\rho}^{1-3\gamma}$ or, equivalently,
$L = 
N^{1-\kappa},
$
with $\kappa = (2\gamma-1)/(3\gamma-1)$. The constraint on the temperature becomes $T \leq C N^{-2+3\kappa}$. We conclude that the Hamilton operator  \eqref{eq:cHL}, acting on $\cF (\Lambda_L)$, is unitarily equivalent to 
$L^{-2} \cH_N
$, with the new Hamiltonian 
\begin{equation}\label{eq:cHN} \cH_N = \sum_{p \in \Lambda^*} p^2 a_p^* a_p + \frac{1}{2} \sum_{p,q,r \in \Lambda^*} \widehat{V}_N (r) a_{p+r}^* a_q^* a_{r} a_p \end{equation} 
acting on $\cF (\Lambda)$. Here, we introduced the notation $\widehat{V}_N (r) = N^{-1+\kappa} \widehat{V} (r / N^{1-\kappa})$ and we denoted by $\Lambda^* = 2\pi \bZ^3$ the set of momenta on the torus $\Lambda$. The condition $\gamma > 1$ is equivalent to $\kappa > 1/2$. 

In the next sections, we will construct a density matrix $\Gamma_N$ on $\cF (\Lambda)$ that has the correct expected number of particles and the correct energy. That is, we have the bound
\begin{equation}\label{eq:trNGN} c_1 N^{3\kappa /2 - (\kappa-1/2)} \leq \tr \cN \Gamma_N - N \leq c_2 N^{3\kappa/2} , \end{equation}
which is equivalent to (\ref{eq:prop3-NN2}) and also 
\begin{equation}\label{eq:claim-GN}
\begin{split}  
&\tr \cH_N \Gamma_N - T L^2 S (\Gamma_N) 
\\
& \leq 4\pi\mathfrak{a} N^{1+\kappa}  \left( 1 + \frac{128}{15 \sqrt{\pi}} ( N^{-2+3\kappa}\mathfrak{a}^{3} )^{1/2}  \right)  
\\ & \hspace{.55cm} + TN^{2-2\kappa} \hspace{-.6cm} \sum_{N^{\kappa/2-\eps} < |p| \leq N^{\kappa/2+\eps}} \hspace{-.5cm} \log\left[1 - \exp\Big(-\frac{1}{T N^{2-2\kappa}}\sqrt{|p|^4 + 16 \pi \mathfrak{a} N^\kappa p^2}\;\Big)\right] + C N^{5\kappa/2-\eps/2}  
\end{split} 
\end{equation} 
for all $T \leq C N^{-2 + 3\kappa}$ and some fixed $\eps > 0$. After rescaling and approximating the sum with an integral, we will obtain the statement of \cshref{prop:loc-bound}. Note that, after rescaling to the unit torus, the effective temperature is given by $TL^2 = TN^{2-2\kappa}$.

To construct the trial state $\Gamma_N$, we will apply a strategy similar to the one used in \cite{BCS} to prove an upper bound for the ground state energy at temperature $T=0$.
The work \cite{BCS} was partly motivated by the results of \cite{BBCS1,BBCS2}, establishing the validity of Bogoliubov theory in the Gross--Pitaevskii regime, where the Hamiltonian takes the form (\ref{eq:cHN}), with $\kappa=0$ (see \cite{HST,B} for alternative approaches). The main observation is that \eqref{eq:cHN} can be reduced, through a series of unitary transformations, to a Hamilton operator that, on the range of our trial state, can be approximated by a linear combination of decoupled harmonic oscillators. 

Let us explain the procedure a bit more precisely; the mathematical details will follow in the next section. First of all, since the temperatures under consideration are below the critical temperature, we need to factor out the Bose--Einstein condensate. 
To this end, we conjugate (\ref{eq:cHN}) with a Weyl operator $W_{N_0}$, producing a condensate with $N_0$ particles. The parameter $N_0 \in \bR_+$ will be fixed later on; it will be chosen so that $0 \leq N - N_0 \lesssim N^{3\kappa/2}$.  
Let us introduce the momentum sets
\begin{equation}\label{eq:mom-sets} \begin{split} \textrm{High momenta:} \quad H &= \{ p \in \Lambda^* : |p| > N^{1-\kappa-\eps} \},  \\
\textrm{Shell momenta:} \quad S &= \{ p \in \Lambda^* : N^{\kappa/2-\eps} < |p| \leq N^{\kappa/2+\eps} \} \end{split} \end{equation} 
with a parameter $\vep>0$ that is chosen so that the sets do not overlap, i.e. $-2+3\kappa+4\vep<0$.
To approximate $\weyl^* \cH_N \weyl$ with a quadratic Hamiltonian, we first remove the short-range correlation structure. This renormalization is achieved by conjugating $ \weyl^* \cH_N \weyl $ with a Bogoliubov transformation $e^{\cB_1}$, acting only on momenta $|p| > N^{\kappa/2+\vep}$. 
Unfortunately, to reach the correct energy, quasi-free states constructed by quadratic transformations are not enough. Therefore, we conjugate the resulting renormalized excitation Hamiltonian with a unitary cubic operator $T_c$; the choice of $T_c$ is the crucial step in our analysis. 
The Hamiltonian $T_c^* e^{-\cB_1}  \weyl^* \cH_N \weyl e^{\cB_1} T_c$ is well approximated by a quadratic operator, which can be diagonalized by a second Bogoliubov transformation $e^{\cB_2}$, acting on momenta in the shell $S$. Up to error terms, that are negligible in an appropriate sense, we find 
\begin{equation}\label{eq:approx}  \begin{split}  
e^{-\cB_2} T_c^* e^{-\cB_1} &\weyl^* \cH_N \weyl e^{\cB_1} T_c e^{\cB_2}  \\ &\simeq 4\pi \mathfrak{a} N^{1+\kappa} \left( 1 + \frac{128}{15\sqrt{\pi}} ( N^{-2+3\kappa}\mathfrak{a}^3)^{1/2} \right) + \sum_{p \in S} \sqrt{|p|^4 + 16 \pi \mathfrak{a} N^\kappa p^2} \, a_p^* a_p, \end{split} 
\end{equation} 
provided that $\kappa >1/2$ is small enough. At temperature $T N^{2-2\kappa}$, the Gibbs state associated with the quadratic operator on the right hand side of (\ref{eq:approx}) has the form 
\begin{equation}\label{eq:Gamma0_intro}
\Gamma_0 = Z^{-1} \mathds{1}_{\{\cN_{S^c} = 0\}} \, \exp\bigg( -\frac{1}{T N^{2-2\kappa}} \sum_{p \in S} \sqrt{|p|^4 + 16 \pi \mathfrak{a} N^\kappa p^2} \, a_p^* a_p\bigg)
\end{equation}
with the normalization constant $Z>0$ chosen so that $\tr \Gamma_0 = 1$ and where, for any $F \subset \Lambda^*$, we introduced the notation 
\[ \cN_F = \sum_{p \in F} a_p^* a_p \]
for the operator measuring the number of particles with momentum in the set $F$.
Therefore, we use 
\begin{equation}\label{eq:gamma_N}
\Gamma_N = Z^{-1} \weyl e^{\cB_1} T_c e^{\cB_2} \Gamma_0 e^{-\cB_2} T_c^* e^{-\cB_1} \weyl^*
\end{equation}
as a trial state for the Hamiltonian (\ref{eq:cHN}). With the approximation (\ref{eq:approx}), it is then not difficult to check that this trial state has indeed the correct free energy, completing the proof of (\ref{eq:claim-GN}). 


We now highlight the main novelties with respect to \cite{BCS}, where a trial state was constructed to estimate the ground state energy of (\ref{eq:cHN}). Since we work at positive temperature, $\Gamma_0$ must be chosen as a mixed state that describes thermal excitations. In contrast, at zero temperature it is sufficient to take $\Gamma_0$ as the projection onto the vacuum vector in the Fock space $\cF (\Lambda)$. The analysis of \cite{BCS} strongly relied on the fact that the cubic transformation that is needed to create the correct correlation structure acted directly on the vacuum. For this reason, it was possible to implement this transformation through a non-unitary operator, given by the exponential of a cubic expression involving only creation operators, and to compute its action almost explicitly. Here, we follow a different strategy and implement a unitary transformation mainly for two reasons. First, the action of the cubic transformation on $\Gamma_0$ is more involved than its action on the vacuum and explicit computations are the exception. Second, to estimate the entropy of $\Gamma_N$ we need to compute the spectral distribution of the transformed state, which is considerably simpler if $T_c$ is unitary.

Let us briefly explain why it is challenging to control the action of a unitary cubic transformation in the current setting. The action of a unitary operator $e^\cB$ can be formally computed via the commutator expansion
\begin{align}\label{duhamel:intro}
e^{-\mathcal B} X e^{\mathcal B} = \sum_{n\geq 0} \frac{(-1)^n}{n!} \ad_{\mathcal B}^{(n)}(X),
\end{align}
with $\ad_{\mathcal B}(X)= [\cB,X]$ and where $\ad_{\mathcal B}^{(n)}$ denotes its $n-$fold iteration. For quadratic transformations, say $\mathcal B = \frac{1}{2}\sum_{k \in \Lambda^*} \eta_{k} a^*_ka^*_{-k} - \hc$ with some real numbers $\eta_k = \eta_{-k}$, it is well known that (\ref{duhamel:intro}) converges. For example for $X = a^*_p, p \in \Lambda^*$, it leads to the explicit formula 
\begin{align*}
e^{-\mathcal B} a^*_p e^{\mathcal B}  = \ch (\eta_p)a^*_p + \sh (\eta_p) a_p.
\end{align*}
Cubic transformations, however, do not enjoy the same algebraic structure. In the Gross--Pitaevskii regime ($\kappa = 0$) considered in \cite{BBCS1,BreSchSch-22,NamTri-23}, or in slightly more singular regimes ($\kappa>0$ small), as in \cite{AdhBreSch-21,BreCapSch-21,HabHaiNamSeiTri-23}, the kernel of the cubic transformation $\mathcal B= \sum_{k,r \in \Lambda^*} \eta_{r} a^*_{k+r}a^*_{-r}a_k - \hc$ can be taken to be small so that the expansion (\ref{duhamel:intro}) converges and can be truncated after a few iterations. In our case, since we consider $\kappa>1/2$, the trial state $\Gamma_0$ has too many excitations and the expansion (\ref{duhamel:intro}) is not convergent.

To deal with this problem, we implement the cubic renormalization as product of many ``smaller'' unitary operators, constructed in such a way that the expansion (\ref{duhamel:intro}) converges and yields a closed formula, similar to the one obtained for quadratic transformations. To achieve this goal, we introduce cutoff functions analogous to the ones used in \cite{BCS}, to make sure that for a given momentum $k$ in the shell $S$, we only create one pair $(-r,k+r)$ of particles with high momentum in $H$. For $k \in S$, we say that a pair $(p,q)$ of momenta in $H$ forms a $k$-connection if $p+q = k$. We define $\mathcal B_k = \mathcal B_k^\sharp - \mathcal B_k^{\circ}$, see \eqref{eq:def:Bk}  below, so that $ \mathcal B_k^\sharp$ only creates a $k$-connection if it acts on a state with no $k$-connection, and it cannot create $k'$-connections, for $k' \not = k$. Thanks to these exclusion rules, we will show in Lemma \ref{eq:BcircBdag} below that the unitary transformation $e^{\cB_k}$ is explicitly given by  
\begin{align*}
e^{\mathcal B_k} &= \cos X_k + \mathcal B_k^\sharp \frac{\sin X_k}{X_k} - \frac{\sin X_k}{X_k} \mathcal B_k^\circ +	 \mathcal B_k^\sharp \frac{\cos X_k -1}{X_k^2}\mathcal B_k^\circ,
\end{align*}
where $X_k = | \mathcal B_k^\sharp|$. Using that $X_k$ is small in average on the trial state, we will be able to iteratively combine the action of $e^{\cB_k}$, for all momenta in the shell $S$. 
Additionally, we will show that the cubic transformation preserves some properties of the trial state, which is used to simplify computations.

\medskip

\textit{Organization of the paper.} In Section \ref{sec:trial_states}, we precisely define the previously discussed unitary transformations and our trial state. In Lemma \ref{lem:TLc_main} and in Lemma \ref{lem:Tc_main}, we describe the action of the Bogoliubov transformation $e^{\cB_1}$ and of the cubic transformation $T_c$, respectively. Using these lemmas, we may conclude the proof of \cshref{prop:loc-bound}. 
Section \ref{sec:B1} is devoted to the proof of \Cref{lem:TLc_main} and Section \ref{sec:cubic} to the proof of Lemma \ref{lem:Tc_main}. Finally, in the appendix, we show properties of the kernel of the transformations and we give a proof of \cshref{prop:localization}.

\medskip

\textit{Acknowledgements.} B.S. would like to gratefully acknowledge support from the Swiss National Science Foundation through the Grant ``Dynamical and energetic properties of Bose-Einstein condensates'', from the NCCR SwissMAP and from the European Research Council through the ERC-AdG CLaQS.  This work was partially funded by the Deutsche Forschungsgemeinschaft (DFG, German Research Foundation) – Project-ID 470903074 – TRR 352.
 
\section{The trial state} 
\label{sec:trial_states}

In this section we make explicit our trial state $\Gamma_N$ given in (\ref{eq:gamma_N}). In particular, we construct the unitarity transformations appearing in its definition and compute their actions on the Hamiltonian.


\subsection{The Weyl transformation \texorpdfstring{$\weyl$}{W\textunderscore N\textunderscore 0}}
To generate the Bose-Einstein condensate we use the Weyl operator
\begin{equation*}\label{eq:weyl} \weyl = \exp \big[ \sqrt{N_0} (a_0^* - a_0) \big], \end{equation*}
with a parameter $N_0 \in \bN$, which will be specified later on, see \eqref{eq:N0-def}.
The Weyl operator leaves \(a_p, a^*_p\) invariant for all \(p \in \Lambda^* \setminus \{0\}\). On $a_0, a_0^*$, it acts as a shift, i.e.\ 
\begin{equation}\label{eq:weyl2} 
W^*_{N_0} a_0 \weyl = a_0 + \sqrt{N_0}, \quad W^*_{N_0} a^*_0 \weyl = a^*_0 + \sqrt{N_0}.
\end{equation} 
We obtain the excitation Hamiltonian 
\begin{align} \label{eq:Weyl}
\weyl^* \cH_N \weyl &= \frac{N_0^2}{2} \hat{V}_N(0) + Q_1 + Q_2 + Q_3 + Q_4 + \sum_{p\in \Lambda^*} \left(p^2 +N_0 \hat{V}_N(0) + N_0 \hat{V}_N(p)\right) a_p^*a_p
\end{align}
with
\begin{align*}
Q_1 &= N_0^{3/2} \hat{V}_N(0) a_0 + \hc
&Q_2 &= \frac{N_0}{2}\sum_{p \in \Lambda^*} \hat{V}_N(p)a_p^* a_{-p}^* + \hc
\\
Q_3 &= N_0^{1/2} \sum_{p,r \in \Lambda^*} \hat{V}_N(r) a_{-r}^* a_{r+p}^* a_{p} + \hc
&Q_4 &= \frac{1}{2} \sum_{p,q,r \in \Lambda^*} \hat{V}_N(r) a^*_{p+r}a^*_q a_{q+r}a_p
\end{align*}

\subsection{Quadratic transformation \texorpdfstring{$e^{\mathcal B_1}$}{e\textasciicircum B1} on momenta \texorpdfstring{$|p|>N^{\kappa/2+\vep}$}{|p|>N\textasciicircum\{k/2+eps\} } }

To factor out the short-range correlations, we conjugate $\weyl^* \cH_N \weyl$ with a Bogoliubov transformation that acts on momenta that are higher than the shell momenta. 
For $p\in \Lambda^* \setminus \{0\}$, we define $\vphi_p$ as the unique solution to the box  scattering equation 
\begin{equation} \label{eq:vphi}
p^2 \vphi_p + \frac{1}{2} \sum_{q \neq 0} \widehat{V}_N(p-q)\vphi_q = - \frac{1}{2} \widehat{V}_N(p),
\end{equation}
where we introduced the notation $"q \neq 0"$ for $q \in \Lambda^*\setminus\{0\}$.
The existence and the uniqueness of a solution to this equation was established in \cite{HST}. For the sake of completeness we provide the proof in \Cref{app:scat}.
Next, we define the box scattering length $\mathfrak{a}_N$, which appears more naturally in our setting and, as we will show, is close to the full space scattering length $\mathfrak{a}$.
\begin{equation} \label{eq:scatlength}
8\pi \mathfrak{a}_N := \widehat{V}(0) + N^{1-\kappa} \sum_{p \neq 0} \widehat{V}_N(p) \vphi_p.
\end{equation}
Important properties of the function $\vphi \colon p \mapsto \vphi_p$ with $\vphi_0 = 0$ and of the box scattering length $\mathfrak{a}_N$ are stated in the next lemma. Its proof is given in \Cref{app:scat}, see also \cite{HST}, \cite[Appendix]{B} or \cite[Appendix]{LamTri-24} for alternative proofs.
\begin{lemma} \label{lem:vphi_prop}
For $p\in \Lambda^* \setminus \{0\}$, we have the pointwise bound 
$$
\vphi_p \lesssim p^{-2} N^{-1+\kappa}.
$$
Moreover, we have the estimates 
\begin{align*}
\|\varphi\|_{1} \lesssim 1, \qquad \|\varphi\|_{2} \lesssim N^{-1+\kappa}, \qquad  \|\varphi\|_{\infty} \lesssim N^{-1+\kappa}, \qquad
\|p \varphi\|_{2}^2 \lesssim N^{-1+\kappa}.
\end{align*}
For $\varphi^{(\alpha)}_p = \varphi_p \mathds{1}_{|p| > N^\alpha}, \alpha>0$, we obtain 
\begin{align*}
 \|\varphi^{(\alpha)}\|_{2} &\lesssim N^{-1+\kappa - \alpha/2}, \qquad  \|\varphi^{(\alpha)}\|_{\infty} \lesssim N^{-1+\kappa - 2\alpha} .
\end{align*} 
Moreover,
\begin{align} \label{eq:scat_comparison}
|\mathfrak{a}_N - \mathfrak{a}| \lesssim N^{-1+\kappa}.
\end{align}
 \end{lemma}


Using the solution $\vphi_p$, we define the unitary Bogoliubov transformation $e^{\cB_1}$, with the antisymmetric operator 
$$
\cB_1 = \frac{1}{2} \sum_{|p| > N^{\kappa/2+\vep}} \hspace{-0.3cm} \sinh^{-1}(N_0 \varphi_p)  a_p^*a_{-p}^* - \hc  
$$
Recall that the parameter $\vep > 0$ was introduced in  \eqref{eq:mom-sets}.
For $p \in \Lambda^*$, the action of $e^{\cB_1}$ on an annihilation operator is given by 
\begin{equation} \label{eq:bogo1}
e^{-\cB_1} a_p e^{\cB_1} = c_p a_p + s_p a_{-p}^*,
\end{equation}
where we denote 
\begin{equation} \label{eq:sp_def}
c_p = \cosh\big(\sinh^{-1}(N_0 \varphi_p)\big) = \sqrt{1 + N_0^2 \ph_p^2},  \qquad s_p = \sinh\big(\sinh^{-1}(N_0\varphi_p) \big) = N_0\vphi_p  
\end{equation}
for all $|p| > N^{\kappa/2+\vep}$, while the transformation acts trivially otherwise, i.e.\ $c_p = 1$ and $s_p = 0$ for $|p| \leq N^{\kappa/2+\vep}$.

Conjugating \eqref{eq:Weyl} with the Bogoliubov transformation $e^{\cB_1}$ we obtain a new, renormalized excitation Hamiltonian. Its form is described in the next lemma, whose proof is deferred to Section \ref{sec:B1}. Here we introduce the orthogonal projection 
\[ \Xi := \mathds{1}_{\{\mathcal N_{(S\cup H)^c}=0\}} \mathds{1}_{\{\mathcal N_H \in 2 \mathbb{N}_0\}}. \] Since the state that we are going to use to estimate the free energy of $e^{-\cB_1} \weyl^* \cH_N \weyl e^{\cB_1}$ will be in the range of $\Xi$ (because (\ref{eq:Gamma0_intro}) is clearly in $\text{Ran } \Xi$ and because $\Xi$ commutes with $T_c$ and $e^{\cB_2}$), it is enough for us to estimate the error arising from conjugation with $e^{\cB_1}$ on this subspace.   
\begin{lemma} \label{lem:TLc_main}
Suppose that $0 \leq N-N_0 \leq C N^{3\kappa/2}$ for some constant $C>0$. Then we have 
\begin{equation}\label{eq:claimB1} 
\begin{split}
e^{-\cB_1} \weyl^* &\cH_N \weyl e^{\cB_1} \\ = \; &4\pi\mathfrak{a} N^{1+\kappa} - 8\pi\mathfrak{a} N^\kappa (N-N_0) + \sum_{p \in S} \frac{(4\pi\mathfrak{a} N^\kappa)^2}{p^2}
\\ &+ \sum_{p \in \Lambda^*} p^2a_p^*a_p +  2\hat{V}(0)N^\kappa \cN   + \sum_{p \in S} 4\pi\mathfrak{a}  N^\kappa  \big( a_p^*a_{-p}^* + \hc \big)  + Q_3 + Q_4
+ \mathcal{E}_1,
\end{split} \end{equation}
where the error term $\cE_1$ is bounded, on the range of the projection $\Xi$, by  
\begin{equation}\label{eq:claimB1E} 
\begin{split} 
\mathcal{E}_1 \lesssim \; &N^{-3\eps} \sum_{p \in \Lambda^*} p^2 a_p^* a_p + N^{-5\vep/2} Q_4 + N^{-\kappa/2-\vep/2} \cN^2 + N^{-3+7\kappa +4\vep} \cN_H^2  + N^{2-2\kappa-2\eps} \cN_H 
\\
&+ N^{5\kappa/2-\eps/2} 
\end{split} 
\end{equation} 
if $1/2 < \kappa < 8/15-2\vep/3$ and $\eps > 0$ is chosen as in (\ref{eq:mom-sets}). 
\end{lemma}

\textit{Remark.} We will show that the trial state that we are going to use for the Hamiltonian $e^{-\cB_1} \weyl^* \cH_N \weyl e^{\cB_1}$ is such that, for $j=1,2$, 
\[ \cN^j_H \lesssim N^{(-2+9\kappa/2+\eps) j}; \qquad \cN^j \lesssim N^{j 3\kappa/2}; \qquad Q_4, \sum p^2 a_p^* a_p \lesssim N^{5\kappa/2} \] in expectation. This explains why, for $1/2 < \kappa < 14/27 - \vep$, all terms on the right-hand side of (\ref{eq:claimB1E}) are subleading, i.e.\ they are smaller than $N^{5\kappa/2}$, which is the resolution we are trying to achieve.   

\subsection{The cubic transformation \texorpdfstring{$T_c$}{T\textunderscore c}}

Next, we conjugate the renormalized excitation Hamiltonian $e^{-\cB_1} \weyl^* \cH_N \weyl e^{\cB_1}$ with cubic transformations, annihilating a particle with shell momentum $k \in S$ and creating a pair of particles with high momenta $-r , k+r \in H$, or, vice versa, annihilating a pair of particles with high momenta $-r, k+r \in H$ and creating a particle with shell momentum $k \in S$. 
Instead of a single cubic transformation, we consider a product of unitary operators, one for each low momentum $k \in S$. More precisely, for $k \in S$, we define 
 \begin{align} \label{eq:def:Bk}
 \mathcal B_k^\sharp &=  \sum_{r \in H_k} N^{1/2} \varphi_r a^*_{-r} a^*_{r+k}  a_k \Theta_{k,r}, \qquad \mathcal B_k^\circ = \left( \mathcal B_k^\sharp\right)^*,  \qquad
 \mathcal B_k = \mathcal B_k^\sharp - \mathcal B_k^\circ, 
 \end{align}
with $H_k = \{ r \in H \,|\, r+k \in H\}$ and where the cutoff $\Theta_{k,}$ is defined, similarly as in \cite{BCS}, by
 \begin{equation} \label{eq:Thetakr}
 \begin{split} 
 \Theta_{k,r} &=  \Theta_{k}^{(1)} \times  \Theta_{k,r}^{(2)}, \\
\Theta_{k}^{(1)} &=  \prod_{t \in H} (1-\mathds{1}_{\{\mathcal N_{-t}>0\}} \mathds{1}_{\{\mathcal N_{t+k}>0\}}), \qquad \Theta_{k,r}^{(2)} = \prod_{q \in S} (1-\mathds{1}_{\{\mathcal N_{r+q} + \mathcal N_{-(k+r) + q} >0\}}).
\end{split}
\end{equation}
Here and in the following we use the notation $\cN_q = a_q^* a_q, \, q\in \Lambda^*$, for the operator measuring the number of particles with momentum $q$. Then we define the unitary operator $T_k = e^{\cB_k}$ for every $k \in S$, and the product 
\begin{equation}\label{eq:Tc-def}
T_c = \prod_{k \in S} T_k.
\end{equation} 
In fact, since $[T_k, T_{k'}] \not = 0$, for $k' \not = k$, we need to choose an order in the finite set $S$ to define $T_c$. However, our analysis will not depend on this choice.

We say that a pair of high momenta $(r,t) \in H^2$ is a \emph{$k$-connection} if $r+t =k$. The cut-off $\Theta_{k,r}$ has two roles. The first factor $\Theta_{k,r}^{(1)}$ ensures that $ \mathcal B_k^\sharp$ only creates a  $k$-connection if there is not already one. The second factor $\Theta_{k,r}^{(2)}$ ensures that $\cB_k^\sharp$ does not create any $q$-connection, for $q\neq k$. This latter condition is implemented by asking that the shell-neighborhoods of $r$ and $-(k+r)$ are empty, that is, that there is no occupied mode with momentum $t \in r + S$ or $t \in -(k+r) + S$.

%

In the next lemma we describe the action of $T_c$ on the renormalized excitation Hamiltonian $e^{-\cB_1} \weyl^* \cH_N \weyl e^{\cB_1}$. When applying the lemma, 
we will only be interested in controlling the action of $T_c$ on a specific trial state $\Gamma$. For this reason, we are going to restrict our attention to states with certain additional properties, which will be satisfied by the $\Gamma$ that we are going to choose. First of all, we can restrict our attention to the range of the projection $\mathds{1}_{\{\cN_{S^c} =0\}}$ (because the trial state will not have particles with momentum outside $S$). Additionally, we can focus on states that commute with the parity operator for the sum of the number of particles with momentum $\pm k$ and the number of $\pm k$-connections. For $k \in S$, consider 
\begin{align} \label{eq:def_Mx}
\mathcal M_k = \mathcal N_{k} + \frac{1}{2}\sum_{t \in H_{k}} \mathcal N_{-t}\, \mathcal N_{t+k}
\end{align}
and the parity operators
\begin{equation}\label{eq:parity}
\begin{split} 
{\mathbb{P}}_k &= \1_{\{\mathcal M_k + \mathcal M_{-k} \in 2 \mathbb{N}_0\}}, \\
	{\mathbb{Q}}_k = 1- {\mathbb{P}}_k &= \1_{\{\mathcal M_k + \mathcal M_{-k} \in 2 \mathbb{N}_0+1\}}.
\end{split}\end{equation}
In accordance with the intuition that $\cB_k^\sharp$ annihilates a particle with momentum $k$ while creating a $k$-connection, one readily checks that
$[\cB_q^\sharp, \cM_k] =0$ for all $q \in S$. Therefore we find that $[ T_c, \cM_k] = [T_c, \mathbb{P}_k] = [T_c, \mathbb{Q}_k] = 0$ for all $k \in S$. For this reason, restricting to $\Gamma$ such that $[ \Gamma, \mathbb{P}_k] = [\Gamma, \mathbb{Q}_k] = 0$ for all $k \in S$, we will be able to neglect all terms in the renormalized excitation Hamiltonian $e^{-\cB_1} \weyl^* \cH_N \weyl e^{\cB_1}$ that do not preserve parity.

\begin{lemma} \label{lem:Tc_main}
Let $\Gamma$ be a normalized density matrix on $\cF (\Lambda)$ with $\Gamma = \mathds{1}_{\{\cN_{S^c} =0\}} \Gamma \mathds{1}_{\{\cN_{S^c} =0\}}$ and such that $[ \Gamma, \mathbb{P}_k] = 0$ for all $k \in S$. Assume $0 \leq N- N_0 \lesssim N^{3\kappa/2}$. Then we have
\begin{equation} \label{eq:cubic} 
\begin{split} 
\tr\, T_c^* &e^{-\cB_1} \weyl^* \cH_N \weyl e^{\cB_1} T_c  \Gamma 
\\ \leq \; &4\pi\mathfrak{a} N^{1+\kappa} - 8\pi\mathfrak{a} N^\kappa (N-N_0) + \sum_{p \in S} \frac{(4\pi\mathfrak{a} N^\kappa)^2}{p^2} 
\\
& + \tr \Big[ \sum_{p \in S} p^2a_p^*a_p + 16\pi\mathfrak{a} N^\kappa \cN_S + \sum_{p \in S}  4\pi\mathfrak{a} N^\kappa  \big(a_p^*a_{-p}^* + \hc \big) \Big] \Gamma + \delta (\Gamma) 
\end{split} 
\end{equation} 
where
\begin{equation}\label{eq:deltaG} \begin{split} \delta (\Gamma) &\lesssim \; N^{-1+ \kappa + 6\eps} \, \tr \, \big( \cN_S + N^{3\kappa/2+3\eps} \big) \, \sum_{k \in S} \cN_k^2 \, \Gamma
\\
&\qquad  + \big( N^{-\kappa/2 -\eps/2} + N^{-7+13\kappa+6\vep}\big) \, \tr \, \cN_S^2 \Gamma + N^{5\kappa/2-\eps/2} 
\end{split} \end{equation} 
if $1/2 < \kappa < 8/15-2\vep/3$ and $\eps > 0$ is chosen as in (\ref{eq:mom-sets}).
\end{lemma}
This lemma is the main novelty of the paper. We defer its proof to Section \ref{sec:cubic}.

\subsection{Quadratic transformation \texorpdfstring{$e^{\mathcal B_2}$}{e\textasciicircum B2} on shell momenta}

Consider now the quadratic Hamiltonian in (\ref{eq:cubic}). Since we are going to choose $N_0$  such that $N- N_0 = \tr \cN_S \Gamma + o(N^{3\kappa/2})$, see (\ref{eq:N0-def}) below, we are effectively left with a factor $8\pi \mathfrak{a} N^\kappa \cN_S$ in the last line of (\ref{eq:cubic}). With the aim of diagonalizing this quadratic operator, we implicitly define the coefficients $\tau_p \in \R$, for $p \in S$, via 
\[ \tanh(2\tau_p) = - \frac{8\pi \mathfrak{a}N^\kappa}{p^2+8\pi\mathfrak{a}N^\kappa}. \]
The action of the unitary Bogoliubov transformation $e^{\cB_2}$, with the antisymmetric operator
\[ \cB_2 = \frac{1}{2} \sum_{p\in S} \tau_p a_p^*a_{-p}^* - \hc, \]
is given explicitly, similarly to (\ref{eq:bogo1}), by 
\begin{equation}\label{eq:bogo2}  e^{-\cB_2} a_p e^{\cB_2} = \gamma_p a_p + \sigma_p a_{-p}^*  , \qquad  e^{-\cB_2} a^*_p e^{\cB_2} = \gamma_p a^*_p + \sigma_p a_{-p}  \end{equation} 
with $\gamma_p = \cosh \tau_p$, $\sigma_p = \sinh \tau_p$, for all $p \in S$. 
A straightforward calculation shows that
\begin{equation}\label{eq:diago} \begin{split} e^{-\cB_2} &\Big[ \sum_{p \in S}  \, p^2 a_p^*a_p + 8\pi\mathfrak{a} N^\kappa \cN_S + \sum_{p \in S}  4\pi\mathfrak{a} N^\kappa  \big(a_p^*a_{-p}^* + \hc \big) \Big]  e^{\cB_2} \\ &= \frac{1}{2} \sum_{p \in S} \big[\sqrt{|p|^4 + 16 \pi \mathfrak{a} N^\kappa p^2} - p^2 - 8\pi \mathfrak{a} N^\kappa\big] + \sum_{p \in S} \sqrt{|p|^4 + 16 \pi \mathfrak{a} N^\kappa p^2} \, a_p^* a_p. \end{split} \end{equation} 
We may now precisely define the trial state of the transformed Hamiltonian as the Gibbs state of the diagonal quadratic Hamiltonian in (\ref{eq:diago}). We set 
\begin{equation} \label{eq:Gamma0} \Gamma_0 = Z^{-1} \mathds{1}_{\{\cN_{S^c} = 0\}} \, \exp\bigg( -\frac{1}{T N^{2-2\kappa}} \sum_{p \in S} \sqrt{|p|^4 + 16 \pi \mathfrak{a} N^\kappa p^2} \, a_p^* a_p\bigg),  \end{equation} 
with the normalization constant 
\[ Z = \tr \, \mathds{1}_{\{\cN_{S^c} = 0\}} \, \exp \bigg( -\frac{1}{T N^{2-2\kappa}} \sum_{p \in S} \sqrt{|p|^4 + 16 \pi \mathfrak{a} N^\kappa p^2} \, a_p^* a_p\bigg).  \]
Furthermore, we define 
\begin{equation}\label{eq:G-def}  \Gamma = e^{\cB_2} \Gamma_0 e^{-\cB_2} \end{equation} 
as the trial state for the Hamiltonian appearing in Lemma \ref{lem:Tc_main}.
%
Some important properties of $\Gamma$ are listed in the following lemma.
\begin{lemma} \label{lm:G-prop} 
Let $\Gamma$ be defined as in (\ref{eq:G-def}). Then $\Gamma \geq 0$ with $\tr \Gamma = 1$. Moreover, we have $\Gamma = \mathds{1}_{\{\cN_{S^c} = 0\}} \Gamma \mathds{1}_{\{\cN_{S^c} = 0\}}$ and $[\Gamma, \mathbb{P}_k] = [\Gamma, \mathbb{Q}_k] = 0$ for all $k\in S$ and with the parity operators $\mathbb{P}_k, \mathbb{Q}_k$ defined in (\ref{eq:parity}). Furthermore, for $j =1,2$, there exist positive constants $c < C$ such that 
\begin{equation}\label{eq:NG-claim} \begin{split} c N^{3j \kappa/2} \leq \tr \, \cN^j_S \Gamma &\leq C N^{3j \kappa /2} \\  
\tr \cN_S^{j-1} \sum_{k \in S} \cN_k^2 \Gamma & \leq C N^{j (3\kappa/2+2\eps)}. \end{split} \end{equation} 
\end{lemma} 

\begin{proof} 
The fact that $\Gamma \geq 0$, $\tr \Gamma = 1$, $\Gamma = \mathds{1}_{\{\cN_{S^c} = 0\}} \Gamma \mathds{1}_{\{\cN_{S^c} = 0\}}$ and $[\Gamma, \mathbb{P}_k] = [\Gamma, \mathbb{Q}_k] = 0$ for all $k \in S$ follows readily from the definition of $\Gamma$. 
It remains to show the inequalities (\ref{eq:NG-claim}). We begin with $j=1$. From the definition of $\Gamma_0$
a standard ideal gas computation gives
\begin{equation}\label{eq:ideal} \begin{split} 
\tr \cN_p \Gamma_0 &= \frac{1}{\exp \big( \frac{1}{TN^{2-2\kappa}} \sqrt{|p|^4 + 16 \pi \mathfrak{a} N^\kappa p^2} \big) -1}  \\ 
\tr \cN_p^2 \Gamma_0 &= \frac{2}{\Big[ \exp \big( \frac{1}{TN^{2-2\kappa}}\sqrt{|p|^4 + 16 \pi \mathfrak{a} N^\kappa p^2}\big)  -1 \Big]^2} + \frac{1}{\exp \big( \frac{1}{TN^{2-2\kappa}}\sqrt{|p|^4 + 16 \pi \mathfrak{a} N^\kappa p^2}\big) -1}  \end{split} \end{equation}  
and, by symmetry, 
\[ \tr a_p^* a_{-p}^* \Gamma_0 = \tr a_p a_{-p} \Gamma_0 = 0 \]
for all $p \in S$. With (\ref{eq:bogo2}), we find 
\begin{equation}\label{eq:NG2} \begin{split} 
\tr \cN_S \Gamma = \; &\sum_{p\in S} \tr \, (\gamma_p a_p^* + \sigma_p a_{-p}) ( \gamma_p a_p + \sigma_p a_{-p}^*) \Gamma_0 \\
= \; &\sum_{p \in S} \frac{p^2+8\pi\mathfrak{a} N^\kappa}{\sqrt{|p|^4 + 16 \pi \mathfrak{a} N^\kappa p^2}} \frac{1}{\exp \big(\frac{1}{TN^{2-2\kappa}} \sqrt{|p|^4 + 16 \pi \mathfrak{a} N^\kappa p^2} \big) -1} \\ &\hspace{5cm} + \frac{1}{2} \sum_{p \in S} \Big[ \frac{p^2+8\pi\mathfrak{a} N^\kappa}{\sqrt{|p|^4 + 16 \pi \mathfrak{a} N^\kappa p^2}} - 1 \Big].
\end{split} \end{equation}  
Scaling $p \mapsto p/N^{\kappa/2}$, interpreting the sums as Riemann sums, approximating them with integrals over $\bR^3$, and using the restriction $T N^{2-3\kappa} \leq C$, we conclude that $cN^{3\kappa/2} \leq \tr \cN_S \Gamma \leq C N^{3\kappa/2}$ (both terms are positive, the second term is of order $N^{3\kappa/2}$, independently of $T$ and the first term is at most of order $N^{3\kappa/2}$, for $T N^{2-3\kappa} \leq C$). 

To prove the second bound in (\ref{eq:NG-claim}), we proceed similarly. Also here, we use the fact that only observables preserving the number of particles have non-trivial expectation w.r.t. $\Gamma_0$.
\begin{equation}\label{eq:N2S} \begin{split} 
\tr \, \sum_{p \in S} \cN_p^2 \Gamma =\; & \sum_{p \in S} \tr \, (\gamma_p a_p^* + \sigma_p a_{-p}) ( \gamma_p a_p + \sigma_p a_{-p}^*) (\gamma_p a_p^* + \sigma_p a_{-p}) ( \gamma_p a_p + \sigma_p a_{-p}^*) \Gamma_0 \\ 
\lesssim \; & \sum_{p \in S} (\gamma_p^2 + \sigma_p^2)^2  \, \tr (\cN_p^2 + \cN_{-p}^2) \Gamma_0 + \sum_{p \in S} \sigma_p^2 (\gamma_p^2 + \sigma_p^2) \\ 
\lesssim \; & \sum_{p \in S} \frac{\big[ p^2+8\pi\mathfrak{a} N^\kappa \big]^2}{|p|^4 + 16 \pi \mathfrak{a} N^\kappa p^2} \, \tr \, \cN_p^2 \Gamma_0 \\ &+ \frac{1}{2} \sum_{p \in S} \Big[ \frac{p^2+8\pi\mathfrak{a} N^\kappa}{\sqrt{|p|^4 + 16 \pi \mathfrak{a} N^\kappa p^2}} - 1 \Big] \frac{p^2+8\pi\mathfrak{a} N^\kappa}{\sqrt{|p|^4 + 16 \pi \mathfrak{a} N^\kappa p^2}}
 \end{split} \end{equation} 
Inserting (\ref{eq:ideal}), we can again scale $p \mapsto p/ N^{\kappa/2}$ and we can approximate with an integral. Compared with (\ref{eq:NG2}), however, the singularity at $p/N^{\kappa/2} \simeq 0$ is more severe since there $\tr \cN_p^2 \Gamma_0 \simeq (p/N^{\kappa/2})^{-2}$. For this reason, in the region $N^{-\eps} \leq |p| / N^{\kappa/2} \leq c$, we estimate $\tr \cN_p^2 \Gamma_0 \lesssim N^{2\eps}$. We conclude that $\sum_{p \in S} \tr \cN_p^2 \Gamma \lesssim N^{3\kappa/2 + 2\eps}$. 

Let us now consider the case $j=2$. We have
\begin{equation}\label{eq:NS2} \begin{split} \tr \cN_S^2 \Gamma = \sum_{p,q \in S} \tr \; (\gamma_p a_p^* + \sigma_p a_{-p})(\gamma_p a_p + \sigma_p a_{-p}^*) (\gamma_q a_q^* + \sigma_q a_{-q} ) (\gamma_q a_q + \sigma_q a_{-q}^*) \Gamma_0 \end{split} \end{equation} 
Since the Gibbs state $\Gamma_0$ factorizes, we find that $\tr \cN_S^2 \Gamma$ is the same as $(\tr \cN_S \Gamma )^2$, up to terms that are associated with $p = q$ on the right-hand side of (\ref{eq:NS2}) and that can be handled as we did in (\ref{eq:N2S}), producing errors of order $N^{3\kappa/2+2\eps} \ll N^{3\kappa}$. We conclude that $cN^{3\kappa} \leq \tr \cN_S^2 \Gamma \leq C N^{3\kappa}$. The second bound in (\ref{eq:NG-claim}) for $j=2$ can be proven similarly; we leave the details to the reader.
\end{proof}
%
%
%
\subsection{Definition of \texorpdfstring{$N_0$}{N\textunderscore 0} and proof of Proposition \ref{prop:loc-bound}}

From (\ref{eq:G-def}), we are led to the definition   
\begin{equation} \label{eq:GammaN} 
\Gamma_N = \weyl e^{\cB_1} T_c \Gamma T_c^* e^{-
\cB_1} \weyl^*= \weyl e^{\cB_1} T_c e^{\cB_2} \, \Gamma_0 \,  e^{-\cB_2} T_c^* e^{-
\cB_1} \weyl^*
\end{equation}
as trial state for the Hamiltonian $\cH_N$ on the small, rescaled box $\Lambda$. In the next proposition, we estimate the expected number of particles in the state $\Gamma_N$. Its proof is deferred to Section \ref{sec:propN}, because it makes use of some properties of $T_c$ that will be discussed in Section \ref{sec:cubic}.
\begin{proposition}\label{prop:NN2} 
Let $\Gamma_N$ be given by (\ref{eq:GammaN}). Then there exists a constant $C>0$ such that for $\vep$  chosen as in (\ref{eq:mom-sets}) we have
\begin{align}
N_0 + \tr \cN_S \Gamma \leq \tr \cN  \Gamma_N \leq N_0 + \tr \cN_S \Gamma + CN^{3\kappa/2-\vep}. \label{eq:NN2}
\end{align}
\end{proposition}
Prop. \ref{prop:NN2} motivates the following choice of $N_0 \in \mathbb{R}_+$
\begin{equation}\label{eq:N0-def} 
N_0:= N  - \tr \cN_S \Gamma + N^{3\kappa/2-(\kappa-1/2)}.
\end{equation}
From Lemma \ref{lm:G-prop}, we conclude that $c N^{3\kappa/2} \leq N - N_0 \leq C N^{3\kappa/2}$ for all $\kappa > 1/2$ and $N$ large enough. The term $N^{3\kappa/2-(\kappa-1/2)}$, which is small compared with $\tr \cN_S \Gamma$ for $\kappa > 1/2$, will be used to make sure that the lower bound in \eqref{eq:trNGN} holds true. We are now set to show \cshref{prop:loc-bound}. 

\begin{proof}[Proof of \cshref{prop:loc-bound}]
Combining (\ref{eq:N0-def}) with \cshref{prop:NN2}, we obtain constants $c_1, c_2 > 0$ such that 
\begin{equation} \label{eq:NGN} N + c_1 N^{3\kappa/2 - (\kappa -1/2)} \leq \tr \cN \Gamma_N \leq N + c_2 N^{3\kappa/2-\vep}.
\end{equation} 
Furthermore, we observe that with the definition (\ref{eq:G-def}) of the density matrix $\Gamma$ and the choice  \eqref{eq:N0-def} of $N_0$, all assumptions in Lemma \ref{lem:Tc_main} are satisfied. Thus, inserting \eqref{eq:diago} into \Cref{lem:Tc_main} we find
\begin{equation} \begin{split}\label{eq:ener-G}
 \tr  \cH_N &\Gamma_N - T N^{2-2\kappa} S (\Gamma_N) \\ 
\leq \; &4 \pi \mathfrak{a} N^{1+\kappa} + \frac{1}{2} \sum_{p\in S} \big[\sqrt{|p|^4 + 16 \pi \mathfrak{a} N^\kappa p^2} - p^2 - 8\pi \mathfrak{a} N^\kappa + \frac{(8\pi \mathfrak{a} N^\kappa)^2}{2p^2} \big] \\  &+ \tr \sum_{p \in S} \sqrt{|p|^4 + 16\pi \mathfrak{a} N^\kappa p^2} \, a_p^*a_p \Gamma_0 -  T N^{2-2\kappa} S (\Gamma_0)  + 8\pi\mathfrak{a} N^{5\kappa/2-(\kappa-1/2)} + \delta(\Gamma).
\end{split} \end{equation}
Combining (\ref{eq:deltaG}) and Lemma \ref{lm:G-prop}, we conclude that
\[ \delta (\Gamma) \lesssim N^{5\kappa/2 - \eps/2} \]
if $1/2 < \kappa < 14/27$ and $\vep > 0$ is small enough. As in the proof of \Cref{lm:G-prop}, we scale $p \mapsto p/N^{-\kappa/2}$ $p \mapsto q = p/N^{\kappa/2}$ in the first sum on the r.h.s. of  \eqref{eq:ener-G}. We find
\begin{equation}\label{eq:sum1} \begin{split} 
\frac{1}{2} \sum_{p \in S} & \big[\sqrt{|p|^4 + 16 \pi \mathfrak{a} N^\kappa p^2} - p^2 - 8\pi \mathfrak{a} N^\kappa + \frac{(8\pi \mathfrak{a} N^\kappa)^2}{2p^2} \big]  \\ 
= &\frac{N^{5\kappa/2}}{(2\pi)^3} \frac{(2\pi)^3}{N^{3\kappa/2}} \sum_{\substack{q \in N^{-\kappa/2} 2\pi  \bZ^3: \\ N^{-\eps} \leq |q| \leq N^\eps}}  \big[ \sqrt{|q|^4 + 16 \pi \mathfrak{a} q^2}  - q^2 - 8\pi \mathfrak{a} + \frac{(8\pi \mathfrak{a})^2}{2q^2} \big] \\
\leq \; &\frac{N^{5\kappa/2}}{(2\pi)^3} \int_{\bR^3} \big[  \sqrt{|q|^4 + 16 \pi \mathfrak{a} q^2}  - q^2 - 8\pi \mathfrak{a} + \frac{(8\pi \mathfrak{a})^2}{2q^2} \big]  dq  + C N^{5\kappa/2-\eps} \\
= \; &4 \pi \mathfrak{a} N^{5\kappa/2}  \frac{128 \mathfrak{a}^{3/2}}{15 \sqrt{\pi}}  + C N^{5\kappa/2 - \eps}. \end{split} \end{equation} 
Here, we approximated the Riemann sum with the corresponding integral over $\bR^3$, which can be computed explicitly, see, for example, \cite[eq.\,(8.11)]{HabHaiNamSeiTri-23} for more details.

Next, we consider the term in the second line of \eqref{eq:ener-G}. From the choice (\ref{eq:Gamma0}) of $\Gamma_0$ we conclude, by the Gibbs principle, that
\begin{align*}
\tr \sum_{p \in S} &\sqrt{|p|^4 + 16\pi \mathfrak{a} N^\kappa p^2} \, a_p^*a_p \Gamma_0 -  T N^{2-2\kappa} S (\Gamma_0) \\ &= - T N^{2-2\kappa} \log Z = T N^{2-2\kappa} \sum_{p \in S} \log \left[ 1 - \exp \Big( -\frac{1}{TN^{2-2\kappa}} \sqrt{|p|^4 + 16 \pi \mathfrak{a} N^\kappa p^2} \Big) \right],
\end{align*}
where the last equality is a standard ideal gas computation as in \eqref{eq:ideal}. As before we rescale $p\mapsto p /N^{\kappa/2}$ and replace the resulting Riemann sum by an integral, while
keeping in mind the condition $T N^{2-3\kappa} \leq C$. We obtain
\begin{align}\label{eq:sum2} 
& T N^{2-2\kappa} \sum_{p\in S} \log \Big[ 1- \exp \Big( -\frac{1}{TN^{2-2\kappa}} \sqrt{|p|^4 + 16 \pi \mathfrak{a} N^\kappa p^2} \Big) \Big] \nn
\\
&\qquad \leq T^{5/2}N^{5-5\kappa} (2\pi)^{-3} \int_{\R^3} \log \Big[ 1- \exp \Big( - \sqrt{|q|^4 + \frac{16 \pi  \mathfrak{a} q^2} { T N^{2-3\kappa}}} \; \Big) \Big] dq + C N^{5\kappa/2-\eps} ,
\end{align} 
see \cite[eq.\,(9.13)]{HabHaiNamSeiTri-23} for more details.
Inserting \eqref{eq:sum2} and \eqref{eq:sum1} into \eqref{eq:ener-G}, we arrive at 
\begin{equation}\label{eq:en-GN} 
\begin{split} 
&\tr \cH_N \Gamma_N - T N^{2-2\kappa} S(\Gamma_N) \\ & \quad \leq \; 4\pi \mathfrak{a} N^{1+\kappa} \left( 1 + \frac{128}{15 \sqrt{\pi}} (N^{-2+3\kappa} \mathfrak{a}^3 )^{1/2} \right) \\ & \qquad + T^{5/2} N^{5-5\kappa} (2\pi)^{-3} \int_{\R^3} \log \Big[ 1- \exp \Big( - \sqrt{|q|^4 + \frac{16 \pi  \mathfrak{a} q^2} { T N^{2-3\kappa}}} \; \Big) \Big] dq + C N^{5\kappa/2-\eps/2} \, .
\end{split} \end{equation} 

Next, we scale back to the box $\Lambda_L = [-L/2, L/2]^3$, with $L = \rho^{-\gamma}$, for some $\gamma > 1$. We recall the choice $N = \rho L^3 = \rho^{1-3\gamma}$ which translates into $L = N^{1-\kappa}$, with $\kappa = (2\gamma-1)/(3\gamma -1)$. We observe that the condition $\gamma > 1$ is equivalent to $\kappa > 1/2$. We define the unitary operator $\mathcal{U}_L : \cF (\Lambda) \to \cF (\Lambda_L)$,  $(\mathcal{U}_L \psi)^{(n)} (x_1, \dots , x_n) = L^{-3n/2} \psi^{(n)} (x_1/ L , \dots , x_n / L)$, for all $n \in \bN$ and all $\Psi = \{ \psi^{(n)} \}_{n\in \bN} \in \cF (\Lambda)$ and we remark that the Hamiltonian (\ref{eq:cHN}) satisfies $L^{-2} \cU_L \cH_N \cU^*_L = \cH$. Therefore, a good trial state for $\cH$ is given by the density matrix $\Gamma_L = \cU_L \Gamma_N \cU_L^*$ on $\cF (\Lambda_L)$. Since $\tr \cN \Gamma_L = \tr \cN \Gamma_N$, it follows from (\ref{eq:NGN}) that 
\[  c_1 \rho^{(\gamma+2)/2} \leq \frac{1}{L^3} \tr \cN \Gamma_L - \rho \leq c_2 \rho^{3/2} \]
in accordance with \eqref{eq:prop3-NN2}. Here we used the facts that $N^{\kappa/2} / L = \rho^{1/2}$, that $N^{-2+3\kappa} = \rho$ and that $\kappa = (2\gamma-1)/(3\gamma-1)$.
Moreover, from (\ref{eq:en-GN}) we obtain
\[ \begin{split} 
& \frac{1}{L^3} \Big[ \tr \cH \Gamma_L - T S (\Gamma_L) \Big] = \; L^{-5} \big[ \tr \cH_N \Gamma_N - T N^{2-2\kappa} S (\Gamma_N)  \Big]  
\\ &\leq \;  4\pi \mathfrak{a} \rho^2 \Big( 1 + \frac{128}{15 \sqrt{\pi}} (\rho \mathfrak{a}^3)^{1/2} \Big) 
+ \frac{T^{5/2}}{(2\pi)^3}  \int_{\bR^3} \log \Big[ 1-  \exp \big( - \sqrt{|q|^4 + 16\pi \mathfrak{a} \rho q^2 / T} \big) \Big]  dq + \rho^{5/2 + \epsilon} 
\end{split} 
\] 
for a sufficiently small $\epsilon > 0$ and for all $T \leq C \rho \mathfrak{a}$. This concludes the proof of \cshref{prop:loc-bound}.
\end{proof}

\section{Quadratic Transformation on Large Momenta} 
\label{sec:B1} 

The goal of this section is to prove \Cref{lem:TLc_main}. 
Throughout, we will assume $0 \leq N - N_0 \lesssim N^{3\kappa/2}$ as in \Cref{lem:TLc_main}. 
With this aim, we consider separately the action of the Bogoliubov transformation $e^{\cB_1}$, determined by \eqref{eq:bogo1}, on each term in the excitation Hamiltonian \eqref{eq:Weyl}. We start with the diagonal term.
\begin{lemma} \label{lem:TLc_kinetic}
We have 
\begin{align} \label{eq:kin-1} 
e^{-\cB_1} \sum_{p \in \Lambda^*} & \left(p^2 +N_0\hat{V}_N(0) + N_0 \hat{V}_N(p)\right) a_p^*a_p e^{\cB_1} \nn
\\
&= 2\hat{V}(0) N^\kappa \cN + \sum_{p \in \Lambda^*} p^2a_p^*a_p + \sum_{p\in\Lambda^*} p^2 s_p  \big( a_p^*a_{-p}^* + \hc \big) + \sum_{p \in \Lambda^*} p^2s_p^2 + \mathcal{E}_{1}^{(\rm diag)},
\end{align} 
with
\begin{align*}
\pm \mathcal{E}_{1}^{(\rm diag)} \lesssim N^{-2+3\kappa} \sum_{p \in \Lambda^*}p^2a_p^*a_p + N^{\kappa - \vep/2}\cN + N^{5\kappa/2 - \vep/2},
\end{align*}
for all $0 < \kappa < 2/3$ and with $\vep>0$ as introduced in \eqref{eq:mom-sets}.
\end{lemma}
\begin{proof}
Denote $\beta_p = \left(p^2 + N_0\hat{V}_N(0) + N_0\hat{V}_N(p)\right)$. With (\ref{eq:bogo1}) and $c_p^2 - s_p^2=1$, we compute
\begin{align*}
e^{-\cB_1} \sum_{p \in \Lambda^*} \beta_p a_p^*a_p e^{\cB_1} - \sum_{p \in \Lambda^*} \beta_p a_p^*a_p
&= 2 \sum_{p\in \Lambda^*} \beta_p s_p^2 a_p^*a_p  +  \sum_{p\in\Lambda^*} \beta_p c_p s_p \big( a_p^*a_{-p}^* + \hc \big) + \sum_{p \in \Lambda^*} \beta_p s_p^2
\end{align*}
and we obtain
\begin{align} \label{eq:B1_kin_main}
&e^{-\cB_1} \sum_{p \in \Lambda^*} \beta_p a_p^*a_p e^{\cB_1} - 2\hat{V}(0)N^\kappa \cN - \sum_{p\in \Lambda^*} p^2a_p^*a_p - \sum_{p\in \Lambda^*}p^2 s_p (a_p^*a_{-p}^* + \hc) - \sum_{p \in \Lambda^*} p^2s_p^2 \nn
\\
&= \sum_{p \in \Lambda^*} \left(N_0 \hat{V}_N(0) + N_0 \hat{V}_N(p) - 2\hat{V}(0)N^\kappa + 2\beta_ps_p^2 \right) a_p^*a_p 
\\
&+ \sum_{p\in\Lambda^*} \big((N_0 \hat{V}_N(0)c_p + N_0\hat{V}_N(p)c_p + p^2 (c_p-1)\big) s_p \big( a_p^*a_{-p}^* + \hc \big) + \sum_{p \in \Lambda^*} N_0 (\hat{V}_N(0) + \hat{V}_N(p))s_p^2. \nn
\end{align}
Let us denote the term on the right-hand side of the previous equation by $\mathcal{E}_1^{(\text{diag})}$. We obtain \eqref{eq:kin-1} and it remains to estimate the error term. In order to control the diagonal part of $\mathcal{E}_1^{(\text{diag})}$, we use $|\hat{V}_N(p)-\hat{V}_N(0)| \lesssim p^2 N^{-3+3\kappa}$, which easily follows from the fact that $V$ is even, the assumption of \Cref{lem:TLc_main} that $0\leq N-N_0\lesssim N^{3\kappa/2}$, as well as the estimate $\beta_p s_p^2 \lesssim N^{\kappa-2\vep}$ from (\ref{eq:sp_def}) and \Cref{lem:vphi_prop}. We find 
\begin{align} \label{eq:B1_kin_diag}
\pm \sum_{p \in \Lambda^*} & \left(N_0 \hat{V}_N(0) + N_0 \hat{V}_N(p) - 2\hat{V}(0)N^\kappa + 2\beta_ps_p^2 \right) a_p^*a_p  \nn
\\
&\qquad \lesssim N^{-2+3\kappa} \sum_{p \in \Lambda^*} p^2 a_p^*a_p + N^{\kappa - 2\vep}\cN + N^{-1 + 5\kappa/2} \cN.
\end{align}
As for the constant term in $\mathcal{E}_1^{(\text{diag})}$, we have
\begin{align} \label{eq:B1_kin_const}
\pm \sum_{p \in \Lambda^*} N_0 (\hat{V}_N(0) + \hat{V}_N(p))s_p^2 \lesssim N^\kappa \|s\|_2^2 \lesssim N^{5\kappa/2-\vep}.
\end{align}
Finally, we consider the off-diagonal part of $\mathcal{E}_1^{(\text{diag})}$. From $0\leq c_p -1 \leq s_p^2$ we find the inequality $N_0 \hat{V}_N(0)c_p + N_0\hat{V}_N(p)c_p + p^2 (c_p-1) \lesssim N^\kappa$. Thus, Cauchy--Schwarz implies
\begin{align*}
&\pm \sum_{p\in\Lambda^*} \big((N_0 \hat{V}_N(0)c_p + N_0\hat{V}_N(p)c_p + p^2 (c_p-1)\big) s_p \big( a_p^*a_{-p}^* + \hc \big) 
\\
&\qquad \qquad  \lesssim N^\kappa \left(N^{-\vep/2} \cN + N^{\vep/2} \|s\|_\infty^2 \cN + N^{\vep/2} \|s\|_2^2\right) 
\lesssim N^{\kappa - \vep/2} \cN + N^{5\kappa/2-\vep/2}.
\end{align*}
Inserting \eqref{eq:B1_kin_diag}, \eqref{eq:B1_kin_const} and the last estimate into \eqref{eq:B1_kin_main} yields \Cref{lem:TLc_kinetic} (using also the assumption $-2+3\kappa + 4\eps < 0$ from \eqref{eq:mom-sets}).
\end{proof}

Next, we consider the action of $e^{\cB_1}$ on the quadratic term $Q_2$, defined in \eqref{eq:Weyl}. 
\begin{lemma} \label{lem:TLc_Q2}
We have 
\begin{equation} \label{eq:Q2} 
e^{-\cB_1} Q_2 e^{\cB_1} = Q_2 + N_0 \sum_{p\in \Lambda^*} \hat{V}_N(p) s_p + \mathcal{E}_{1}^{(Q_2)}
\end{equation} 
with
\begin{align*}
\pm \mathcal{E}_{1}^{(Q_2)} \lesssim N^{\kappa-2\eps}\cN + N^{5\kappa/2-3\eps},
\end{align*}
for all $0 < \kappa < 2/3$ and with $\vep>0$ as introduced in \eqref{eq:mom-sets}. 
\end{lemma}
\begin{proof}
With (\ref{eq:bogo1}), we compute
\[ \begin{split} 
e^{-\cB_1} &Q_2 e^{\cB_1} - Q_2 
\\ &=  N_0 \sum_{p\in \Lambda^*} \hat{V}_N(p) s_p^2  \, (a_p^*a_{-p}^*+ \hc ) +  N_0 \sum_{p\in \Lambda^*} \hat{V}_N(p) c_p s_p \, a_p^*a_{p} + N_0 \sum_{p\in \Lambda^*} \hat{V}_N(p) s_p c_p.
\end{split} \]
We obtain  
$$
\mathcal{E}_1^{(Q_2)} = N_0 \sum_{p\in \Lambda^*} \hat{V}_N(p)s_p^2 \, (a_p^*a_{-p}^*  + \hc) +  N_0 \sum_{p\in \Lambda^*} \hat{V}_N(p) c_ps_p a_p^*a_{p} + N_0 \sum_{p\in \Lambda^*} \hat{V}_N(p) s_p (c_p - 1).
$$
The bound for $\mathcal{E}_1^{(Q_2)}$ follows from Cauchy--Schwarz, $0 \leq c_p -1 \leq s_p^2$ and \Cref{lem:vphi_prop}.
\end{proof}

Next, we consider the cubic term $Q_3$. Notice that here, and also in the next lemma which devoted to $Q_4$, we only bound the restriction of the error term on the range of the projection $\Xi$; this is enough, since our trial state will be in the range of $\Xi$.
\begin{lemma} \label{lem:TLc_Q3}
We have  
\begin{equation}\label{eq:BQ3B} 
e^{-\cB_1} Q_3 e^{\cB_1} = Q_3 + \mathcal{E}^{(Q_3)}_{1},
\end{equation} 
where, on the range of $\Xi$, 
\[
\pm  \mathcal{E}^{(Q_3)}_{1}  \lesssim N^{-2 + 4\kappa + 5 \eps} \big(\cN_H+1\big)^2 + N^{\kappa-4\vep} \big( \cN + N^{3\kappa/2+3\vep} \big), \]
for all $0 < \kappa < 2/3$ and with $\vep>0$ as introduced in \eqref{eq:mom-sets}. 
\end{lemma}
\begin{proof}
To compute $e^{-\cB_1} Q_3 e^{\cB_1}$, we apply (\ref{eq:bogo1}). We find (\ref{eq:BQ3B}), with an error $\cE_1^{(Q_3)}$ given, as a quadratic  form on the range of $\Xi$, by 
\begin{align} \label{eq:cEQ3}
\cE_1^{(Q_3)} = \; &N_0^{1/2} \sum_{p,r \in \Lambda^*} \widehat{V}_N(r) (c_{r} c_{r+p} c_{p} - 1) \, a_{-r}^*a_{r+p}^* a_{p}  + N_0^{1/2} \sum_{p,r \in \Lambda^*} \widehat{V}_N(r) c_{r}c_{r+p} s_p
\, a_{-r}^*a_{r+p}^* a_{-p}^* \nn
\\ &+ N_0^{1/2} \sum_{p,r \in \Lambda^*} \widehat{V}_N(r) c_{r} s_{r+p} c_p \, a_{-r}^*a_{-(r+p)} a_{p} + N_0^{1/2} \sum_{p,r \in \Lambda^*} \widehat{V}_N(r) c_{r} s_{r+p} s_p \, a_{-r}^*  a_{-p}^* a_{-(r+p)} \nn
\\ &+ N_0^{1/2} \sum_{p,r \in \Lambda^*} \widehat{V}_N(r) s_{r} c_{r+p} c_p \, a_{r+p}^* a_{r} a_{p} + N_0^{1/2} \sum_{p,r \in \Lambda^*} \widehat{V}_N(r) s_{r} c_{r+p} s_p \,  a_{r+p}^* a^*_{-p} a_{r} \nn
\\ &+ N_0^{1/2} \sum_{p,r \in \Lambda^*} \widehat{V}_N(r) s_{r} s_{r+p} c_p \,  a_{r} a_{-(r+p)} a_{p} + \text{h.c.} 
\end{align} 
Here we already used the fact that terms that annihilate or create particles with momenta in $(H \cup S)^c$, as well as terms that do not preserve the parity of $\cN_H$ vanish on the range of $\Xi$. 
In particular, all contributions where two momenta contract in the canonical commutation relation must vanish since, by momentum conservation, these terms are either proportional to $a_0$ or to $a_0^*$. Additionally, the terms proportional to $s_r s_{r+p} s_p$ must vanish on the range of $\Xi$ because, since $s$ is zero on $S$, we would need $-r,p+r,p \in H$; this, however, breaks the parity of $\cN_H$. To bound the contributions in (\ref{eq:cEQ3}), we argue again with the parity of $\cN_H$ to conclude that among the three momenta $r,r+p,p$, two must be in $H$, one in $S$. Keeping this in mind, and moving around factors of $(\cN_H+1)^{\pm 1}$, $(\cN+1)^{\pm 1/2}$ as needed, we can estimate the terms on the right-hand side of (\ref{eq:cEQ3}) as follows.
Let us consider the first term. We use $|c_r c_{r+p} c_{p} - 1| \leq s_r^2 + s_{r+p}^2 + s_{p}^2$ and the notation $s_H$ for the restriction of the coefficients $s_p$ to the set $H$. Then, for $\xi \in \text{Ran} \, \Xi$, we obtain
\[ \begin{split}  \Big| N_0^{1/2} \sum_{p,r \in \Lambda^*} &\widehat{V}_N(r) (c_r c_{r+p} c_{p} - 1) \langle \xi, a_{r}^*a_{r+p}^* a_{p} \xi \rangle \Big|  \\ &\lesssim N^{-1/2+\kappa} \big( \| s^2_{H} \|_2 + \| s^2_{H} \|_\infty N^{3\kappa/4+3\eps/2} \big) \| (\cN_H + 1) \xi \| \| (\cN + 1)^{1/2} \xi \| \\ &\lesssim N^{-3 + 11\kappa/2 + 5\eps/2}  \| (\cN_H + 1) \xi \| \| (\cN + 1)^{1/2} \xi \|. \end{split} \]
Here we used that in the term proportional to $s_p^2$ we have $p \in H$ so that either $r \in S$ or $r+p \in S$; this allows us to sum up the second momentum after a Cauchy--Schwarz inequality (recall that $|S| \leq C N^{3\kappa/2 + 3\eps}$). 
Moreover, we used Lemma \ref{lem:vphi_prop} to bound $\| s^2_{H} \|_2 \leq C N^{-5/2 +9 \kappa/2 + 5 \eps/2}$, $\| s^2_{H} \|_\infty \leq CN^{-4 + 6\kappa + 4\eps}$ and the condition $4\vep-2+3\kappa<0$ from \eqref{eq:mom-sets}.

To estimate the second term on the right-hand side of (\ref{eq:cEQ3}), we assume for example that $r,p \in H$, $r+p \in S$ (the case $r+p, p \in H$, $r \in S$ can be treated similarly). Since $c$ is bounded uniformly in $N$ and with the estimate $\| s_{H} \|_2 \lesssim N^{(-1+3\kappa+\vep)/2}$ we find
\[ \begin{split} \Big| N_0^{1/2} \sum \widehat{V}_N (r)  c_{r} c_{r+p} s_p \, \langle \xi , a_{-r}^*a_{r+p}^* a_{-p}^* \xi \rangle \Big| 
&\lesssim N^{-1/2+\kappa} \sum s_p \| a_{-r} a_{-p}  \xi \| \big[ \| a_{r+p}  \xi \| + \| \xi \| \big]  \\ &\lesssim N^{-1/2 + \kappa} \| s_{H} \|_2 \| \cN_H  \xi \|  \big( \| \cN^{1/2} \xi \| + N^{3\kappa/4 + 3\eps/2} \| \xi \| \big) \\ &\lesssim N^{\frac{-2+5\kappa+\eps}{2}}\| \cN_H  \xi \|  \big( \| \cN^{1/2} \xi \| + N^{3\kappa/4 + 3\eps/2} \| \xi \| \big). \end{split} \] 

The next four terms on the right-hand side of (\ref{eq:cEQ3}) can all be estimated in the same way applying a Cauchy--Schwarz inequality. We obtain 
\begin{align*}
&\Big|  N_0^{1/2} \sum_{p,r \in \Lambda^*} \widehat{V}_N(r) c_{r} s_{r+p} c_p
 \, \langle \xi , a_{-r}^*a_{-(r+p)} a_{p} \xi \rangle \Big|, \; \Big| N_0^{1/2} \sum_{p,r \in \Lambda^*} \widehat{V}_N(r) c_{r} s_{r+p} s_p\, \langle \xi, a_{-r}^*  a_{-p}^* a_{-(r+p)} \xi \rangle \Big|,
\\  &\Big| N_0^{1/2} \sum_{p,r \in \Lambda^*} \widehat{V}_N(r) s_{r} c_{r+p} c_p 
\, \langle \xi, a_{r+p}^* a_{r}  a_{p} \xi \rangle \Big| , \; \Big| N_0^{1/2} \sum_{p,r \in \Lambda^*} \widehat{V}_N(r) s_{r} c_{r+p} s_p \, \langle \xi,  a_{r+p}^* a^*_{-p} a_{r}  \xi \rangle \Big|
\\ & \qquad \qquad \qquad \lesssim N^{-1/2 + \kappa} \| s_{H} \|_2 \| (\cN_H + 1) \xi \| \| (\cN +1)^{1/2} \xi \|  
\\ 
&\qquad \qquad \qquad\lesssim N^{\frac{-2 + 5\kappa+\eps}{2}} \| (\cN_H + 1) \xi \| \| (\cN +1)^{1/2} \xi\|.
\end{align*}

Finally, we control the last term on the right-hand side of (\ref{eq:cEQ3}). We find 
\[ \begin{split}  
\Big| N_0^{1/2} \sum_{p,r \in \Lambda^*} \widehat{V}_N(r) s_{r} &s_{r+p} c_p
\,  \langle \xi, a_{r}  a_{-(r+p)} a_{p} \xi \rangle \Big| \\ &\lesssim N^{-1/2+\kappa} \sum s_r s_{r+p} \| a_{r} a_{-(r+p)} \xi \| \big( \| a_p \xi \| + \| \xi \| \big) 
\\ &\lesssim N^{-1/2 + \kappa} \| \cN_H \xi \|  \big( \| s_{H} \|_2 \| s_{H} \|_\infty  \| \cN^{1/2} \xi \| + \| s_{H} \|_2^2 \| \xi \| \big)   \\ &\lesssim N^{-3+11\kappa /2 + 5\eps/2} \| \cN_H \xi \| \| \cN^{1/2} \xi \| + N^{-3/2 + 4\kappa +\eps} \| \cN_H \xi \| \| \xi \|.
\end{split} \]

Putting all together, we conclude that, for any $0< \kappa < 2/3$ and $\vep>0$ as in  \eqref{eq:mom-sets}, we have
\[ \begin{split}  |\langle \xi , \cE_1^{(Q_3)} \xi \rangle | 
&\lesssim N^{\frac{-2+5\kappa+\eps}{2}}\| (\cN_H +1) \xi \|  \big( \| \cN^{1/2} \xi \| + N^{3\kappa/4+3\vep/2} \| \xi \| \big) \\ &\leq  N^{-2+4\kappa+5\eps} \| (\cN_H +1) \xi \|^2 + N^{\kappa- 4\eps} \big( \| \cN^{1/2} \xi \|^2 + N^{3\kappa/2+3\vep} \| \xi \|^2 \big).  \end{split}  \]
\end{proof}

Finally, let us compute the action of $e^{\cB_1}$ on the quartic term $Q_4$. Here, the restriction on $\kappa$ and $\vep$ only serves to simplify the form of the error terms.
\begin{lemma} \label{lem:TLc_Q4}
We have 
\begin{equation} \label{eq:BQB} 
e^{-\cB_1} Q_4 e^{\cB_1} = Q_4 + \frac{1}{2}\sum_{p,r \in \Lambda^*} \hat{V}_N(r-p) s_r \big( a_p^*a_{-p}^* + \hc \big) + \frac{1}{2}\sum_{p,r \in \Lambda^*} \hat{V}_N(p-r)s_p s_r + \mathcal{E}_{1}^{(Q_4)},
\end{equation} 
where, on the range of $\Xi$, 
\[ \begin{split} \pm \cE^{(Q_4)}_1 \lesssim  N^{-\kappa/2-\vep} \cN^2  + N^{-3+7\kappa+4\eps} \cN_H^2 +  N^{2-2\kappa-2\eps} \cN_H +  N^{\kappa-3\eps} \cN + N^{5\kappa/2-\eps},
 \end{split} \]
if $1/2 < \kappa < 8/15-2\vep/3$ and $\eps > 0$ is chosen as in \eqref{eq:mom-sets}.
\end{lemma}
\begin{proof}
With (\ref{eq:bogo1}), we decompose
\begin{equation}\label{eq:Aj}  e^{-\cB_1} Q_4 e^{\cB_1} = \frac{1}{2} \sum_{j=0}^4 A_j \end{equation}
where $A_j$ collects all contributions with exactly $j$ coefficients $s_k$, with $k \in \{p+r,q, q+r, p \}$. We will see that the contributing terms stem from $A_0,A_1$ and $A_2$.
For $j=0$, we clearly have 
\[ A_0 = \sum_{p,r,q \in \Lambda^*}  \widehat{V}_N (r) c_{p+r} c_q c_{q+r} c_p a^*_{p+r} a^*_q a_{q+r} a_p. \]
We observe that 
\[ A_0 - 2 Q_4 = \sum_{p,r,q \in \Lambda^*}  \widehat{V}_N (r)  \big[ c_{p+r} c_q c_{q+r} c_p -1 \big] a^*_{p+r} a^*_q a_{q+r} a_p. \]
Estimating $|c_{p+r} c_q c_{q+r} c_p -1| \leq s_{p+r}^2 + s_q^2 +s_{q+r}^2 + s_p^2$, we find 
\[ \big| \langle \xi, (A_0 - 2Q_4) \xi \rangle \big| \leq 4 \Big| \sum \widehat{V}_N (r)  s_p^2 \langle  a_{p+r} a_q \xi , a_{q+r} a_p \xi \rangle \Big|.  \]
For $\xi$ in the range of $\Xi$, we must have $p+r, q, q+r , p \in H \cup S$. Since $s_p = 0$ for $p \in S$, we can assume $p \in H$. To preserve the parity of $\cN_H$, there must be a second momentum in $H$. Assuming for example $q \in H$, we can estimate, using Lemma \ref{lem:vphi_prop}, 
\[  \begin{split} \Big| \sum_{p,q \in H, r \in \Lambda^*} &\widehat{V}_N (r)  s_p^2 \langle  a_{p+r} a_q \xi , a_p  a_{q+r} \xi \rangle \Big| \\  \lesssim \; & \sum_{p,q \in H, r\in \Lambda^*}  | \widehat{V}_N (r)|  |s_p| \| a_{p+r} a_q \xi \| \|  a_p  a_{q+r} \xi \| \\ \lesssim \; & \Big[ \sum_{p,q \in H, r \in \Lambda^*} s_p^4 \| a_{p+r} a_q \xi \|^2 \Big]^{1/2} \Big[  \sum_{p,q \in H, r \in \Lambda^*} \widehat{V}_N^2 (r) \| a_p a_{q+r} \xi \|^2 \Big]^{1/2} \\ \lesssim \; & N^{-2+4\kappa+5\eps/2}\| \cN^{1/2}  \cN^{1/2}_H \xi \|^2, \end{split} \]
where we used $\| \hat{V}_N\|_2 \lesssim N^{(1-\kappa)/2}$ and $s_p \lesssim N^\kappa p^{-2}$.
Moving factors of $(\cN+1)^{1/2}, (\cN_H + 1)^{1/2}$ around, we can similarly bound also contributions arising when $p+r \in H$ or $q+r \in H$. With the assumption $-2+3\kappa+4\vep<0$ from \eqref{eq:mom-sets}, we conclude that, on the range of $\Xi$,  
\begin{equation} \label{eq:A0f} \pm \Big(\frac{1}{2} A_0 - Q_4 \Big)  \lesssim N^{-2+4\kappa +5\eps/2} (\cN+1) (\cN_H +1) \lesssim N^{-\kappa/2-\vep}\cN^2 + 
N^{-4+17\kappa/2 + 6\eps} \cN_H^2 + N^{5\kappa/2-\vep}.
\end{equation} 
Next, we consider the term $A_1$ in (\ref{eq:Aj}). Arranging operators in normal order we find 
\[ \begin{split} A_1 = 2 \hspace{-.2cm}  \sum_{p,q,r \in \Lambda^*} \hspace{-.1cm} \widehat{V}_N (r) s_{p+r} c_q c_{q+r} c_p \big( a_q^* a_{-p-r}  a_{q+ r} a_p + \text{h.c.} \big)+ \hspace{-.1cm}  \sum_{p,r \in \Lambda^*} \hspace{-.1cm} \widehat{V}_N (r) s_{p+r} c_{p+r} c_p^2 \left( a_p a_{-p} + \text{h.c.} \right).  \end{split} \]
Let us start to bound the first term. On the range of $\Xi$, all momenta must be in $H \cup S$. In particular, $p+r \in H$. To preserve parity of $\cN_H$, there can be, in total, two or four momenta in $H$. Handling these two cases separately (if all momenta are in $H$, we need to use $\widehat{V}_N$ to perform one of the sum; if instead 2 momenta are in $S$, we can use $|S| \leq N^{3\kappa/2 + 3\eps}$), we arrive at
\[ \begin{split} \Big| \sum_{p,q,r} &\widehat{V}_N (r) s_{p+r} c_q c_{q+r} c_p \langle \xi, a_q^* a_{-p-r} a_{q+ r} a_p  \xi \rangle \Big| \\ &\lesssim N^{\kappa+\vep/2} \| (\cN_H +1) \xi \|^2 + N^{\frac{-6+13\kappa+2\vep}{4}} \| (\cN+1)^{1/2}  (\cN_H +1)^{1/2} \xi \|^2. \end{split} \]
As for the quadratic contribution to $A_1$, we observe that 
\[ \begin{split} \Big| & \sum_{p,r} \widehat{V}_N (r) s_{p+r} ( c_{p+r} c_p^2 - 1 ) \langle \xi, a_p a_{-p} \xi \rangle \Big| 
\\ & \lesssim \sum_{p,r} |\widehat{V}_N (r)| s_{p+r} \big( s_{p+r}^2  + s_p^2 \big)  \big[ \| a_p \xi \|^2 + \| a_p \xi \| \| \xi \| \big] 
\\
&\lesssim \|s^2\|_\infty \||\hat{V}_N| \ast s\|_\infty \| \cN^{1/2} \xi \|^2 + \||\hat{V}_N| \ast s\|_\infty \|s^2\|_2  \|\cN^{1/2} \xi\| \|\xi\| + \sum_p \big(|\hat{V}_N| \ast s^3\big) (p) \|a_p \xi\| \| \xi\|.
\end{split} \]
The first two terms are bounded with $\| |\hat{V}_N| \ast s\|_\infty \lesssim N^\kappa$.
To bound the last term, we distinguish the cases $p \in H$ and $p \in S$ (on the range of $\Xi$, there is no other possibility). We find, with $\|  |\widehat{V}_N | * s^3 \|_2 \lesssim \|\hat{V}_N\|_2 \|s^3\|_1 \lesssim N^{1/2+\kappa -3\eps}$,
\[ \begin{split}  \sum_{p} \big(|\widehat{V}_N| \ast s^3 \big) (p) \| a_p \xi \| \| \xi \|  &\lesssim \| |\widehat{V}_N| * s^3 \|_2 \| \cN_H^{1/2} \xi \|  \| \xi \| + \| |\widehat{V}_N| * s^3 \|_\infty N^{3\kappa/4+3\eps/2} \| \cN^{1/2} \xi \| \| \xi \| 
\\
&\lesssim N^{1/2+\kappa-3\vep}\| \cN_H^{1/2} \xi \|  \| \xi \| + N^{7\kappa/4-5\vep/2} \|\cN^{1/2} \xi \| \| \xi \| 
\end{split} \]
We conclude that, on $\text{Ran } \Xi$, 
\begin{align} \label{eq:A1f} &\pm  \Big[ A_1 - \sum_{p,r \in \Lambda^*} \widehat{V} (r) s_{p+r} \big( a_p a_{-p} + \text{h.c.} \big) \Big] 
\\
& \qquad \lesssim N^{\kappa+\vep/2} (\cN_H^2+1) + N^{-\kappa/2-\vep}\cN^2 + N^{-3+7\kappa+2\vep} \cN_H^2 + N^{\kappa - 4\vep} \cN + N^{1-\kappa/2-5\vep}\cN_H + N^{5\kappa/2-\vep} \nn
\end{align} 
for any $0 < \kappa < 2/3$ and $\eps>0$ as in \eqref{eq:mom-sets}.

We switch our attention to the term
\begin{equation}\label{eq:A2} \begin{split} A_2 = \; &\sum \widehat{V}_N (r) c_{p+r} c_q s_{q+r} s_p \big( a_{p+r}^* a_q^* a_{-q-r}^* a_{-p}^* + \text{h.c.} \big) \\ &+2 \sum \widehat{V}_N (r) c_{p+r} s_q c_{q+r} s_p a_{p+r}^* a_{-p}^* a_{-q} a_{q+r}  \\ &+ 2  \sum \widehat{V}_N (r) c_{p+r} s_q s_{q+r} c_p a_{p+r}^* a_{-q-r}^* a_{-q} a_p \\
&+  \sum \widehat{V}_N (r) c_q s_q c_{q+r} s_{q+r} \big( 4 a_q^* a_q + 1) \\
&+ 2  \sum \widehat{V}_N (r) c_p^2 s_{p+r}^2  a_p^* a_p + 2  \widehat{V}_N (0) \sum c_p^2 s_q^2  a_p^* a_p. 
\end{split} \end{equation} 
The only contributing term is the constant in the fourth line. We estimate the other terms.
To control the term in the first line, we observe that, on the range of $\Xi$, the four momenta $p+r, q, -q-r,-p$ must be in $H \cup S$ and, more precisely, either 4 or 2 of them must be in $H$ (to preserve parity). Let us assume, first, that all four momenta are in $H$. Denoting by $\scriptstyle \overset{*}{ \sum}$ the sum over all $p,q,r \in \Lambda^*$, with $p+r, q, -q-r, -p \in H$, we can bound 
\[ \begin{split} \Big| \sum^* &\widehat{V}_N (r) c_{p+r} c_q s_{q+r} s_p \langle \xi, a_{p+r}^* a_q^* a_{-q-r}^* a_{-p}^* \xi \rangle \Big| 
\\ \lesssim \; &\sum^* |\widehat{V}_N (r)| s_{q+r} s_p \| a_{p+r} a_{-p} a_{-q-r} (\cN_H + 1)^{-1/2} \xi \| \left[ \| a_q (\cN_H+1)^{1/2} \xi \| + \| (\cN_H + 1)^{1/2} \xi \|\right] \\ 
\lesssim \; &\| (\cN_H+1) \xi \| \bigg( \| \widehat{V}_N \|_\infty \Big[ \sum^* s_{q+r}^2 s_p^2 \| a_q (\cN_H+1)^{1/2} \xi \|^2 \Big]^{1/2} \\ &\hspace{6cm} + \| (\cN_H + 1)^{1/2} \xi \| \Big[ \sum^* |\widehat{V}_N (r)|^2 s_{q+r}^2 s_p^2 \Big]^{1/2} \bigg) \\
\lesssim \; & N^{-2 + 4\kappa + \vep} \| (\cN_H + 1) \xi \|^2  + N^{-1/2 + 5\kappa/2 + \vep} \| (\cN_H + 1) \xi \| \| (\cN_H + 1)^{1/2}  \xi \|  \\
\lesssim \; & \left(N^{-2+4\kappa+\vep} + N^{-3+7\kappa+4\eps}\right) \| (\cN_H + 1) \xi \|^2 +  N^{2-2\kappa-2\eps} \| ( \cN_H + 1)^{1/2} \xi \|^2 
\end{split} \]
for $0 < \kappa < 2/3$ (we absorb the second contribution by a Cauchy--Schwarz inequality with appropriate weights).
If two momenta are in $H$ and two in $S$ (in this case, we must have $p+r,q \in S$, $q+r, p \in H$, since $s = 0$ on $S$), we proceed similarly, but instead of using the potential $\widehat{V}_N$ to sum over the third momentum (as we did in the fourth line of the last equation), we use the restriction onto $S$, estimating $\widehat{V}_N$ in $\ell^\infty$. Indicating with $\scriptstyle\overset{**}{\sum}$ the sum over $p+r, q \in S, q+r,p \in H$, we find 
\[ \begin{split}  \Big| \sum^{**} & \widehat{V}_N (r) c_{p+r} c_q s_{q+r} s_p \langle \xi, a_{p+r}^* a_q^* a_{-q-r}^* a_{-p}^* \xi \rangle \Big| \\ 
\lesssim \; & N^{-2+4\kappa+\vep} \| (\cN \!+\! 1)^{1/2} (\cN_H  \!+\! 1)^{1/2} \xi \|^2  +  N^{-2+\frac{19\kappa}{4} + \frac{5\eps}{2}}  \| (\cN  \!+\! 1)^{1/2} (\cN_H  \!+\! 1)^{1/2} \xi \| \| (\cN_H  \!+\! 1)^{1/2}  \xi \|   \\
\lesssim \; & N^{-2+4\kappa+2\vep}\| (\cN+1)^{1/2} (\cN_H + 1)^{1/2} \xi \|^2 + N^{-2+11\kappa/2+3\vep} \| (\cN_H + 1)^{1/2} \xi \|^2.
\end{split} \]
The second and the third terms on the right-hand side of (\ref{eq:A2}) can be bounded by a simple Cauchy--Schwarz inequality (remarking that, on the range of $\Xi$, the momenta associated with the coefficients $s$ must lie in $H$). We find
\[ \begin{split}  \Big| \sum \widehat{V}_N (r) c_{p+r} s_q c_{q+r} s_p \langle \xi, a_{p+r}^* a_{-p}^* a_{-q} a_{q+r} \xi \rangle \Big| &\lesssim N^{-2+4\kappa+\eps} \| \cN_H^{1/2} \cN^{1/2}  \xi \|^2 \\  \Big| \sum \widehat{V}_N (r) c_{p+r} s_q s_{q+r} c_p \langle \xi, a_{p+r}^* a_{-q-r}^* a_{-q} a_{p} \xi \rangle \Big| &\lesssim N^{-2+4\kappa+\eps} \| \cN_H^{1/2} \cN^{1/2} \xi \|^2. \end{split}  \]
The quadratic part of the fourth term in (\ref{eq:A2}) can be bounded, using that $\| \widehat{V}_N * s \|_\infty \lesssim N^\kappa$ and that $q \in H$ on the range of $\Xi$ (because $s_q = 0$ for $|q| \leq N^{\kappa/2+\eps}$). We find 
\[ \Big| \sum \widehat{V}_N (r) c_{q} s_q c_{q+r} s_{q+r} \langle a_q \xi , a_q \xi \rangle  \Big| 
\lesssim N^{-2+4\kappa+2\vep} \| \cN_H^{1/2} \xi \|^2.
\]
As for the constant term, we extract a contribution observing that  
\begin{align*}
& \Big| \sum \widehat{V}_N (r) (c_q c_{q+r} -1) s_q s_{q+r} \Big| 
\leq \sum |\hat{V}_N(r)| (s_q^3 s_{q+r} + s_q s_{q+r}^3)
 \lesssim \| \widehat{V}_N\|_\infty \|s\|_1\|s^3\|_1 \lesssim N^{5\kappa/2 - 3\eps},
\end{align*}
where we used that $\|s\|_1 \lesssim N\|\vphi\|_1 \lesssim N$.
The terms on the last line of (\ref{eq:A2}) are simply estimated by
\begin{align*}
\Big| \sum \widehat{V}_N (r)  c_p^2 s_{p+r}^2 \langle \xi, a_p^* a_p \xi \rangle \Big|  ,  \; |\widehat{V}_N (0)| \Big| \sum c_p^2 s_q^2 \langle \xi, a_p^* a_p \xi\rangle \Big| 
& \lesssim N^{-1 + 5\kappa/2-\eps} \| \cN^{1/2} \xi \|^2.
\end{align*}
Thus, on $\text{Ran } \Xi$, 
\begin{align}\label{eq:A2f}  \pm \Big[ A_2 - \sum_{q,r \in \Lambda^*} \widehat{V}_N (r) s_q s_{q+r} \Big] &\lesssim \; N^{-3+7\kappa + 4\eps} \cN^2_H +  \big(N^{2-2\kappa-2\eps} + N^{-2+11\kappa/2+3\vep}\big) (\cN_H +1) \nn
\\
& \qquad + N^{-2+4\kappa+2\eps} \cN \cN_H + N^{-1+5\kappa/2-\vep}\cN  + N^{5\kappa/2-\eps} \nn
\\
&\leq N^{-3+7\kappa + 4\eps} \cN^2_H +  \big(N^{2-2\kappa-2\eps} + N^{-2+11\kappa/2+3\vep}\big) (\cN_H +1) \nn
\\
& \qquad + N^{-\kappa/2-\vep}\cN^2 + N^{-1+5\kappa/2-\vep}\cN  + N^{5\kappa/2-\eps}
\end{align} 
for all $0<\kappa <2/3$ and for $\eps>0$ as in \eqref{eq:mom-sets}.

Next, we consider 
\[ \begin{split} 
A_3 = \; & 2  \sum \widehat{V}_N (r) c_{p+r} s_q s_{q+r} s_p \big( a_{p+r}^* a_{-q-r}^* a_{-p}^* a_{-q} + \text{h.c.} \big) \\
&+ 2  \sum \widehat{V}_N (r) s_q c_q s_{q+r}^2 \big(a_q^* a_{-q}^* + \text{h.c.} \big) \\
&+  \sum \widehat{V}_N (r) s_q c_q s_{q+r}^2 \big( a_{q+r}^* a_{-q-r}^* + \text{h.c.} \big) \\
&+2  \widehat{V}_N (0)  \sum s_q^2 s_p c_p \big(a_p^* a_{-p}^* + \text{h.c.} \big).
\end{split} \]
We can control the quartic term noticing that, on the range of $\Xi$, all momenta $p+r, q, q+r, p$ must be in $H$ (because $s=0$ on $S$ and because of parity). With $\| s_{H} \|_\infty \lesssim N^{3\kappa -2 +2\eps}$, we find, for $\xi \in \text{Ran } \Xi$,  
\[ \begin{split} \Big| \sum \widehat{V}_N (r) c_{p+r} s_q s_{q+r} s_p \langle \xi, a_{p+r}^* a_{-q-r}^* a_{-p}^* a_{-q} \xi \rangle \Big| 
&\lesssim \|\hat{V}_N\|_\infty \|s_H\|_\infty \|s_H\|_2^2 \|(\cN_H+1) \xi \|^2
\\
&\lesssim N^{-4+7\kappa+3\vep} \|(\cN_H+1) \xi \|^2.
  \end{split} \]
The quadratic terms can be estimated, for $\xi \in \text{Ran } \Xi$, by 
\[ \begin{split} \Big| \sum \widehat{V}_N (r) s_q &c_q s_{q+r}^2 \langle \xi, a_q^* a_{-q}^* \xi \rangle \Big| , \Big| \widehat{V}_N (0) \sum s_q^2 s_p c_p  \langle \xi, a_p^* a_{-p}^* \xi \rangle \Big| 
\\ &\lesssim \|\hat{V}_N\|_\infty \| s \|_2^2 \| s_{H} \|_2 \| (\cN_H+1)^{1/2} \xi \|^2 \lesssim N^{-3/2 + 4\kappa -\eps /2} \| (\cN_H + 1)^{1/2} \xi \|^2 \end{split} \] 
and by 
\[ \begin{split} \Big| \sum \widehat{V}_N (r) s_{q+r} c_{q+r} s_{q}^2 \langle \xi, a_q^* a_{-q}^* \xi \rangle \Big| &\lesssim \| \widehat{V}_N * s \|_\infty \| s^2_{H} \|_2  \| (\cN_H + 1 )^{1/2} \xi \|^2 \\ &\lesssim N^{-5/2 + 11\kappa/2  + 5\eps /2} \| (\cN_H+1)^{1/2} \xi \|^2 \end{split} \] 
We conclude that, on the range of $\Xi$, 
\begin{equation} \label{eq:A3f} \pm A_3 \lesssim  N^{-4+7\kappa+3\eps} \cN_H^2 +  N^{-3/2 + 4\kappa +\eps/2} \cN_H + N^{5\kappa/2-\vep}  \end{equation} 
for all $0 <\kappa < 2/3$ and for $\eps > 0$ as in \eqref{eq:mom-sets}. 

Finally, we bound the contribution
\[ \begin{split} 
A_4 = \; & \sum \widehat{V}_N (r) s_{p+r} s_q s_{q+r} s_p a^*_{-q-r} a^*_{-p} a_{-p-r} a_{-q} \\
&+  \sum \widehat{V}_N (r) s_{p+r}^2 s_p^2 \big( 2 a_{-p}^* a_{-p} +1 \big) +   \widehat{V}_N (0) \sum s_p^2 s_q^2 \big( 2 a_{-q}^* a_{-q} + 1 \big).
\end{split} \]
Proceeding similarly as we did for $A_3$, we easily find that, on $\text{Ran } \Xi$, 
\begin{align*}
A_4 &\lesssim \|\hat{V}_N\|_\infty \|s_H^2\|_\infty \|s_H\|_2^2 \cN_H^2 + \|\hat{V}_N\|_\infty \|s_H^2\|_\infty \|s^2\|_1 \cN_H + \|\hat{V}_N\|_\infty \|s^2\|_1^2
\\
&\lesssim N^{-6+10\kappa+5\eps} \cN_H^2 + N^{-5+17\kappa/2 + 3\eps} \cN_H + N^{-1+4\kappa-2\vep}.
\end{align*}
Combining the last equation with (\ref{eq:A0f}), (\ref{eq:A1f}), (\ref{eq:A2f}), (\ref{eq:A3f}), we obtain (\ref{eq:BQB}), with an error $\cE^{(Q_4)}_1$ satisfying, on $\text{Ran } \Xi$, 
\[ \begin{split} \pm  \cE^{(Q_4)}_1  \lesssim \; &N^{-\kappa/2-\vep} \cN^2  + N^{-3+7\kappa+4\eps} \cN_H^2 + N^{2-2\kappa-2\eps} \cN_H + N^{\kappa-3\eps} \cN + N^{5\kappa/2-\eps}, \end{split} \]
where we used the assumption $1/2 < \kappa < 8/15-2\vep/3$ (in particular $\vep < 1/20$).
\end{proof} 

\subsection{Proof of Lemma \ref{lem:TLc_main}}
We are now ready to show Lemma \ref{lem:TLc_main}.
Combining the statements of Lemmas \ref{lem:TLc_kinetic}, \ref{lem:TLc_Q2}, \ref{lem:TLc_Q3} and \ref{lem:TLc_Q4}, we find 
\begin{equation} \label{eq:TLc_main}
\begin{split} 
e^{-\cB_1} \weyl^* &\cH_N \weyl e^{\cB_1} 
\\
= \; &\frac{N_0^2}{2} \widehat{V}_N (0) + \sum_{p \in \Lambda^*}p^2 a_p^*a_p + 2N^\kappa \widehat{V} (0) \cN + Q_3 + Q_4 \\
& + \sum_{p\in\Lambda^*} \Big[ p^2 s_p  + \frac{1}{2} N_0 \hat{V}_N(p) + \frac{1}{2}\sum_{r \in \Lambda^*} \hat{V}_N(r-p) s_r \Big] \, ( a_p^*a_{-p}^* + \hc ) \\
& +  \sum_{p \in \Lambda^*} \Big[ p^2s_p^2 + N_0 \hat{V}_N(p) s_p + \frac{1}{2} \sum_{r \in \Lambda^*} \hat{V}_N(p-r)s_p s_r \Big]  + \widetilde{\cE}_1 \\
\end{split} 
\end{equation} 
where, on the range of $\Xi$,  
\begin{equation} \label{eq:wtE1} \begin{split}  \widetilde{\cE}_1 \lesssim \; & N^{-3\vep} \sum_{p \in \Lambda^*} p^2a_p^*a_p +  N^{-\kappa/2-\vep} \cN^2 + N^{-3+7\kappa+4\vep} \cN_H^2  \\ &+  N^{2-2\kappa-2\eps} \cN_H +  N^{\kappa-\eps/2} \cN + N^{5\kappa/2 - \vep/2} , \end{split} \end{equation}  
if $1/2 < \kappa < 8/15-2\vep/3$. 
Here we also used the fact that the contribution $e^{-\cB_1} Q_1 e^{\cB_1}$, with $Q_1$ as in (\ref{eq:Weyl}), vanishes on the range of $\Xi$. With a slight abuse of notation that is only used in this subsection, we introduce the following momentum sets 
$$
L = \{ p \in \Lambda^* \backslash \{ 0 \} : |p| \leq N^{\kappa/2+\eps} \}, \quad L^c = \{ p \in \Lambda^* : |p| > N^{\kappa/2+\eps} \}.
$$

The constant term on the last line of (\ref{eq:TLc_main}) can be rewritten, using \eqref{eq:sp_def}, \eqref{eq:vphi} and \eqref{eq:scatlength}. We find
\begin{align*}
\sum_{p \in \Lambda^*} \Big[ p^2s_p^2 &+ N_0 \hat{V}_N(p) s_p + \frac{1}{2} \sum_{r \in \Lambda^*} \hat{V}_N(p-r)s_p s_r \Big]
\\
= \; &\frac{N_0^2}{2} \sum_{p \in L^c} \widehat{V}_N (p) \ph_p -\frac{N_0^2}{2} \sum_{\substack{ r\in L \\ p \in L^c}} \widehat{V}_N (p-r) \ph_p \ph_r \\
= \; &\frac{N_0^2}{2N^{1-\kappa}} (8\pi \mathfrak{a}_N - \widehat{V} (0) ) - N_0^2 \sum_{p \in L} \ph_p \left[ \frac{\widehat{V}_N (p)}{2} + \frac{1}{2} \sum_{r \in L^c } \widehat{V}_N (p-r) \ph_r \right] \\
=\; & \frac{N_0^2}{2N} N^\kappa (8\pi \mathfrak{a}_N - \widehat{V} (0) ) + N_0^2 \sum_{p \in L} p^2 \ph_p^2 + \frac{N_0^2}{2} \sum_{p,r \in L} \widehat{V}_N (p-r) \ph_p \ph_r.
\end{align*}
The last term is bounded by $C N^{-1+4\kappa+2\vep}$. We write
\[ \frac{N_0^2}{2} \sum_{p,r \in L} \widehat{V}_N (p-r) \ph_p \ph_r = \mathcal{O}\left( N^{-1+4\kappa+2\eps}\right). \] 
As for the second term, we observe that 
$$
N_0^2 \sum_{p \in L} p^2\vphi_p^2 = N_0^2 \sum_{p \in S}p^2\vphi_p^2 + \mathcal{O} \big( N^{5\kappa/2 - \vep}\big).
$$
Estimating  
$
|\hat{V}_N(p-r) - \hat{V}_N(r)| \lesssim N^{-2+2\kappa} |p| 
$
and using again equations \eqref{eq:vphi}, \eqref{eq:scatlength}, we deduce that 
\begin{equation} \label{eq:vphi_approx}
\Big| p^2 \vphi_p + \frac{4\pi \mathfrak{a}_N N^\kappa}{N} \Big| \lesssim N^{-2+2\kappa} |p|.
\end{equation}
Together with the bound $|\mathfrak{a}_N - \mathfrak{a} | \lesssim N^{-1+\kappa}$ from Lemma \ref{lem:vphi_prop}, this yields
\[
\begin{split} 
& \sum_{p \in \Lambda^*} \Big[ p^2 s_p^2 + N_0 \hat{V}_N(p) s_p +  \frac{1}{2} \sum_{r \in \Lambda^*} \hat{V}_N(p-r)s_p s_r \Big]  
\\
&= \frac{N_0^2}{2N} N^\kappa (8\pi \mathfrak{a} - \widehat{V} (0))  + \frac{N_0^2}{N^2} \sum_{p \in S} \frac{(4\pi \mathfrak{a} N^\kappa)^2}{p^2} + \mathcal{O}\big( N^{5\kappa/2-\eps}\big),
\end{split} 
\]  
if $-2+3\kappa+6\vep < 0$. Combining this with the first term on the right-hand side of (\ref{eq:TLc_main}), we conclude that 
\begin{equation}  \label{eq:TLc_constant}
\begin{split} \frac{N_0^2}{2} \widehat{V}_N (0) + &\sum_{p \in \Lambda^*} \Big[ p^2 s_p^2 + N_0 \hat{V}_N(p) s_p + \frac{1}{2} \sum_{r \in \Lambda^*} \hat{V}_N(p-r)s_p s_r \Big] \\ 
&=  4\pi \mathfrak{a} N^{1+\kappa} - 8\pi \mathfrak{a} N^\kappa (N-N_0) +  \sum_{p \in S} \frac{(4\pi \mathfrak{a} N^\kappa)^2}{p^2} + \mathcal{O}\big(
N^{5\kappa/2-\eps} \big),
\end{split} 
\end{equation} 
where we used the assumption $0 \leq N- N_0 \lesssim N^{3\kappa/2}$.

To deal with the off-diagonal quadratic term in \eqref{eq:TLc_main}, we use again \eqref{eq:vphi}, \eqref{eq:sp_def} and \eqref{eq:scatlength}. As a form on the range of the projection $\Xi = \1_{\{\cN_{(S \cup H)^c} = 0\}}\1_{\{\cN_H \in 2\mathbb{N}_0\}} $, we find  
\begin{equation} \begin{split}  \label{eq:Q2_renorm}
\sum_{p \in \Lambda^*} \Big[ p^2 s_p  &+ \frac{1}{2} N_0 \hat{V}_N(p) + \frac{1}{2}\sum_{r \in \Lambda^*} \hat{V}_N(r-p) s_r \Big]  (a_p^*a_{-p}^* + \hc ) \\ &= -N_0 \sum_{p \in S} p^2 \vphi_p \, (a_p^*a_{-p}^* + \hc )   - \frac{N_0}{2}\sum_{\substack{p \in \Lambda^* \\ r\in L}} \hat{V}_N(r-p) \vphi_r \, (a_p^*a_{-p}^* + \hc ).
\end{split}
\end{equation} 
Recalling \eqref{eq:vphi_approx}, we can estimate 
\[ \pm \Big[ - N_0 \sum_{p \in S} p^2 \vphi_p \, ( a_p^*a_{-p}^* + \hc )  - 4\pi \mathfrak{a} N^\kappa \sum_{p\in S} ( a_p^* a_{-p}^* + \hc ) \Big] \lesssim N^{-1+5\kappa/2+5\eps/2} (\cN + N^{3\kappa/2}). \]
As for the last term on the right-hand side of (\ref{eq:Q2_renorm}), we can bound it switching to position space. For $\xi \in \cF(\Lambda)$ we find  
\begin{align*}
&\Big| \frac{N_0}{2} \sum_{\substack{p \in \Lambda^* \\ r\in L}} \hat{V}_N(r-p) \vphi_r  \langle \xi, a_p^*a_{-p}^* \xi \rangle \Big| = \Big| \frac{N_0}{2} \int_{\Lambda^2} V_N(x-y)\sum_{r\in L} e^{i(x-y)r} \vphi_r \langle a_x a_y \xi , \xi \rangle \, \dd x \dd y \Big| \\
&\leq N_0 \sum_{r\in L} |\vphi_r| \left(\delta  \langle \xi, Q_4 \xi \rangle + C \delta^{-1} N^{-1+\kappa} \| \xi \|^2 \right)
\lesssim N^{ -1/2+3\kappa/4} \left(\langle \xi, Q_4 \xi \rangle + N^{5\kappa/2+2\vep} \| \xi \|^2 \right), 
\end{align*}
choosing $\delta = N^{-1/2-3\kappa/4-\vep}$. We conclude that
\[ \sum_{p \in \Lambda^*} \Big[ p^2 s_p  + \frac{1}{2} N_0 \hat{V}_N(p) + \frac{1}{2}\sum_{r \in \Lambda^*} \hat{V}_N(r-p) s_r \Big]  (a_p^*a_{-p}^* + \hc ) = 4\pi \mathfrak{a} N^\kappa \sum_{p\in S} ( a_p^* a_{-p}^* + \hc ) + \cE'_1 \]
where, on the range of $\Xi$, 
\[  \pm \cE'_1 \lesssim N^{-5\eps/2} Q_4 + N^{-1+5\kappa/2+5\eps/2} \cN + N^{5\kappa/2-\eps/2}, \]
if $-2+3\kappa+10\vep < 0$ .
Combining this estimate with (\ref{eq:wtE1}), and simplifying the absorbing the term proportional to $\cN$ by a Cauchy--Schwarz inequality, we arrive at (\ref{eq:claimB1}), (\ref{eq:claimB1E}) given that $1/2 < \kappa < 8/15-2\vep/3$ (note that this implies the previous inequality on $\vep$ and $\kappa$).
\qed

\section{Cubic transformation} 
\label{sec:cubic} 

The goal of this section is to show Lemma \ref{lem:Tc_main} and Proposition \ref{prop:NN2}.
All lemmas in this section hold for all $0 < \kappa < 2/3$ and $\vep>0$ as introduced in \eqref{eq:mom-sets}.

\subsection{Properties of the cubic transformation} 

We start by observing that the action of the transformation $T_k = e^{\cB_k}$, defined in \eqref{eq:Tc-def}, can be computed explicitly. To this end, we define, for $k \in S$, 
\[ X_k = |\cB_k | = \sqrt{\cB_k^\circ \cB_k^\sharp}. \]

\begin{lemma}\label{lem:com_BB}
Let $k\in S$. We have
\begin{align}\label{eq:exp_Xk2}
X_k^2 = \cB_k^\circ \cB_k^\sharp = \sum_{r \in H_k} N \vphi_r (\varphi_r + \varphi_{r+k})  a^*_k a_k\Theta_{k,r}.
\end{align}
Furthermore, recalling the notation $\cN_k = a_k^* a_k$, 
\begin{equation} \label{eq:Xp_bound}
X_k^2\leq 2 N \|\varphi_H\|^2_2 \cN_k \lesssim N^{-2+3\kappa+\vep}\cN_k. 
\end{equation}

\end{lemma}

\begin{proof}
Let us compute 
\begin{equation}\label{eq:BcircBdag} 
\begin{split} 
\cB_k^\circ \cB_k^\sharp 
	&= \sum_{r,t \in H_k} N \varphi_r\varphi_t\Theta_{k,r} a^*_k a_{r+k} a_{-r} a^*_{-t} a^*_{t+k}  a_k\Theta_{k,t} \\
	&= \sum_{r,t \in H_k} N \varphi_r\varphi_t \Theta_{k,r} a^*_k [a_{k+r} a_{-r}, a^*_{-t} a^*_{t+k} ] a_k\Theta_{k,t} \\
	&= \sum_{r,t \in H_k} N \varphi_r\varphi_t \Theta_{k,r} a^*_k \left( \delta_{t,r} + \delta_{k+r+t,0}\right) \left(1 + a^*_{-r} a_{-r} + a^*_{r+k}a_{r+k}\right) a_k\Theta_{k,t} \\
	&= \sum_{r \in H_k} N \vphi_r (\varphi_r + \varphi_{r+k})  a^*_k a_k\Theta_{k,r}.
\end{split}
\end{equation}
To obtain the second equality, we used that $a_{r+k} a_{-r} \Theta_{k,t} = 0$ (since the cutoff imposes that there is no $k$-connection). In the last identity, we used $\Theta_{k,r} a_{-r}^* = \Theta_{k,r} a_{r+k}^* = 0$ (because the cutoff imposes that the shell-vicinity of $-(k+r)$ and $r$ are empty) and $\Theta_{k,-(r+k)} = \Theta_{k,r}$.
The bound (\ref{eq:Xp_bound}) then follows from \Cref{lem:vphi_prop}.
\end{proof}

\begin{lemma}\label{lem:Tk}
For $k\in S$ we have
\begin{align*}
T_k &= \cos X_k + \mathcal B_k^\sharp \frac{\sin X_k}{X_k} - \frac{\sin X_k}{X_k} \mathcal B_k^\circ +	 \mathcal B_k^\sharp \frac{\cos X_k -1}{X_k^2}\mathcal B_k^\circ
\end{align*}
\end{lemma}

\begin{proof}
Due to the cutoff, $ \mathcal B_k^\sharp$ cannot create a $k$-connection if there is one already. That is, $$(\cB_k^\sharp)^2 =  (\cB_k^\circ)^2 = 0.$$ This allows us to expand the exponential explicitly. We find 
\begin{align*}
T_k &= e^{\cB_k} 
	= \sum_{m \geq 0} \frac{( \mathcal B_k^\sharp - \mathcal B_k^\circ)^m}{m!} 
	= \sum_{m\geq 0}  \frac{( \mathcal B_k^\sharp - \mathcal B_k^\circ)^{2m}}{(2m)!} + \frac{( \mathcal B_k^\sharp - \mathcal B_k^\circ)^{2m+1}}{(2m+1)!} \\
	&= \sum_{m\geq 1} \frac{(-1)^m}{(2m)!} \mathcal B_k^\sharp \left(\mathcal B_k^\circ \mathcal B_k^\sharp \right)^{m-1} \mathcal B_k^\circ + \sum_{m\geq 0} \frac{(-1)^m}{(2m!)} \left(\mathcal B_k^\circ \mathcal B_k^\sharp \right)^{m} \\
	&\quad  + \sum_{m\geq 0} \frac{(-1)^m}{(2m+1)!} \mathcal B_k^\sharp \left(\mathcal B_k^\circ \mathcal B_k^\sharp \right)^m -  \sum_{m\geq 0} \frac{(-1)^m}{(2m+1)!} \left(\mathcal B_k^\circ \mathcal B_k^\sharp \right)^m \mathcal B_k^\circ.
	\end{align*}
With Lemma \ref{lem:com_BB}, we find 
\begin{align*} 
T_k &= \mathcal B_k^\sharp \left( \sum_{m\geq 1} \frac{(-1)^m}{(2m)!} X_k^{2m-2}\right) \mathcal B_k^\circ + \sum_{m\geq 0} \frac{(-1)^m}{(2m!)} X_k^{2m} \\
	&\quad + \mathcal B_k^\sharp \sum_{m\geq 0} \frac{(-1)^m}{(2m+1)!}  X_k^{2m} -  \sum_{m\geq 0} \frac{(-1)^m}{(2m+1)!} X_k^{2m} \mathcal B_k^\circ  \\
	&= \mathcal B_k^\sharp \frac{\cos X_k -1}{X_k^2}\mathcal B_k^\circ + \cos X_k + \mathcal B_k^\sharp \frac{\sin X_k}{X_k} - \frac{\sin X_k}{X_k} \mathcal B_k^\circ \, .
\end{align*}
\end{proof}

For $k \in S$, we introduce the notation 
\begin{align*}
\Lambda_k := \Theta_{k}^{(1)} \mathds{1}_{\{\mathcal N_{(S\cup H)^c}=0\}} \mathds{1}_{\{\mathcal N_H \in 2 \mathbb{N}_0\}}.
\end{align*} 
We observe that 
\[ \prod_{k \in S} \Lambda_k \mathds{1}_{\{\cN_{S^c} = 0\}} = \mathds{1}_{\{\cN_{S^c} = 0\}}. \]
Thus, on the range of $\mathds{1}_{\{\cN_{S^c} = 0\}}$, we have 
\begin{equation}\label{eq:TTL} T_k = T_k \Lambda_k  =  \Big( \cos X_k + \cB_k^\sharp \frac{\sin X_k}{X_k} \Big) \Lambda_k =: \wt{T}_k. \end{equation} 
Moreover, since $\Lambda_k$ commutes with $T_{k'}$ for all $k' \not = k$, we obtain, on the range of $\mathds{1}_{\{\cN_{S^c} = 0\}}$, 
\begin{equation}\label{eq:TcwtT}  T_c = \prod_{k \in S} T_k = \prod_{k \in S} \wt{T}_k. \end{equation} 

In the next lemma, we control moments of the number operator with respect to the action of the cubic transformation.
\begin{lemma} \label{lem:Tc_N}
For $S' \subset S$, let $T_{c,S'}:= \prod_{p\in S'} T_p$. There is a constant $C>0$ such that, for all $S' \subset S$ and all $j \geq 1$, we have, on the range of $\mathds{1}_{\{\cN_{S^c} = 0\}}$, 
\begin{align} \label{eq:lem:Tc_N:1}
\pm  \left( T_{c,S'}^* \mathcal N_S^j T_{c,S'} - \mathcal N_S^j\right)  
&\leq C j N^{-2+3\kappa + \varepsilon} \mathcal N_S^{j},
\\
 T_{c,S'}^* \mathcal N_S^j T_{c,S'} &\leq \cN_S^j \label{eq:lem:Tc_N:2}.
\end{align}
Moreover, for all $j \geq 1$, there is a constant $C_j>0$ such that, for all $S' \subset S$, we have, on the range of $\mathds{1}_{\{\cN_{S^c} = 0\}}$,   
\begin{align}
 T_{c,S'}^* \mathcal N_H^j T_{c,S'}	&\leq C_j  N^{-2+3\kappa+\vep} \cN_S \left(N^{-2+3\kappa+\vep} \cN_S+1\right)^{j-1}. \label{eq:lem:Tc_N:3}
\end{align}
\end{lemma}


\begin{proof}
We start by proving (\ref{eq:lem:Tc_N:1}) and (\ref{eq:lem:Tc_N:2}). First of all, we observe that $[ \cB_k , \cN_S + \cN_H/2] = 0$, for all $k \in S$. This implies that 
\begin{equation}\label{eq:TN-inv} T_{c,S'}^* \big( \cN_S + \cN_H/2 \big)^j T_{c,S'} = \big( \cN_S + \cN_H/2 \big)^j  \end{equation} 
for all $j\in \bN$ and all $S' \subset S$. We immediately obtain
\begin{equation} \label{eq:TS'-up} T_{c,S'}^* \cN_S^j T_{c,S'} \leq  \big( \cN_S + \cN_H / 2 \big)^j = \cN_S^j \end{equation} 
on the range of $\mathds{1}_{\{\cN_{S^c} = 0\}}$. To prove the lower bound we choose $k$ as the ``first'' element in $S' \subset S$ (according to the order used to define $T_c$) and we compute
\begin{align} \label{eq:TcN^k_calc}
\T_k^* \mathcal N_S^j \T_k  
	&= \Lambda_k \Big( \cos X_{k}  + \frac{\sin X_{k}}{X_{k}} \mathcal B_{k}^\circ \Big)  \mathcal N_S^j \Big( \cos X_{k}  + \mathcal B_{k}^\sharp \frac{\sin X_{k}}{X_{k}} \Big) \Lambda_k \nn \\
		&= \Lambda_k \Big(  \cos^2 X_{k} \mathcal N_S^j  + \frac{\sin X_{k}}{X_{k}} \mathcal B_{k}^\circ \mathcal N_S^j \mathcal B_{k}^\sharp \frac{\sin X_{k}}{X_{k}} \Big) \Lambda_k \nn \\
		&= \Lambda_k \Big(  \cos^2 X_{k} \mathcal N_S^j + \frac{\sin X_{k}}{X_{k}} (\mathcal N_S - 1)^j X_k^2 \frac{\sin X_{k}}{X_{k}} \Big)\Lambda_k \nn \\
		&= \Lambda_k \Big( \mathcal N_S^j + ((\mathcal N_S-1)^j - \mathcal N_S^j) \sin^2 X_{k} \Big)\Lambda_k.
\end{align}
On the range of $\Lambda_k$, we find (with $\ph_H$ the restriction of $\ph$ on $H$) 
\[ T^*_k  \cN_S^j T_k \geq \cN_S^j -  j 2 N \| \ph_{H} \|_2^2 \cN_k \cN_S^{j-1} \geq \cN_S^j - C j N^{-2 + 3\kappa + \eps} \cN_S^{j-1} \cN_k  \]
Conjugating with $T_{c,S' \backslash \{ k \}}$ and using (\ref{eq:TS'-up}) we obtain, on the range of $\mathds{1}_{\{\cN_{S^c} = 0\}}$, 
 \begin{equation*} \label{eq:TNS-low} 
 \begin{split} 
 T_{c,S'}^* \cN_S^j T_{c,S} &\geq T_{c,S' \backslash \{ k \}}^* \cN_S^j T_{c,S' \backslash \{ k \}} - C j N^{-2+3\kappa +\eps}  \cN_k^{1/2} T_{c,S' \backslash \{ k \}}^*  \cN_S^{j-1} T_{c,S' \backslash \{ k \}} \cN^{1/2}_k \\ &\geq T_{c,S' \backslash \{ k \}}^* \cN_S^j T_{c,S' \backslash \{ k \}} - C j N^{-2+3\kappa +\eps}  \cN_S^{j-1}  \cN_k.
\end{split}  \end{equation*} 
Iterating to cover all $k \in S'$, we obtain 
\[  T_{c,S'}^* \cN_S^j T_{c,S} \geq \cN_S^j - C j N^{-2+3\kappa + \eps} \cN^{j-1}_S \sum_{k \in S'} \cN_k \geq  \cN_S^j - C j N^{-2+3\kappa + \eps} \cN^{j}_S. \]
Let us now show (\ref{eq:lem:Tc_N:3}). Similarly as in (\ref{eq:TcN^k_calc}), we find  
\begin{align*} \label{eq:proof_Tc_N:3}
\T_k^* \cN_H^j \T_k = \Lambda_k \big (\cN_H^j + ((\mathcal N_H+2)^j - \mathcal N_H^j) \sin^2 X_{k} \big)\Lambda_k.
\end{align*}
Thus, there exists $C > 0$ depending on $j$ such that, on the range of $\Lambda_k$, 
\begin{equation*}\label{eq:TNH-up} \T_k^* \cN_H^j \T_k \leq \cN_H^j + C N^{-2+3\kappa + \eps} (\cN_H^{j-1} +1) \cN_k. \end{equation*} 
Conjugating with $T_{c,S' \backslash \{ k \}}$ we obtain, on the range of $\mathds{1}_{\{\cN_{S^c} = 0\}}$,  
\begin{equation}\label{eq:TNHT} T_{c,S'}^* \cN_H^j T_{c,S'} \leq T_{c,S' \backslash \{ k \}}^* \cN_H^j T_{c,S' \backslash \{ k \}} +  C N^{-2+3\kappa + \eps} \cN_k^{1/2} ( T_{c,S' \backslash \{ k \}}^*  \cN_H^{j-1}  T_{c,S' \backslash \{ k \}}  +1) \cN^{1/2}_k. \end{equation} 
Next, we claim that for all $j \in \bN$ there exists $C_j > 0$ independent of $S'$ such that, on the range of $\mathds{1}_{\{\cN_{S^c} = 0\}}$, 
\begin{equation}\label{eq:induc} T_{c,S'}^* \cN_H^j T_{c,S'} \leq C_j N^{-2+3\kappa+\eps} \cN_S  \big( N^{-2+3\kappa + \eps} \cN_S + 1 \big)^{j-1}. \end{equation} 
For $j =1$, the claim follows (\ref{eq:TNHT}), iterating to cover all $k \in S'$. Assuming (\ref{eq:induc}), we use (\ref{eq:TNHT}) to estimate 
\[ \begin{split} T_{c,S'}^* \cN_H^{j+1} T_{c,S'} 
&\leq T_{c,S'\backslash \{ k \}}^* \cN_H^{j+1} T_{c,S' \backslash \{ k \}} +  C N^{-2+3\kappa + \eps} \cN_k^{1/2} ( T_{c,S' \backslash \{ k \}}^*  \cN_H^{j}  T_{c,S' \backslash \{ k \}}  +1) \cN^{1/2}_k \\ 
&\leq T_{c,S' \backslash \{ k \}}^* \cN_H^{j+1} T_{c,S' \backslash \{ k \}} +  C_j C N^{-2+3\kappa + \eps}  (N^{-2+3\kappa+\eps} \cN_S +1)^{j}  \cN_k. \end{split} \]
Iterating over all $k \in S'$ we conclude that, on the range of $\mathds{1}_{\{\cN_{S^c} = 0\}}$, 
\[ \begin{split} 
T_{c,S'}^* \cN_H^{j+1} T_{c,S'} &\leq \cN_H^{j+1} +  C_{j+1} N^{-2+3\kappa + \eps}   (N^{-2+3\kappa+\eps} \cN_S +1)^j \sum_{k \in S'} \cN_k \\ &\leq  C_{j+1} N^{-2+3\kappa+\eps} \cN_S  (N^{-2+3\kappa+\eps} \cN_S +1)^{j}. \end{split} \] 
By induction, we obtain (\ref{eq:lem:Tc_N:3}).
\end{proof}

\subsection{Action on the off-diagonal quadratic term} 

Next, we control the action of $T_c$ on the quadratic off-diagonal term on the right-hand side of (\ref{eq:claimB1}). We introduce the  notation 
\begin{align*}
\widetilde{Q}_2^\sharp = 4\pi\mathfrak{a} N^\kappa  \sum_{p\in S} a_p^*a_{-p}^* = (\widetilde{Q}_2^\circ)^*.
\end{align*}
 
\begin{lemma} \label{lem:Tc_nondiag} On the range of $\mathds{1}_{\{\cN_{S^c} = 0\}}$, we have  
\begin{align*}
&\pm \Big[ T_c^*(\widetilde{Q}_2^\sharp + \widetilde{Q}_2^\circ) T_c - (\widetilde{Q}_2^\sharp + \widetilde{Q}_2^\circ) \Big] \lesssim N^{-2+4\kappa+\vep}  (\cN_S + N^{3\kappa/2+3\vep}) .
\end{align*}
\end{lemma}
\begin{proof}
Let $p, k\in S$. If $p \notin \{\pm k\}$, then clearly 
\begin{align} \label{eq:nondiag_commute}
\T_k^* a_p^*a_{-p}^* \T_k = \Lambda_k a_p^*a_{-p}^* \Lambda_k.
\end{align}
Let us now consider the case $p \in \{\pm k\}$, say $p =k$. Since $\Theta^{(1)}_k$ only depends on operators $\cN_t$, with $t \in H$, we have $[a_p^*a_{-p}^*, \Theta^{(1)}_k] =0$  and therefore $\Lambda_k a_p^*a_{-p}^* \mathcal B_c^\sharp = \mathcal B_c^\circ a_p^*a_{-p}^* \Lambda_k = 0$. We find
\begin{align*}
\T_k^* a_k^*a_{-k}^* \T_k &= \Lambda_k \Big(\cos X_k a_k^*a_{-k}^* \cos X_k + \frac{\sin X_k}{X_k} \cB_k^\circ a_k^*a_{-k}^* \cB_k^* \frac{\sin X_k}{X_k} \Big) \Lambda_k
\\
&= \Lambda_k \Big(\cos X_k a_k^*a_{-k}^* \cos X_k + \frac{\sin X_k}{X_k} a_k^*a_{-k}^* X_k \sin X_k \Big) \Lambda_k,
\end{align*}
where we used that $[\cB_k^\circ, a_k^*a_{-k}^*] =0$. 
With the formula (\ref{eq:exp_Xk2}) for $X_k^2$ and since $\mathcal N_k a_k^*a_{-k}^* = a_k^*a_{-k}^* (\mathcal N_k+1)$, we obtain 
\begin{align*}
X_k^2 a_k^*a_{-k}^* = a_k^*a_{-k}^* Y_k^2,
\end{align*}
with $Y_k$ being defined as the square root of the positive operator 
\begin{align*}
Y_k^2 &:=  X_k^2 + \sum_{t\in H_k} N\vphi_t(\vphi_t + \vphi_{t+k}) \Theta_{k,t}.
\end{align*}
Since (with $\Theta_{k,t} = \Theta_{k,-t-k}$) 
\begin{align} \label{eq:positivity}
Y_k^2 - X_k^2  =\sum_{t\in H_k} N \vphi_t(\vphi_t + \vphi_{t+k})\Theta_{k,t} = \frac{N}{2} \sum_{t\in H_k}(\vphi_t+\vphi_{t+k})^2 \Theta_{k,t}
\end{align}
we conclude that $Y_k^2 \geq X_k^2$. Since $X_k$ and $Y_k$ commute, we arrive at 
\begin{equation}\label{eq:TaaT}
\begin{split} 
\T_k^* a_k^*a_{-k}^* \T_k 
&= \Lambda_k a_k^*a_{-k}^* \Big( \cos Y_k \cos X_k + X_k \sin X_k \frac{\sin Y_k}{Y_k} \Big) \Lambda_k
\\
&= \Lambda_k a_k^*a_{-k}^* \Big( \cos(Y_k - X_k) + (X_k - Y_k) \sin X_k \frac{\sin Y_k}{Y_k} \Big) \Lambda_k
\\
&=: \Lambda_k a_k^*a_{-k}^* \big(1 + R_k \big) \Lambda_k.
\end{split} 
\end{equation} 
From $| \cos(y-x) -1 +(x-y)\sin(x) \sin(y) / y| \leq C (y^2-x^2)$ for all $0\leq x\leq y$ and from Lemma~\ref{lem:vphi_prop}, we obtain
$$
\pm R_k \lesssim N \|\varphi_H\|_2^2 \lesssim N^{-2+3\kappa+\vep}.
$$
With (\ref{eq:TaaT}), we conclude that  
\begin{align*}
& \pm \Lambda_k \left(T_k^* a_k^*a_{-k}^* T_k -a_k^*a_{-k}^*\right)\Lambda_k + \hc
\lesssim N^{-2+3\kappa+\vep} (\cN_k + \cN_{-k} + 1).
\end{align*}
Conjugation with $T_{-k}$ can be handled similarly. Since moreover the error term is invariant w.r.t. $T_p$, for $p \in S \backslash \{ \pm k \}$, we find 
\[ \pm \big( T^*_c a_k^* a_{-k}^* T_c - a_k^* a_{-k}^* \big) + \text{h.c.}  \lesssim N^{-2+3\kappa+\eps} (\cN_k + \cN_{-k} + 1) \]
on the range of $\mathds{1}_{\{\cN_{S^c} = 0\}}$. The claim follows by summing over $k \in S$ (recall that $|S| \lesssim N^{3\kappa/2+3\vep}$). 
\end{proof}

\subsection{Action on the kinetic energy operator} 

In this subsection, we control the action of a single cubic transformation $T_k$, for a fixed $k \in S$, on the kinetic energy operator. In the next subsections, we show similar statements for the action of $T_k$ on the cubic term $Q_3$ and on the quartic term $Q_4$, appearing in Lemma \ref{lem:TLc_main} on the right hand side of \eqref{eq:claimB1}. Eventually, we will obtain the action of the full transformation $T_c$ by iteration.
\begin{lemma} \label{lem:Tc_kin}
For $k \in S$ we have, on the range of $\Lambda_k$, 
\begin{equation} \label{eq:kine}
\begin{split} 
&\pm \Big[ T_k^* \sum_{p\in \Lambda^*} p^2 a_p^*a_p T_k - \sum_{p\in \Lambda^*} p^2 a_p^*a_p - 2 \sum_{r\in H_k} N r^2 \vphi_r (\vphi_r+\vphi_{k+r})  a_k^*a_k \Theta_{k,r} \Big] \lesssim  N^{-1+5\kappa/2+3\vep/2}\cN_k^2.
\end{split} \end{equation}
\end{lemma}

\begin{proof}
Using that $\mathcal B_{k}^\circ \Lambda_k = 0$ and that $[X_k^2, \Lambda_k] =0$, we obtain
\begin{align*}
\Lambda_k  T_k^* \sum_{q\in \Lambda^*} q^2 a_q^*a_q T_k \Lambda_k &= \Lambda_k \Big( \cos X_{k}  + \frac{\sin X_{k}}{X_{k}} \mathcal B_{k}^\circ \Big)  \sum_{q\in \Lambda^*} q^2 a_q^*a_q \Big( \cos X_{k}  + \mathcal B_{k}^\sharp \frac{\sin X_{k}}{X_{k}} \Big)\Lambda_k \\
	&= \Lambda_k \Big( \cos^2 X_{k} \sum_{p\in \Lambda^*} p^2 a_p^*a_p + \frac{\sin X_{k}}{X_{k}} \mathcal B_{k}^\circ \sum_{p\in \Lambda^*} p^2 a_p^*a_p \mathcal B_{k}^\sharp \frac{\sin X_{k}}{X_{k}} \Big)\Lambda_k \\
	&= \Lambda_k \Big(  \sum_{p\in \Lambda^*} p^2 a_p^*a_p  + \frac{\sin X_{k}}{X_{k}} \big[ \mathcal B_{k}^\circ,\sum_{p\in \Lambda^*} p^2 a_p^*a_p \big] \mathcal B_{k}^\sharp \frac{\sin X_{k}}{X_{k}} \Big)\Lambda_k.
\end{align*}
Let us now compute 
\begin{align*}
\big[ \mathcal B_{k}^\circ,\sum_{p\in \Lambda^*} p^2 a_p^*a_p \big] \mathcal B_{k}^\sharp
&= \sum_{r\in H_k} N^{1/2}\vphi_r \left((r^2+(k+r)^2-k^2\right) \Theta_{k,r} a_k^* a_{-r}a_{k+r} \cB_k^\sharp
\\
	&= 2 \sum_{r\in H_k} N \vphi_r (\vphi_r+\vphi_{k+r}) \left(r^2+k \cdot r \right)  a_k^* a_k\Theta_{k,r},
\end{align*}
where we used that, similarly as in the proof of Lemma \ref{lem:com_BB},
\begin{align*}
a_{-r}a_{k+r} \cB_k^\sharp = (\varphi_r + \varphi_{r+k}) a_k\Theta_{k,r}.
\end{align*}
We therefore obtain 
\begin{align*}
& \Lambda_k \Big( T_k^* \sum_{p\in \Lambda^*}p^2 a_p^*a_p  T_k - \sum_{p\in \Lambda^*}p^2 a_p^*a_p - 2 \sum_{r\in H_k} N r^2 \vphi_r (\vphi_r+\vphi_{k+r})  a_k^* a_k\Theta_{k,r} \Big) \Lambda_k
\\
	&= \Lambda_k \Big( 2 \sum_{r\in H_k} N  \vphi_r (\vphi_r+\vphi_{k+r}) a_k^* a_k\Theta_{k,r}
\Big[ r^2 \Big (\frac{\sin^2 X_{k}}{X^2_{k}} -1\Big) + r \cdot k \, \frac{\sin^2 X_{k}}{X^2_{k}} \Big] \Big)\Lambda_k.
\end{align*}
From \Cref{lem:vphi_prop} we find that
\begin{align*}
\sum_{r\in H_k} N r^2 |\vphi_r(\vphi_r + \vphi_{k+r})| &\lesssim N^\kappa \| \varphi  \|_1 \lesssim N^{\kappa}
\end{align*}
and that 
\begin{align*}
\sum_{r\in H_k} N |r||k| |\vphi_r||\vphi_r + \vphi_{k+r}| \lesssim N |k| \|p \varphi\|_2 \|\varphi_{H} \|_2 \lesssim N^{-1+5\kappa/2 + 3\vep/2}.
\end{align*}
since $|k|\leq N^{\kappa/2 + \varepsilon}$. Moreover, \eqref{eq:Xp_bound} implies that 
\begin{align*}
\Big( 1 - \frac{\sin^2 X_{k}}{X^2_{k}} \Big) \lesssim N^{-2 + 3\kappa + \varepsilon} \mathcal N_k.
\end{align*}
Since $N^\kappa N^{-2+3\kappa+\eps} = N^{-2 + 4\kappa + \eps} \leq N^{-1+5\kappa/2+3\eps/2}$ for all $\kappa < 2/3$, we obtain (\ref{eq:kine}).
\end{proof}

\subsection{Action on the cubic term}

Next, we proceed with the conjugation of the cubic term $Q_3$ appearing on the right-hand side of (\ref{eq:claimB1}). On the range of the projection $\Xi = \mathds{1}_{\{\mathcal N_{(H\cup S)^c}=0\}} \mathds{1}_{\{\mathcal N_H \in 2\mathbb{N}_0\}}$, we can decompose 
\begin{equation}\label{eq:Q3-deco} Q_3 = \sqrt{\frac{N_0}{N}} \left(Q_3^H + Q_3^M + Q_3^S\right) \end{equation} 
with
\begin{align*}
Q_3^H &= \sum_{r\in H_p,p\in S} N^{1/2} \hat{V}_N (r) a^*_{-r} a^*_{r+p} a_p + \hc
\\
	Q_3^M &=  \sum_{\substack{r,p\in H: \\ r+p \in S}} N^{1/2} \hat{V}_N (r) a^*_{-r} a^*_{r+p} a_p  + \sum_{p\in H_r,r\in S} N^{1/2} \hat{V}_N (r) a^*_{-r} a^*_{r+p} a_p + \hc
\\
	Q_3^S &= \sum_{\substack{r,p \in S : \\ r+p \in S}} N^{1/2} \hat{V}_N (r) a^*_{-r} a^*_{r+p} a_p + \hc
\end{align*}
Note that since $0\leq N - N_0 \lesssim N^{3\kappa/2}$, we have $\sqrt{N_0/N} = 1 + \mathcal O(N^{-1+3\kappa/2})$.

In the next lemma we conjugate the operator $Q_3^H$. 
\begin{lemma} \label{lem:Q3_Tp}
Let $k \in S$. On the range of $\Lambda_k$, we have 
\begin{align}
&\pm \Big( T_k^* Q_3^H T_k - Q_3^H - 2\sum_{r\in H_k} N \hat{V}_N(r) (\vphi_r + \vphi_{r+k}) \cN_k \Theta_{k,r}  \Big)  \nn \\
	& \lesssim N^{-2+4\kappa +\vep} \cN_k^2 + N^{-5+15\kappa/2+7\vep}
\Big(N^{2-2\kappa}\sum_{r,q\in H} \1_S(r+q) a_{r}^* a_{q}^* a_{q}a_{r} + N^{\kappa} \cN_S \Big)(\cN_k+1)
\label{eq:Q3H} 
\end{align}
\end{lemma}

\begin{proof} Let us write $Q_3^H = \sum_{p\in S}Q_{3,p}^{H,\sharp} + Q_{3,p}^{H,\circ}$ with $Q_{3,p}^{H,\circ} = (Q_{3,p}^{H,\sharp})^*$ and
$$
Q_{3,p}^{H,\sharp}:= \sum_{r\in H_p} N^{1/2}\hat{V}_N(r) a_{-r}^*a_{r+p}^*a_p.
$$
Let $k,p \in S$. With \eqref{eq:TTL} and using that $\Lambda_k Q_{3,p}^{H,\sharp} \cB_k^\sharp = 0$ (since $Q_{3,p}^{H,\sharp}$ cannot annihilate a $k$-connection), we find
\begin{align}\label{eq:Q3H_expansion}
\Lambda_k T_k^* Q_{3,p}^{H,\sharp} T_k \Lambda_k
= \; &\Lambda_k \cos X_k Q_{3,p}^{H,\sharp} \cos X_k \Lambda_k + \Lambda_k \frac{\sin X_k}{X_k} \cB_k^\circ Q_{3,p}^{H,\sharp} \cB_k^\sharp \frac{\sin X_k}{X_k} \Lambda_k \nn
\\
& + \Lambda_k \frac{\sin X_k}{X_k} \cB_k^\circ Q_{3,p}^{H,\sharp} \cos X_k \Lambda_k.
\end{align}
We compute
\begin{align} \label{eq:Q3H_Tp_comm}
\cB_k^\circ Q_{3,p}^{H,\sharp} &= \sum_{t\in H_k, r\in H_p} N \vphi_t \hat{V}_N(r) \Theta_{k,t} a_k^* a_{k+t} a_{-t} a_{-r}^*a_{r+p}^* a_p \nn
\\
&= \sum_{t\in H_k, r\in H_p} N \vphi_t \hat{V}_N(r) \Theta_{k,t} \Big( a_{-r}^*a_{r+p}^* a_{k+t} a_{-t} + \delta_{k,p} (\delta_{t,r} + \delta_{t,-(p+r)}) \nn \\ &\hspace{2cm} +  a_{p+r}^* a_{r+k} (\delta_{t,r} + \delta_{t,-r-k} )  + a^*_{-r} a_{-p-r+k} (\delta_{t, -p-r} + \delta_{t, p+r-k})  \Big) \, a_k^* a_p  \nn
\\
&= \sum_{t\in H_k, r\in H_p} N \vphi_t \hat{V}_N(r) \Theta_{k,t} \Big( a_{-r}^*a_{r+p}^* a_{k+t} a_{-t} + \delta_{k,p} (\delta_{t,r} + \delta_{t,-(p+r)}) \Big) a_k^* a_p,
\end{align}
where, in the last step, we used the definition of $\Theta_{k,t}$ (which implies for example that $\Theta_{k,r} a_{p+r}^* = 0$ and similarly for the other terms on the third line). 

Next, we distinguish the cases $k=p$ and $k \not = p$. Let us first assume that $k=p$.
Since $\Lambda_k Q_{3,k}^{H,\sharp} = 0$, the first term on the right-hand side of \eqref{eq:Q3H_expansion} vanishes. Using (\ref{eq:Q3H_Tp_comm}) and $\Theta_{k,t} a_{-r}^*a_{r+k}^* = 0$, we conclude that also the second term on the right-hand side of \eqref{eq:Q3H_expansion} is zero (because $\Lambda_k$ can be moved through $\cB_k^\circ Q_{3,p}^{H,\sharp}$ and $\Lambda_k \cB_k^\sharp = 0$). Thus, we have 
\begin{align*}
\Lambda_k T_k^* Q_{3,k}^{H,\sharp} T_k \Lambda_k =  \Lambda_k \frac{\sin X_k\cos X_k}{X_k} \sum_{r\in H_k} N \hat{V}_N(r) (\vphi_r + \vphi_{r+k}) a_k^* a_k \Theta_{k,r}  \Lambda_k.
\end{align*}
With the elementary inequality 
\[  0 \leq 1-\frac{\sin x\cos x}{x}  \leq \frac{2}{3} x^2 , \] using Lemma \ref{lem:vphi_prop} to estimate \[ \Big|\sum_{r\in H_k} N \hat{V}_N(r) (\vphi_r + \vphi_{r+k}) \Big| \lesssim N^\kappa \] 
and applying the bound (\ref{eq:Xp_bound}), we obtain 
\begin{align} \label{eq:Q3Hp}
\pm \Lambda_k \Big( T_k^* Q_{3,k}^{H,\sharp} T_k - \sum_{r\in H_k} N \hat{V}_N(r) (\vphi_r + \vphi_{r+k}) a_k^* a_k \Theta_{k,r} \Big) \Lambda_k 
\lesssim N^{-2+4\kappa + \vep} \cN_k^2.
\end{align}
We now consider the case $p \neq k$. We will prove that
\begin{align} \label{eq:Q3Hq}
&\pm \Lambda_k \Big(T_k^* \sum_{p\in S\setminus\{k\}} Q_{3,p}^{H,\sharp} T_k - \sum_{p\in S\setminus\{k\}} Q_{3,p}^{H,\sharp} \Big) \Lambda_k + \hc \nn
\\
&\qquad \lesssim N^{-5+15\kappa/2+7\vep}
\Big(N^{2-2\kappa}\sum_{r,q\in H} \1_S(r+q) a_{r}^* a_{q}^* a_{q}a_{r} + N^{\kappa} \cN_S \Big) (\cN_k+1).
\end{align}
Together with \eqref{eq:Q3Hp} and $\Lambda_k \big(Q_{3,k}^{H,\sharp} + Q_{3,k}^{H,\circ} \big)\Lambda_k  = 0$ this implies \eqref{eq:Q3H}. 

To show (\ref{eq:Q3Hq}), we observe that, from (\ref{eq:Q3H_Tp_comm}), $ \cB_k^\circ Q_{3,p}^{H,\sharp} \Lambda_k = 0$ since $a_{k+t} a_{-t} \Lambda_k = 0$ for $t\in H_k$. Therefore, the last term in  \eqref{eq:Q3H_expansion} vanishes. Using \eqref{eq:Q3H_Tp_comm} to rewrite the second term on the right-hand side of (\ref{eq:Q3H_expansion}), we find 
\begin{align} \label{eq:Q3Hq_expanded}
\Lambda_k T_k^* Q_{3,p}^{H,\sharp} T_k \Lambda_k =& \; \Lambda_k \cos X_k Q_{3,p}^{H,\sharp} \cos X_k \Lambda_k \nn \\
	& + \Lambda_k \frac{\sin X_k}{X_k} \sum_{t\in H_k, r\in H_p} N \vphi_t \hat{V}_N(r) \Theta_{k,t}^{(2)} a_{-r}^*a_{r+p}^* a_{k+t} a_{-t} a_k^* a_p \cB_k^\sharp \frac{\sin X_k}{X_k} \Lambda_k 
\end{align}
In the second term, we move the cutoff $\Theta^{(2)}_{k,t}$ to the right to reconstruct the operator $\cB_k^\circ$. To this end, we use the identity 
%
\begin{align} \label{eq:comm_a_Theta}
\Theta_{k,t}^{(2)} a_r^* =a_r^* \Theta_{k,t}^{(2)} \left(1- \mathds{1}_{(t+S)\cup (-(t+k)+S)}(r) \right) \quad \forall k\in S,t\in H_k, r \in \Lambda^*
\end{align}
We find 
\begin{align*}
\sum_{t\in H_k, r\in H_p} & N \vphi_t \hat{V}_N(r) \Theta_{k,t}^{(2)} a_{-r}^*a_{r+p}^* a_{k+t} a_{-t} a_k^* a_p \cB_k^\sharp
\\
	&= \sum_{t\in H_k, r\in H_p} N \vphi_t \hat{V}_N(r) a_{-r}^*a_{r+p}^*  a_p \Theta_{k,t}^{(2)} a_{k+t} a_{-t} a_k^*  \cB_k^\sharp \big( 1 - \mathds{1}(-r,r+p \simS -(t+k),t) \big)
\end{align*}
with the notation 
$$
\mathds{1}(-r,r+p \simS -(t+k),t) :=
\begin{cases}
1 & \text{if } \{-r,r+p\} \cap \big((t+S)\cup (-(t+k)+S)\big) \neq \emptyset \\
0 & \text{otherwise.}
\end{cases}
$$
Proceeding as in  the proof of \Cref{lem:com_BB}, we find 
\[ \begin{split} 
\sum_{t\in H_k, r\in H_p} & N \vphi_t \hat{V}_N(r) \Theta_{k,t}^{(2)} a_{-r}^*a_{r+p}^* a_{k+t} a_{-t} a_k^* a_p \cB_k^\sharp =
\sum_{r \in H_p} N^{1/2} \widehat{V}_N (r) a_{-r}^* a_{r+p}^* a_p (X_k^{(-r, r+p)} )^2
\end{split} \]
where we defined
\begin{align} \label{eq:def_Xk^-r}
\big(X_k^{(-r,r+p)}\big)^2 &= \sum_{t\in H_k}\left( 1 - \mathds{1}(-r,r+p \simS -(t+k),t) \right)\Theta_{k,t} N \vphi_t(\vphi_t + \vphi_{t+k}) a_k^*a_k 
\end{align}
We can decompose $\big(X_k^{(-r,r+p)}\big)^2 = X_k^2 - \delta_k^{(-r,r+p)}$ with 
\begin{equation} \label{eq:delta_bound}
\delta_k^{(-r,r+p)} = \sum_{t\in H_k} \mathds{1}(-r,r+p \simS -(t+k),t)  \Theta_{k,t} N \vphi_t(\vphi_t + \vphi_{t+k}) \cN_k 
\end{equation} 
Notice that $(X^{(-r, r+p)}_k)^2$ and $\delta_k^{(-r, r+p)}$ are both non-negative operators (this can be shown similarly as in \eqref{eq:positivity}). Moreover, with \Cref{lem:vphi_prop}, we find
\[ \delta_k^{(-r,r+p)} \lesssim N |S|  \|\vphi_H\|_\infty^2 \cN_k
\lesssim N^{-5+15\kappa/2+7\vep} \cN_k.
\] 
Hence, after summing over $p \in S \backslash \{ k \}$,  the contribution of $\delta_k^{(-r, r+p)}$ can be bounded with Cauchy--Schwarz on the range of $\Lambda_k$ by 
%
\begin{align} \label{eq:delta_Q3H_estimate}
&\pm \sum_{p \in S\setminus\{k\}, r\in H_p} N^{1/2} \hat{V}_N(r) a_{-r}^*a_{r+p}^* a_p \delta_k^{(-r,r+p)} +\hc \nn
\\
&\qquad \lesssim N^{-5+15\kappa/2+7\vep}
\Big( N^{2-2\kappa}\sum_{r,q\in H} \1_S(r+q) a_{r}^* a_{q}^* a_{q}a_{r} + N^{\kappa} \cN_S \Big)(\cN_k+1). 
\end{align}
Therefore
\begin{equation} \begin{split}  \label{eq:Q3H_cos_2ndstep}
\pm  \sum_{p \in S\setminus\{k\}} \Lambda_k  \Big[ &T_k^* Q_{3,p}^{H,\sharp} T_k - \cos X_k Q_{3,p}^{H,\sharp} \cos X_k - \frac{\sin X_k}{X_k} Q_{3,p}^{H,\sharp} X_k \sin X_k + \text{h.c.} \Big] \Lambda_k  \\ &\lesssim 
N^{-5+15\kappa/2+7\vep}
\Big(N^{2-2\kappa}\sum_{r,q\in H} \1_S(r+q) a_{r}^* a_{q}^* a_{q}a_{r} + N^{\kappa} \cN_S \Big) (\cN_k+1).
\end{split} \end{equation} 
Arguing as in \eqref{eq:comm_a_Theta}, we also obtain 
\begin{align} \label{eq:comm_Xk^-r}
\Lambda_k
X_k^2 a_{-r}^* a_{r+p}^ *a_p \Lambda_k = \Lambda_k a_{-r}^* a_{r+p}^* a_p \big(X_k^{(-r,r+p)}\big)^2 \Lambda_k.
\end{align}
Since moreover $X_k$ and $X_k^{(-r,r+p)}$ commute, we find
\begin{align*}
&\Lambda_k \cos X_k a_{-r}^*a_{r+p}^*a_p \cos X_k \Lambda_k + \Lambda_k \frac{\sin X_k}{X_k} a_{-r}^*a_{r+p}^*a_p X_k\sin X_k \Lambda_k 
\\
&= \Lambda_k a_{-r}^*a_{r+p}^*a_p \left(\cos X_k^{(-r,r+p)} \cos X_k + \frac{\sin X_k^{(-r,r+p)}}{X_k^{(-r,r+p)}} X_k\sin X_k \right) \Lambda_k
\\
&= \Lambda_k a_{-r}^*a_{r+p}^*a_p \left(\cos (X_k -X_k^{(-r,r+p)}) + (X_k - X_k^{(-r,r+p)}) \sin X_k  \frac{\sin X_k^{(-r,r+p)}}{X_k^{(-r,r+p)}} \right) \Lambda_k
\\
&=: \Lambda_k a_{-r}^*a_{r+p}^*a_p \left(1 + R_k^{(-r,r+p)} \right) \Lambda_k
\end{align*}
With $|\cos(x-y) -1 +(x-y)\sin(x) \sin(y) / y| \leq C (x^2-y^2)$ for all $0\leq y\leq x$, we conclude that 
$$
\pm R_k^{(-r,r+p)} \leq C \delta_k^{(-r,r+p)}.
$$
Thus, proceeding as in \eqref{eq:delta_Q3H_estimate}, we arrive, on the range of $\Lambda_k$, at 
\begin{align*}
& \pm \sum_{p \in S \setminus \{k\}, r \in H_p} N^{1/2}\hat{V}_N(r) a_{-r}^*a_{r+p}^*a_p R_k^{(-r,r+p)} +\hc\\
&\qquad \lesssim N^{-5+15\kappa/2+7\vep}
\left(N^{2-2\kappa}\sum_{r,q\in H} \1_S(r+q) a_{r}^* a_{q}^* a_{q}a_{r} + N^{\kappa} \cN_S \right)(\cN_k+1).
\end{align*}
Inserting this bound into \eqref{eq:Q3H_cos_2ndstep}, we arrive at \eqref{eq:Q3Hq}. 
\end{proof}

As for the observables $Q_3^M, Q_3^S$, we employ the parity operators $\mathbb{P}_k$ defined in (\ref{eq:parity}) and the assumption $[ \mathbb{P}_k, \Gamma] =0$ for all $k\in S$ to show that their expectation vanishes in the state described by the density matrix $\Gamma$. Here, we will make use of the following 
property stating that, for $p \in H$, the density matrix $T_c \Gamma T_c^*$ can only have either zero or one particle with momentum $p$ and that, in the second case, there must be exactly one particle with momentum in the set $-p+S$, forming an $S$-connection with the particle with momentum $p$. 
\begin{lemma} \label{lm:mono} 
Let $\Gamma = \mathds{1}_{\{\cN_{S^c} = 0\}} \Gamma \mathds{1}_{\{\cN_{S^c} = 0\}}$. For $p\in H$, define 
\begin{equation}\label{eq:chip}
\begin{split}  \chi_p &=  \mathds{1}_{\{\cN_p = 0\}} + \mathds{1}_{\{\cN_p = 1\}}\sum_{x \in S} \mathds{1}_{\{\cN_{-p+x} = 1\}} \prod_{ y\in S\setminus\{x\}} \mathds{1}_{\{\cN_{-p+y} = 0 \}} \\
\wt{\chi}_p &= \mathds{1}_{\{\cN_p = 0\}} + \mathds{1}_{\{\cN_p = 1\}}\sum_{x \in S} \mathds{1}_{\{\cN_{-p+x} = 1\}}
\end{split} 
 \end{equation}
Then we have 
\begin{equation}\label{eq:cTGTc}  T_c \Gamma T_c^* = \chi_p T_c \Gamma T_c^* \chi_p = \wt{\chi}_p T_c \Gamma T_c^* \wt{\chi}_p.  \end{equation} 
\end{lemma}
\begin{proof}
Let $p\in H$. Note that $\Gamma = \chi_p \Gamma \chi_p$. We now show that $[a_{-r}^*a_{r+k}^*a_k\Theta_{k,r}, \chi_p] = 0$ for all $k\in S, r\in H_k$, which implies $[T_c,\chi_p] = 0$ and thus the first identity in \eqref{eq:cTGTc}.
 
We start by observing that from the definition of $\Theta_{k,r}$ in (\ref{eq:Thetakr}) we have
\begin{equation}  \label{eq:monogamy_creation}
a_{-r}^*a_{r+k}^*a_k \Theta_{k,r} = \mathds{1}_{\{\cN_{-r} = 1\}} \mathds{1}_{\{\cN_{r+k} = 1\}} \prod_{ y\in S\setminus\{k\}} \mathds{1}_{\{\cN_{r+y} = 0 \}} a_{-r}^*a_{r+k}^*a_k \Theta_{k,r}.
\end{equation}
This implies
\begin{align}
\big[ a_{-r}^* a_{r+k}^* &a_k \Theta_{k,r}, \mathds{1}_{\{\cN_p = 0\}} \big] \nn \\ &= a_{-r}^*a_{r+k}^*a_k \Theta_{k,r} \mathds{1}_{\{\cN_p = 0\}} \big(\delta_{-r,p} + \delta_{r+k,p} \big)
\nn \\
& = \mathds{1}_{\{\cN_p = 1\}} \mathds{1}_{\{\cN_{-p+k} = 1\}} \prod_{ y\in S\setminus\{k\}} \mathds{1}_{\{\cN_{-p+y} = 0 \}} a_{p}^*a_{-p+k}^*a_k \Theta_{k,-p} \left(\delta_{-r,p} + \delta_{r+k,p} \right). \label{eq:monog2}
\end{align}
To compute the commutator with the second summand in the definition (\ref{eq:chip}) of the projection $\chi_p$, we distinguish four cases. 

\medskip

\emph{Case 1:} $p = -r$ or $p = r+ k$. This implies that $p \in -(r+k) + S$ or $p \in r + S$, which by definition of $\Theta_{k,r}^{(2)}$ gives $\Theta_{k,r} \mathds{1}_{\{\cN_p = 1\}} = 0$. Thus
\begin{align*}
\Big[ a_{-r}^* a_{r+k}^*a_k \Theta_{k,r}, \mathds{1}_{\{\cN_p = 1\}} &\sum_{x \in S} \mathds{1}_{\{\cN_{-p+x} = 1\}}  \prod_{ y\in S\setminus\{x\}} \mathds{1}_{\{\cN_{-p+y} = 0\}} \Big]
\\
&= - \mathds{1}_{\{\cN_p = 1\}}\sum_{x \in S} \mathds{1}_{\{\cN_{-p+x} = 1\}} \prod_{ y\in S\setminus\{x\}} \mathds{1}_{\{\cN_{-p+y} = 0\}} a_{p}^* a_{-p+k}^*a_k \Theta_{k,-p}
\\
&= - \mathds{1}_{\{\cN_p = 1\}} \mathds{1}_{\{\cN_{-p+k} = 1\}}\prod_{ y\in S\setminus\{k\}}  \mathds{1}_{\{\cN_{-p+y} = 0 \}} a_p^* a_{-p+k}^*a_k \Theta_{k,-p},
\end{align*}
where the last equality follows from \eqref{eq:monogamy_creation}. With (\ref{eq:monog2}), this proves $[ a_{-r}^* a_{r+k}^*a_k \Theta_{k,r}, \chi_p] = 0$. 

\medskip

\emph{Case 2:} $p = r+z$ for some $z \in S \setminus\{k\}$. Then $p \notin \{-r,r+k\}$ and also $\Theta_{k,r} \mathds{1}_{\{\cN_p = 1\}} = 0$, which implies
\begin{align*}
&\Big[a_{-r}^* a_{r+k}^*a_k \Theta_{k,r}, \mathds{1}_{\{\cN_p = 1\}}\sum_{x \in S} \mathds{1}_{\{\cN_{-p+x} = 1\}} \prod_{ y\in S\setminus\{x\}} \mathds{1}_{\{\cN_{-p+y} = 0\}} \Big] = 0
\end{align*}
Since, in this case, we find from (\ref{eq:monog2}) that $[ a_{-r}^* a_{r+k}^*a_k \Theta_{k,r} , \mathds{1}_{\{\cN_p =0\}} ] = 0$, we conclude again that $[ a_{-r}^* a_{r+k}^*a_k \Theta_{k,r}, \chi_p ]= 0$. 

\medskip

\emph{Case 3:} $p = -(r+k) + z$ for some $z \in S \setminus\{k\}$. Then $p \notin \{-r,r+k\}$ and $\Theta_{k,r} \mathds{1}_{\{\cN_p = 1\}} = 0$. Again, we find $[ a_{-r}^* a_{r+k}^*a_k \Theta_{k,r}, \chi_p ]= 0$. 

\medskip

\emph{Case 4:} If none of the conditions apply, terms commute and $[ a_{-r}^* a_{r+k}^*a_k \Theta_{k,r}, \chi_p ]= 0$.

\medskip

This shows that $\chi_p T_c \Gamma T_c^* \chi_p = T_c \Gamma T_c^*$. The second identity in (\ref{eq:cTGTc}) can be shown similarly.
\end{proof} 

With Lemma \ref{lm:mono}, we can now show that the expectation of $Q_3^S$ and $Q_3^M$ vanish in the state $T_c^* \Gamma T_c$. 

\begin{lemma} \label{lm:Q3SM}
Let $\Gamma$ be a density matrix on $\cF (\Lambda)$ satisfying $\Gamma = \mathds{1}_{\{\cN_{S^c} = 0\}} \, \Gamma \mathds{1}_{\{\cN_{S^c} = 0\}}$ as well as $[\Gamma, \mathbb{P}_k ] = 0$ for all $k \in S$. Then  
\begin{equation}\label{eq:Q3SM}  \tr \, T_c^* Q_3^S T_c \, \Gamma  =  \tr \,  T_c^* Q_3^M T_c  \, \Gamma = 0.
\end{equation} 
\end{lemma} 

\begin{proof}
As noticed below (\ref{eq:parity}), we have $[ \mathbb{P}_k, T_c ] = 0$ for all $k \in S$. By assumption, we consider density matrices $\Gamma$ such that $[\Gamma, \mathbb{P}_k] = 0$ for all $k \in S$. It follows that 
\begin{equation}\label{eq:TGT-sym} T_c \Gamma T_c^* = \mathbb{P}_k T_c \Gamma T_c^* \mathbb{P}_k + \mathbb{Q}_k T_c \Gamma T_c^* \mathbb{Q}_k \end{equation} 
for all $k \in S$. On the other hand, from the definition of $\cM_x$ in \eqref{eq:def_Mx} and $\cN_i a_j^* = a_j^*(\cN_i + \delta_{i,j})$ we find 
\begin{equation} \label{eq:Q3S_parity}
\cM_x a_{-r}^*a_{r+p}^*a_p = a_{-r}^*a_{r+p}^*a_p \left(\cM_x + \delta_{x,-r} + \delta_{x,r+p} - \delta_{x,p}\right).
\end{equation}
The analogous statement holds for $\cM_{-x}$. Distinguishing different cases, we readily show that the parity has to be violated for some $x\in S$, i.e. 
for every $r,p \in S$ with $r+p \in S$ there exists a $x\in S$ such that 
\[ \mathbb{P}_x a_{-r}^* a_{r+p}^* a_p  = a_{-r}^* a_{r+p}^* a_p  \mathbb{Q}_x. \]
Thus we find that $\tr  \, a_{-r}^* a_{r+p}^* a_p T_c \Gamma T_c^* = 0$ for all $r,p \in S$ with $r+p \in S$. Hence $\tr Q_3^S T_c \Gamma T_c^* = 0$. 

To show the second equality in (\ref{eq:Q3SM}), we derive an identity similar to \eqref{eq:Q3S_parity}. Here this task is more subtle than above, because the second term in the definition (\ref{eq:def_Mx}) of $\cM_x$, measuring the number of $x$-connections, does not commute with $Q_3^M$. We are going to use Lemma \ref{lm:mono}. We decompose 
$$
Q_3^M =  \sum_{\substack{r,p\in H: \\ r+p \in S}} N^{1/2} \hat{V}_N (r) a^*_{-r} a^*_{r+p} a_p + \hc + \sum_{p\in H_r,r\in S} N^{1/2} \hat{V}_N (r) a^*_{-r} a^*_{r+p} a_p + \hc =: Q_3^{M,1} + Q_3^{M,2}.
$$
First, we consider $Q_3^{M,1}$. 
From \eqref{eq:cTGTc}, we find
\[ \begin{split}  \tr \,  T_c^*Q_3^{M,1} T_c \Gamma &= \sum_{\substack{r,p\in H: \\ r+p \in S}} N^{1/2} \hat{V}_N (r)  \tr \,  a^*_{-r} a^*_{r+p} a_p T_c \Gamma T_c^* \\ &=  \sum_{\substack{r,p\in H: \\ r+p \in S}} N^{1/2} \hat{V}_N (r)  \tr \, \chi_{-r} a^*_{-r} a^*_{r+p} a_p \chi_p T_c \Gamma T_c^*. \end{split} \]
Using that $ a_{-r}^* =  \mathds{1}_{\{\cN_{-r} \geq 1\}} a_{-r}^*$ and $ a_p = a_p\mathds{1}_{\{\cN_{p} \geq 1\}}$, we may equivalently consider the expectation, in the state $T_c \Gamma T_c^*$, of the observable
\begin{equation}
\label{eq:QM-mono}
\sum_{y,z \in S} \mathds{1}_{\{\cN_{-r} =1\}} \mathds{1}_{\{\cN_{r+y} =1\}} \prod_{ y'\in S\setminus\{y\}}  \mathds{1}_{\{\cN_{r+y'} = 0\} } a_{-r}^*a_{r+p}^*a_p \mathds{1}_{\{\cN_{p} =1\}} \mathds{1}_{\{\cN_{-p+z} =1\}} \prod_{ z'\in S\setminus\{z\}} \mathds{1}_{\{\cN_{-p+z'} = 0 \} }
\end{equation}
for $r,p\in H$ such that $r+p\in S$. 
%
To show that the expectation of each summand in the previous line vanishes, we apply the parity argument as we did for $Q_3^S$. For $x \in S$, using the definition of $\cM_x$ in (\ref{eq:def_Mx}) and $\mathcal N_i a^*_j = a^*_j (\mathcal N_i + \delta_{i,j})$, we find
\begin{align*}
\cM_x & a_{-r}^*a_{r+p}^*a_p = a_{-r}^*a_{r+p}^* \left(\cM_x + \delta_{x,r+p} + \frac{1}{2}\sum_{t\in H_x} \left(\delta_{t,r} \cN_{t+x} + \delta_{t+x,-r}\cN_{-t} \right) \right) a_p
\\
	& = a_{-r}^*a_{r+p}^*a_p \big(\cM_x + \delta_{x,r+p}\big) + \cN_{r+x} \mathds{1}_{H_x}(r) a_{-r}^*a_{r+p}^*a_p - a_{-r}^*a_{r+p}^*a_p \, \cN_{-p+x} \mathds{1}_{H_x}(-p).
\end{align*}
Thus, we obtain, for $y,z \in S$ such that $r+y, -p+z \in H$ (note that otherwise the expectation in \eqref{eq:QM-mono} vanishes as $r+y,-p+z$ cannot be in $S$ and there is no particle in $(H\cup S)^c$)
\begin{align*}
\cM_x \mathds{1}_{\{\cN_{r+y} =1\}} & \prod_{ y'\in S\setminus\{y\}} \mathds{1}_{\{\cN_{r+y'} = 0\} } a_{-r}^*a_{r+p}^*a_p  \mathds{1}_{\{\cN_{-p+z} =1\}}  \prod_{ z'\in S\setminus\{z\}} \mathds{1}_{\{\cN_{-p+z'} = 0\} }
\\
& = \mathds{1}_{\{\cN_{r+y} =1\}}  \prod_{ y'\in S\setminus\{y\}} \mathds{1}_{\{\cN_{r+y'} = 0\} } a_{-r}^*a_{r+p}^*a_p \mathds{1}_{\{\cN_{-p+z} =1\}}  \prod_{ z'\in S\setminus\{z\}} \mathds{1}_{\{\cN_{-p+z'} = 0 \} }\\ &\hspace{8cm} \times \left(\cM_x + \delta_{x,r+p} + \delta_{x,y} - \delta_{x,z} \right)
\end{align*}
(compare this to \eqref{eq:Q3S_parity}).
The same calculation for $\cM_{-x}$ and a case distinction shows that the parity is violated for some $x\in S$, i.e. for every $p,r \in H$ such that $r+p \in S$ and for every $y,z \in S$, there exists 
$x \in S$ such that
\begin{align} \label{eq:Q3M1_parity}
&\mathbb{P}_x \mathds{1}_{\{\cN_{r+y} =1\}}  \prod_{ y'\in S\setminus\{y\}} \mathds{1}_{\{\cN_{r+y'} = 0\} }  a_{-r}^*a_{r+p}^*a_p \mathds{1}_{\{\cN_{-p+z} =1\}}  \prod_{z'\in S\setminus\{y\}} \mathds{1}_{\{\cN_{r+z'} = 0\} }  \nn \\
& \quad =\mathds{1}_{\{\cN_{r+y} =1\}}  \prod_{ y'\in S\setminus\{y\}} \mathds{1}_{\{\cN_{r+y'} = 0\} }  a_{-r}^*a_{r+p}^*a_p \mathds{1}_{\{\cN_{-p+z} =1\}}  \prod_{z'\in S\setminus\{y\}} \mathds{1}_{\{\cN_{r+z'} = 0\} } \mathbb{Q}_x.
\end{align}
With (\ref{eq:TGT-sym}) this implies that the expectation of each term in (\ref{eq:QM-mono}) vanishes. By linearity, we conclude that $\tr \, T_c^* Q_3^{M,1} T_c \Gamma = 0$. Similarly, swapping the roles of $-r$ and $r+p$, we also obtain $\tr \, T_c^* Q_3^{M,2} T_c \Gamma = 0$.
\end{proof}

\subsection{Action on the quartic potential energy operator} 

Finally, we study the conjugation of the quartic term $Q_4$ in (\ref{eq:claimB1}). On the range of the projection $\Xi = \mathds{1}_{\{\mathcal N_{(H\cup S)^c}=0\}} \mathds{1}_{\{\mathcal N_H \in 2\mathbb{N}_0\}}$, we can decompose 
\[ Q_4 = Q_4^H + Q_4^M + Q_4^S \] 
with
\begin{equation} \label{eq:split-Q4} 
\begin{split} 
Q_4^H & = \frac{1}{2} \sum_{\substack{r\in \Lambda^* \\ p,q \in H_r}} \widehat{V}_N (r) a_{p+r}^*a_q^*a_{q+r}a_p
\\
Q_4^M &= \frac{1}{2} \sum_{\substack{q \in H, p\in S, r\in \Lambda^*: \\ p+r \in H, q+r \in S}} \widehat{V}_N(r) a_{p+r}^*a_q^*a_{q+r}a_p + \hc
+ \frac{1}{2} \sum_{\substack{p,q \in S, r \in \Lambda^*: \\ p+r, q+r \in H}} \widehat{V}_N(r) a_{p+r}^*a_q^*a_{q+r}a_p  + \hc
\\
&\, + \frac{1}{2} \sum_{\substack{p \in H, q\in S, r\in \Lambda^*: \\ p+r \in H, q+r \in S}} \widehat{V}_N(r) a_{p+r}^*a_q^*a_{q+r}a_p + \hc =: Q_4^{M,1} + Q_4^{M,2} + Q_4^{M,3}
\\
Q_4^S &= \frac{1}{2} \sum_{\substack{q,p \in S, r\in \Lambda^*: \\ p+r,q+r \in S}} \widehat{V}_N (r) a_{p+r}^*a_q^*a_{q+r}a_p.
\end{split} \end{equation} 

Let us first consider the term $Q_4^H$. It is convenient to define 
\[ \wt{Q}_4^H = \frac{1}{2} \sum_{\substack{r \in \Lambda^*  \\ p, q \in H_r}} \widehat{V}_N (r) \mathds{1}_S (p+q+r) a_{p+r}^* a_q^* a_{q+r} a_p. \]
In the next lemma we show that the difference $Q_4^H - \wt{Q}_4^H$ is small on an appropriate class of trial states. 
\begin{lemma}\label{lem:Q4_HS} 
Suppose that  $\Gamma$ is a density matrix on $\cF (\Lambda)$ satisfying $\Gamma = \mathds{1}_{\{\cN_{S^c} = 0\}} \, \Gamma \mathds{1}_{\{\cN_{S^c} = 0\}}$. Then we have
\begin{align*}
\pm \tr \, T_c^* (Q_4^H - \widetilde{Q}_4^{H})T_c \Gamma \lesssim N^{-1+5\kappa/2+3\vep} \tr \, T_c^* \cN_H^2 T_c\Gamma ,
\end{align*}
\end{lemma}

\begin{proof} 
Let us first explain the idea of the proof, heuristically. Consider the expectation of $a_{p+r}^*a_q^*a_{q+r}a_p$ in the state $T_c \Gamma T_c^*$. The operator $a_p$ annihilates a particle with momentum $p$ which, by \eqref{eq:cTGTc}, has to be connected to exactly one other particle with momentum $-p+x, x\in S$. After the  annihilation and application of the remaining $a_{p+r}^*a_q^*a_{q+r}$ this particle again needs to be connected. So either $q+r = -p + x$, which is the main term of $\widetilde{Q}_4^H$, or $-p+x \in -(p+r)+S \, \cup \, -q+S$.
This leads to the condition $r \in S+S$ or $q-p \in S+S$. Either constraint is enough to show that these contributions are negligible. 

We will now make this argument rigorous. First, let us assume that $r \notin S+S$ and $q-p \notin S +S$. Then from Lemma \ref{lm:mono}, we find  
\begin{align} 
\tr \Big(a_{p+r}^*a_q^* &a_{q+r}a_p T_c \Gamma T_c^* \Big) \nn \\ &= \tr \Big(a_{p+r}^*a_q^*a_{q+r} a_p \mathds{1}_{\{\cN_p = 1\}} \sum_{x \in S}\mathds{1}_{\{\cN_{-p+x} = 1\}} T_c \Gamma T_c^*  \Big) \nn
\\
& = \sum_{x \in S} \tr \Big( \mathds{1}_{\{\cN_{-p+x} = 1\}} a_{p+r}^*a_q^*a_{q+r}a_p   \mathds{1}_{\{\cN_{p} = 1\}}  T_c \Gamma T_c^* \Big) \mathds{1}(q,p+r \neq -p+x) \nn
\end{align} 
because, on the one hand, $-p+x \not = q+r$ (from the assumption $p+q+r \not \in S$) and, on the other hand, the contributions from $-p+x = q$ and $-p+x = p+r$ vanish, since there cannot be two $H$-particles in the same momentum state on the range of $T_c \Gamma T_c^*$, i.e.\ $\mathds{1}_{\{\cN_k \geq 2\}} T_c \Gamma T_c^* = 0$ for all $k \in H$, which follows from (\ref{eq:cTGTc}). Applying again (\ref{eq:cTGTc}), we obtain
\begin{align} 
\tr \Big( &a_{p+r}^*a_q^* a_{q+r}a_p T_c \Gamma T_c^* \Big) \nn \\ &= \sum_{x,y \in S} \tr \Big( \mathds{1}_{\{\cN_{-p+x} = 1\}} \mathds{1}_{\{\cN_{p-x+y} = 1\}}  a_{p+r}^*a_q^*a_{q+r}a_p   \mathds{1}_{\{\cN_{p} = 1\}}    T_c \Gamma T_c^*   \Big) \mathds{1}(q,p+r \neq -p+x) \nn
\\
& = \sum_{x,y \in S} \tr \Big( a_{p+r}^*a_q^*a_{q+r}a_p  \mathds{1}_{\{\cN_{-p+x} = 1\}} \mathds{1}_{\{\cN_{p-x+y} = 1\}}   \mathds{1}_{\{\cN_{p} = 1\}}   T_c \Gamma T_c^*  \Big) \nn
\\
&\qquad \qquad \hspace{3cm} \times \mathds{1}(q,p+r \neq -p+x)\mathds{1}(q+r,p \neq p-x+y) \nn
\end{align}
since $p-x+y \not = p+r , q$ from the assumption $r, q-p \not \in S+S$ and the contribution from $p-x+y = q+r,p$ can be excluded as above, using that $\mathds{1}_{\{\cN_k \geq 2\}} T_c \Gamma T_c^* = 0$ for all $k \in H$. Since the particle with momentum $-p+x$ is $S$-connected both with the particle with momentum $p$ and the particle with momentum $p-x+y$, we conclude that 
\begin{equation}\label{eq:Q4H_tilde_shell} \tr \Big(a_{p+r}^*a_q^*a_{q+r}a_p T_c \Gamma T_c^* \Big) = 0 \end{equation}

To handle the case $r \in S+S$ or $q-p \in S +S$, we observe that, by the Cauchy--Schwarz inequality, we find 
\begin{align*}
\pm \sum_{\substack{r \in \Lambda^* \\ p,q \in H_r}} \hat{V}_N(r) a_{p+r}^*a_q^*a_{q+r}a_p \mathds{1}_{S^c}(p+q+r) &\left( \mathds{1}_{S+S}(r) + \mathds{1}_{(S+S)^c}(r)\mathds{1}_{S+S}(q-p)\right) \\ &\hspace{3cm} \lesssim N^{-1+5\kappa/2+3\vep} \cN_H^2,
\end{align*}
where we used that $\|\widehat V\|_\infty \lesssim N^{-1+\kappa}$ and $|S| \lesssim N^{3\kappa/2 + 3\varepsilon}$. 
\end{proof}

Next, we conjugate $\wt{Q}_4^H$ with $T_k$, for a $k \in S$. 
\begin{lemma}\label{lem:Q4_H}
On the range of $\Lambda_k$, we have, in the sense of quadratic forms,
\begin{align*}
\pm &\bigg( T_k^* \widetilde{Q}_4^H T_k - \widetilde{Q}_4^H -  \sum_{r,p \in H_k} \widehat{V}_N(r-p) N\varphi_r (\varphi_p + \varphi_{p+k}) \Theta_{k,r}^{(2)}  \Theta_{k,p}^{(2)} a_k^*a_k \bigg) 
\\
&\hspace{3cm}  \lesssim N^{-5+15\kappa/2+9\vep} (\cN_k+1) \sum_{p,q \in H} \1_S(p+q) |q|^2 a_p^*a_q^*a_qa_p  + N^{-2+4\kappa+\eps} \cN_k^2.
\end{align*}
\end{lemma}

\begin{proof} 
Let us fix $k \in S$, $r \in \Lambda^*$, $p,q \in H_r$ with $p+q+r \in S$. Using (\ref{eq:TTL}), we obtain, on $\Ran \Lambda_k$,
\begin{align}
T_k^* a_{p+r}^* a_q^* a_{q+r} a_p T_k
	&= \cos(X_k)  a_{p+r}^* a_q^* a_{q+r} a_p  \cos(X_k) + \frac{\sin(X_k)}{X_k} \mathcal B_k^\circ  a_{p+r}^* a_q^* a_{q+r} a_p  \mathcal B_k^\sharp\frac{\sin(X_k)}{X_k} \nn \\
	&\quad + \Big[ \cos(X_k)  a_{p+r}^* a_q^* a_{q+r} a_p \mathcal B_k^\sharp\frac{\sin(X_k)}{X_k} + \hc \Big]  \label{eq:Q4H_0}
\end{align}
Let us first show that the cross terms on the second line vanish. We have 
\begin{align} 
 a_{p+r}^*a_q^*a_{q+r}a_p  \cB_k^\sharp 
	&= \sum_{t \in H_k} a_{p+r}^*a_q^* a_{q+r}a_p a^*_{-t} a^*_{k+t}  a_k  N^{1/2} \varphi_t \Theta_{k,t} \nn 	\\
	&= \sum_{t\in H_k} \left( a^*_{-t} a^*_{k+t}  a_{p+r}^*a_q^* a_{q+r}a_p + a_{p+r}^*a_q^* \big[ a_{q+r}a_p,  a^*_{-t} a^*_{k+t} \big]  \right)a_k  N^{1/2}\varphi_t \Theta_{k,t} \nn \end{align} 
Using 
\begin{equation}\label{eq:commut} \begin{split} 
\big[ a_{q+r}a_p,  a^*_{-t} a^*_{k+t} \big]  = \; & a^*_{k+t} a_{q+r} \delta_{p, -t} + a^*_{-t} a_{q+r} \delta_{p, k+t} +\delta_{q+r,-t} a^*_{k+t} a_p + \delta_{q+r,k+t} a^*_{-t} a_p \\ &+ \delta_{q+r, k+t} \delta_{p, -t} +\delta_{p,k+t} \delta_{q+r, -t}  \end{split} \end{equation} 
and noticing that, since the cutoff $\Theta_{k,t}$ imposes that the $S$-neighbourhoods of $t$ and of $-(t+k)$ are empty and since $p+q+r \in S$, the contribution of the quadratic terms on the right-hand side of (\ref{eq:commut}) vanishes, we obtain    	
\begin{align} \label{eq:Q4H_tilde}	
a_{p+r}^* &a_q^* a_{q+r}a_p  \cB_k^\sharp  \nn \\
	&= \sum_{t\in H_k} \left( a^*_{-t} a^*_{k+t}  a_{p+r}^*a_q^* a_{q+r}a_p + a_{p+r}^*a_q^* (\delta_{q+r,-t}\delta_{p,k+t} + \delta_{q+r,k+t}\delta_{p,-t}) \right)a_k  N^{1/2}\varphi_t \Theta_{k,t} .
\end{align}
Observing now that all terms in the sum on the right-hand side create a $k$-connection, we conclude that $\Lambda_k a_{p+r}^*a_q^*a_{q+r}a_p  \cB_k^\sharp = 0$. Thus, the cross terms vanish on $\Ran \Lambda_k$. 

Let us now consider the second term on the right-hand side of (\ref{eq:Q4H_0}). From \eqref{eq:Q4H_tilde} and using, similarly to (\ref{eq:BcircBdag}), the fact that $\Theta_{k,t}$ excludes $k$-connections and particles in the $S$-neighborhoods of $t$ and $-(t+k)$, we find, on the range of $\Lambda_k$, 
\begin{align*}
&\cB_k^\circ a_{p+r}^*a_q^*a_{q+r}a_p  \cB_k^\sharp 
\\
	&= \sum_{t\in H_k}  \cB_k^\circ \left( a^*_{-t} a^*_{k+t}  a_{p+r}^*a_q^* a_{q+r}a_p + a_{p+r}^*a_q^* (\delta_{q+r,-t} \delta_{p,k+t} + \delta_{q+r,k+t} \delta_{p,-t}) \right)a_k  N^{1/2}\varphi_t \Theta_{k,t} \\
	&= \sum_{t\in H_k}  \Theta_{k,t} N \varphi_t(\vphi_t + \varphi_{k+t}) a^*_k a_k a_{p+r}^*a_q^*  a_{q+r}a_p  \Theta_{k,t} \\
	&\quad + \sum_{t,t'\in H_k} N \varphi_t\vphi_{t'} (\delta_{p+r,-t'} \delta_{q,k+t'} + \delta_{p+r,k+t'} \delta_{q,-t'})  (\delta_{q+r,k+t} \delta_{p,-t} + \delta_{q+r,-t} \delta_{p,k+t}) \Theta_{k,t'} a_k^* a_k   \Theta_{k,t} \\
	&= a_{p+r}^*a_q^* (X_k^{(p+r,q,q+r,p)})^2  a_{q+r}a_p   + N (\vphi_p+\vphi_{p-k})(\vphi_q+\vphi_{q-k})\Theta_{k,-p} \Theta_{k,-q} \delta_{p+q+r,k} \cN_k
\end{align*}
with
\begin{align}
(X_k^{(p+r,q,q+r,p)})^2 
	&= \sum_{\substack{t \in H_k}} N \varphi_t(\vphi_t + \varphi_{k+t}) \Theta_{k,t} a^*_k a_k \big(1-\1(p+r,q,q+r,p \simS -(t+k),t)\big) \nn \\
	&=: X_k^2 - \delta_{k}^{(p+r,q,q+r,p)}. \label{eq:def_delta_4}
\end{align}
Here, similarly to the analysis of $Q_3^H$, we used \eqref{eq:comm_a_Theta} to commute $\Theta_{k,t}$ (on the range of $\Lambda_k$, we have $\Theta_{k,t} = \Theta^{(2)}_{k,t}$) and we defined 
$$
\1(p+r,q,q+r,p \simS -(t+k),t) :=
\begin{cases}
1 & \text{if } \{p+r,q,q+r,p\} \cap \big((-(t+k)+S)\cup (t+S)\big) \neq \emptyset \\
0 & \text{otherwise},
\end{cases}
$$
which checks if at least one of the momenta $p+r,q,q,+r$ or $p$ is in the shell-vicinity of $-(t+k)$ or $t$.
Notice that $\delta_{k}^{(p+r,q,q+r,p)}$ is a non-negative operator. Moreover, \Cref{lem:vphi_prop} implies that 
\begin{align} \label{eq:est_delta_4}
\delta_{k}^{(p+r,q,q+r,p)} \leq 8 |S| N \|\varphi_H\|_{\infty}^2 \mathcal N_k 
\lesssim N^{-5+15\kappa/2 + 7\varepsilon} \mathcal N_k.
\end{align}
Recalling the definition \eqref{eq:def_Xk^-r} and using \eqref{eq:comm_Xk^-r} we have, from \eqref{eq:Q4H_0}, on $\Ran \Lambda_k$,
\begin{align}
& T_k^* a_{p+r}^*a_q^*a_{q+r}a_p T_k \nn \\
&=  a_{p+r}^*a_q^* \Big( \cos(X_k^{(p+r,q)})\cos(X_k^{(p,q+r)}) +  \frac{\sin(X_k^{(p+r,q)} )}{X_k^{(p+r,q)}} (X_k^{(p+r,q,q+r,p)})^2  \frac{\sin(X_k^{(p,q+r)})}{X_k^{(p,q+r)}} \Big) a_{q+r}a_p \nn \\
&\quad +  N (\vphi_p+\vphi_{p-k})(\vphi_q+\vphi_{q-k})\Theta_{k,-p} \Theta_{k,-q} \delta_{p+q+r,k} \cN_k \frac{\sin^2 X_k}{X_k^2}. \label{eq:Q4H_1}
\end{align}
Using trigonometric identities and the fact that the operators $X_k^{(p+r,q)},X_k^{(q,p+r)}, X_k^{(p+r,q,q+r,p)}$ commute with each other, we can write 
\begin{align*}
&\cos(X_k^{(p+r,q)})\cos(X_k^{(p,q+r)}) +  \frac{\sin(X_k^{(p+r,q)} )}{X_k^{(p+r,q)}} (X_k^{(p+r,q,q+r,p)})^2  \frac{\sin(X_k^{(p,q+r)})}{X_k^{(p,q+r)}}\\
&= \cos(X_k^{(p+r,q)} - X_k^{(p,q+r)}) + \frac{\sin(X_k^{(p+r,q)} )}{X_k^{(p+r,q)}}\frac{\sin(X_k^{(p,q+r)})}{X_k^{(p,q+r)}} \left((X_k^{(p+r,q,q+r,p)})^2 -X_k^{(p+r,q)} X_k^{(p,q+r)}\right) \\
&=1 + R_k^{(p+r,q,q+r,p)}.
\end{align*}
We can estimate 
$$
\pm R_k^{(p+r,q,q+r,p)} \leq \delta_{k}^{(p+r,q,q+r,p)} + 3\delta_{k}^{(p+r,q)} + 2\delta_{k}^{(p,q+r)} \lesssim N^{-5+15\kappa/2+7\vep } \cN_k
$$
as follows from the inequality 
\[ \Big| \cos(x-y) + \frac{\sin x}{x}\frac{\sin y}{y} (z^2-xy)-1 \Big| \leq (v^2-z^2) + 3(v^2-x^2) + 2 (v^2-y^2) \] for all $0\leq x,y,z \leq v$ (with $v$ playing the role of $X_k$). The contribution of the remainder terms $R_k^{(p+r,q,q+r,p)}$ to $T_k \wt{Q}_4^H T_k$ can therefore be estimated, using the Cauchy-Schwarz inequality, by 
\begin{align*}
\pm \sum_{\substack{r\in \Lambda^* \\
p,q \in H_r}} &\hat{V}_N(r)   \1_S(p+q+r) a_{p+r}^*a_q^* R_k^{(p+r,q,q+r,p)} a_{q+r}a_p \\
&\leq \sum_{\substack{r\in \Lambda^* \\
p,q \in H_r}} \eta \frac{|\hat{V}_N(r)|}{p^2} \1_S(p+q+r) q^2 a_{p+r}^*a_q^* R_k^{(p+r,q,q+r,p)}(\cN_k+1)^{-1} {R_k^{(p+r,q,q+r,p)}} a_q a_{p+r} 
\\
&\qquad \qquad \qquad + \eta^{-1} \sum_{\substack{r\in \Lambda^* \\
p,q \in H_r}}  \frac{|\hat{V}_N(r)|}{q^2}  p^2 \1_S(p+q+r) a_p^* a_{q+r}^* (\cN_k+1) a_{q+r}a_p 
\end{align*} 
for every $\eta > 0$. Choosing $\eta^{-1} = N^{-5 + 15\kappa/2 + 7\eps}$, we obtain 
\begin{align} \label{eq:QH4-fin}  
\pm \sum_{\substack{r\in \Lambda^* \\
p,q \in H_r}} \hat{V}_N(r) & \1_S(p+q+r) a_{p+r}^* a_q^* R_k^{(p+r,q,q+r,p)} a_{q+r}a_p \nn \\
&\lesssim N^{-5+15\kappa/2+7\vep}  \sum_{\substack{r\in \Lambda^*   \\
q \in H}}  \frac{|\hat{V}_N(r)|}{|p-r|^2} \1_S(p+q) |p|^2a_p^*a_q^*a_pa_q (\cN_k+1) \nn \\
&\lesssim N^{-5+15\kappa/2+9\vep} \sum_{\substack{r\in \Lambda^* \\
q \in H}}  \1_S(p+q) |p|^2a_p^*a_q^*a_pa_q (\cN_k+1).
\end{align}
In the last inequality, we applied Hölder's inequality and the bound $ \|V_N \|_{L^{3/2 + \varepsilon}} \lesssim N^{2\varepsilon}$. 

Let us now consider the last term on the right-hand side of \eqref{eq:Q4H_1}. To evaluate its contribution to $T_k^* \wt{Q}_4^H T_k$, we observe that 
\begin{equation} \label{eq:Q4H_1_second}
\begin{split} 
\frac{1}{2} \sum_{\substack{r\in \Lambda^* \\ p,q \in H_r}} &\hat{V}_N(r) N (\vphi_p+\vphi_{p-k})(\vphi_q+\vphi_{q-k})\Theta_{k,-p} \Theta_{k,-q} \delta_{p+q+r,k} \cN_k \\  &\hspace{3cm} = \sum_{r,p\in H_k} N\widehat{V}_N(r-p) \varphi_r (\varphi_p + \vphi_{p+k})  \Theta_{k,p}\Theta_{k,r}\mathcal N_k.
\end{split} \end{equation} 
By Young's inequality and estimating $\| \ph \|_1 \lesssim C$ with Lemma \ref{lem:vphi_prop}, we find  
\[ \Big| \sum_{r,p\in H_k} N \widehat{V}_N(r-p) \varphi_r (\varphi_p + \vphi_{p+k})  \Big| \lesssim N^\kappa. \]
With the bound
\[ \Big| \frac{\sin^2 x}{x^2} - 1 \Big| \leq C x^2 \]
valid for all $x > 0$ and with the estimate (\ref{eq:Xp_bound}) for $X_k$, we conclude that 
\[ \begin{split} \pm \Big[ \sum_{\substack{r \in \Lambda^* \\ p,q \in H_r}}  \widehat{V}_N (r) &\mathds{1}_S (p+q+r)
 N (\ph_p + \ph_{p-k}) (\ph_q + \ph_{q-k}) \Theta_{k, -p} \Theta_{k, -q} \delta_{p+q+r, k} \cN_k \frac{\sin^2 X_k}{X_k^2} \\ & \hspace{2cm} -  \sum_{r,p\in H_k} N\widehat{V}_N(r-p) \varphi_r (\varphi_p + \vphi_{p+k})  \Theta_{k,p}\Theta_{k,r}\mathcal N_k \Big] \lesssim N^{-2+4\kappa+\eps} \cN_k^2.  \end{split} \]
 Together with (\ref{eq:QH4-fin}), we obtain the statement of the lemma. 
\end{proof}

Next, we consider the term $Q_4^S$, where the four momenta are in the shell. We show that its contribution is negligible.  
\begin{lemma} \label{lem:Q4S}
Let $\Gamma$ be a density matrix on $\cF (\Lambda)$ satisfying $\Gamma = \mathds{1}_{\{\cN_{S^c} = 0\}} \, \Gamma \mathds{1}_{\{\cN_{S^c} = 0\}}$ as well as $[\Gamma, \mathbb{P}_k ] = 0$ for all $k \in S$ (with $\mathbb{P}_k$ the parity operator defined in (\ref{eq:parity})). Then we have 
\begin{align*}
\tr \, T_c^* Q_4^S T_c \Gamma \lesssim N^{-1+\kappa} \tr\, T_c^* {\cN_S}(\cN_S + N^{3\kappa/2 + 3\vep}) T_c \Gamma .
\end{align*}
\end{lemma}

\begin{proof}
Let us once again use the parity argument:
Recall (\ref{eq:def_Mx}), where, for $k\in S$, we defined
\[ \mathcal M_k = \cN_k + \sum_{t \in H_k} \cN_{-t} \cN_{t+k}, \] 
counting the number of particles with momenta $k$ and the number of $k$-connections. 
From the assumption $\Gamma = \mathds{1}_{\{\cN_{S^c} = 0\}} \Gamma \mathds{1}_{\{\cN_{S^c} = 0\}}$ with $[\Gamma , {\mathbb{P}}_k ] = 0$, and from the observation that $[ T_c, \mathbb{P}_k] = 0$, we have $T_c \Gamma T_c^* = \mathbb{P}_k  T_c \Gamma T_c^* \mathbb{P}_k + \mathbb{Q}_k  T_c \Gamma T_c^*\mathbb{Q}_k$, for all $k\in S$. 
After computing the commutator of $\cM_k$ and $\cM_{-k}$ with $a_{p+r}^*a_q^*a_{q+r}a_p$ we find that the contribution of $a_{p+r}^*a_q^*a_{q+r}a_p$ vanishes unless one of the following conditions holds true:
\begin{itemize}
\item $p+r = \pm q$ and $q+r = \pm p$
\item $p+r = \pm q+r$ and $q = \pm p$
\item $p+r = \pm p$ and $q = \pm q+r$,
\end{itemize}
This implies that, necessarily, $-r = p+q$, $p=q$ or $r=0$ must hold true. Therefore, we find that
\begin{align*}
& \tr \, T_c^* Q_4^ST_c\Gamma T_c^*  \\  
	& \lesssim \; \tr \, \Big( \sum_{p,q\in S} \widehat{V}_N(p+q) a^*_{-q} a^*_q a_{-p} a_p + \sum_{p,r\in S} \widehat{V}_N(r) a^*_{p+r} a^*_{p} a_{p+r} a_p &+ \sum_{p,q\in S} \widehat{V}_N (0) a^*_{p} a^*_{q} a_{q} a_p  \Big) T_c\Gamma T_c^*.
\end{align*}
Thus 
\begin{align*}
\tr \, T_c^* Q_4^S T_c  \Gamma  \lesssim \|\widehat{V}_N\|_{L^\infty} \tr  \, \cN_S (\cN_S + N^{3\kappa/2+3\varepsilon}) T_c  \Gamma T_c^* .
\end{align*}
Here we used that
\begin{align*}
\sum_{p,q\in S} \widehat{V}_N (p+q) a^*_{-q} a^*_{q} a_{-p} a_p &= \sum_{p\in S} \widehat{V}_N (0) \mathcal N_p(\mathcal N_p-1)  + \sum_{\substack{p, q\in S \\ p \neq q}} \widehat{V}_N (p+q) (a^*_{-q} a_{-p})(a^*_{q} a_p) \\
	&\lesssim \| \hat{V}_N \|_{L^\infty} \mathcal N_S(\mathcal N_S + N^{3\kappa/2+3\varepsilon}).
\end{align*}
\end{proof}

As for the term $Q_4^M$, we decompose it as in (\ref{eq:split-Q4}), writing $Q_4^M = Q_4^{M,1} + Q_4^{M,2} + Q_4^{M,3}$. First, we handle $Q_4^{M,1}$.  
\begin{lemma} \label{lem:Q4M1}
Let $\Gamma$ be a density matrix on $\cF (\Lambda)$ satisfying $\Gamma = \mathds{1}_{\{\cN_{S^c} = 0\}} \, \Gamma \mathds{1}_{\{\cN_{S^c} = 0\}}$ as well as $[\Gamma, \mathbb{P}_k ] = 0$ for all $k \in S$. Then we have
\begin{align*}
\tr  \, T_c^* Q_4^{M,1} T_c \Gamma = 0.
\end{align*}
\end{lemma}
\begin{proof}
Recall that 
\begin{align*}
Q_4^{M,1} = \frac{1}{2}\sum_{\substack{q \in H,p\in S,r\in \Lambda^*: \\ p+r \in H, q+r \in S}} \widehat{V}_N(r) a_{p+r}^*a_q^*a_{q+r}a_p + \hc.
\end{align*}
From $q \in H, q+r \in S$, we find $r \not \in S+S$. Similarly, from $p+r \in H, q+r \in S$, we obtain $q-p \not \in S+S$. With the analysis at the beginning of the proof of \Cref{lem:Q4_HS}, in particular \eqref{eq:Q4H_tilde_shell}, we conclude that, unless $p+q+r \in S$, the expectation of $a_{p+r}^* a_q^* a_{q+r} a_p$ in the state $T_c \Gamma T_c^*$ vanishes. For $p+q+r \in S$, the operator $a_{p+r}^* a_q^* a_{q+r} a_p$ creates a $(p+q+r)$-connection, annihilates a particle with momentum $p \in S$ and a particle with momentum $q+r \in S$. Consequently, also this term vanishes in expectation due to the same parity argument that was used in the previous lemmata.
\end{proof}

Finally, we estimate the expectation of the terms $Q_4^{M,2}$ and $Q_4^{M,3}$. 
\begin{lemma} \label{lem:Q4M2}
Let $\Gamma$ be a density matrix on $\cF (\Lambda)$ satisfying $\Gamma = \mathds{1}_{\{\cN_{S^c} = 0\}} \, \Gamma \mathds{1}_{\{\cN_{S^c} = 0\}}$ as well as $[\Gamma, \mathbb{P}_k ] = 0$ for all $k \in S$. Then we have
\begin{align*}
\pm \tr \, T_c^* (Q_4^{M,2} + Q_4^{M,3}) T_c \Gamma  \lesssim N^{-1+\kappa} \tr  \, T_c^* \mathcal{N}_H \mathcal (\cN_S + N^{3\kappa/2+3\vep}) T_c \Gamma.
\end{align*}
\end{lemma}
\begin{proof}
Let us focus on the contribution of $Q_4^{M,2}$, the one of $Q_4^{M,3}$ can be handled analogously. Recall that 
\begin{align*}
Q_4^{M,2} = \frac{1}{2} \sum_{\substack{p,q \in S, r\in \Lambda^*: \\ p+r,q+r \in H}} \widehat{V}_N(r) a_{p+r}^*a_q^* a_{q+r} a_p + \hc
\end{align*}
Using \Cref{lm:mono}, we write 
\begin{align*}
& \tr\, a_{p+r}^*a_q^*a_{q+r} a_p T_c\Gamma T_c^* 
\\
	&= \tr\,  \sum_{x,y\in S} \1_{\{\mathcal N_{p+r}=1\}}\1_{\mathcal \{N_{y-(p+r)}=1\}} a_{p+r}^*a_q^*a_{q+r} a_p \1_{\{\mathcal N_{q+r}=1\}} \1_{\{\mathcal N_{x-(q+r)}=1\}} T_c \Gamma T_c^* 
\end{align*} 
The term annihilates a $x$-connection and a particle with momentum $p$, and creates a $y$-connection and a particle with momentum $q$. Thus, by our standard parity argument, we find that one of the following conditions must be fulfilled:
\begin{enumerate}[label=(\arabic*)]
\item $x = \pm p$ and $y = \pm q$, 
\item $x = \pm q$ and $y = \pm p$,
\item $x = \pm y$ and $p = \pm q$.
\end{enumerate}

\emph{Case $(1)$:} With Cauchy--Schwarz we obtain
\begin{align} \label{eq:Q4M2_Case1}
& \pm \1_{\{\mathcal N_{p+r}=1\}}\1_{\mathcal \{N_{y-(p+r)}=1\}} a_{p+r}^*a_q^*a_{q+r} a_p \1_{\{\mathcal N_{q+r}=1\}} \1_{\{\mathcal N_{x-(q+r)}=1\}} \nn
\\
& = \pm \1_{\{\mathcal N_{p+r}=1\}}\1_{\mathcal \{N_{y-(p+r)}=1\}} a_{p+r}^* (a_pa_q^* - \delta_{p,q}) a_{q+r} \1_{\{\mathcal N_{q+r}=1\}} \1_{\{\mathcal N_{x-(q+r)}=1\}} \nn
\\
&\leq  a_{p+r}^* (a_p^*a_p + 1) a_{p+r} \1_{\{\mathcal N_{p+r}=1\}}\1_{\mathcal \{N_{y-(p+r)}=1\}} + a_{q+r}^* (a_q^*a_q + 1) a_{q+r} \1_{\{\mathcal N_{q+r}=1\}} \1_{\{\mathcal N_{x-(q+r)}=1\}} \nn
\\
& \qquad\qquad\qquad\qquad - \delta_{p,q} \, a_{p+r}^* a_{p+r} \1_{\{\mathcal N_{p+r}=1\}} \1_{\mathcal \{N_{y-(p+r)}=1\}} \1_{\{\mathcal N_{x-(p+r)}=1\}}. 
\end{align}
We sum over $x$ and $y$ and write $\pm p, \pm q$ for either value that $x$ respectively $y$ could take. The contribution of the first two terms on the right-hand side of the previous equation can be bounded by 
\begin{align*}
\sum_{\substack{p,q \in S, r\in \Lambda^*: \\ p+r,q+r \in H}} \hat{V}_N&(r) \tr \, a_{p+r}^* (a_p^*a_p + 1) a_{p+r} \1_{\{\mathcal N_{p+r}=1\}}\1_{\mathcal \{N_{\pm q -(p+r)}=1\}} T_c \Gamma T_c^* 
\\
&\leq \| \hat{V}_N\|_\infty  \tr \,   \sum_{p\in S} (a_p^*a_p + 1) \sum_{r\in H} a_{r}^*a_{r} \1_{\{\mathcal N_{r}=1\}} \sum_{q\in S} \1_{\mathcal \{N_{\pm q -r}=1\}} T_c \Gamma T_c^* 
\\
&= \| \hat{V}_N\|_\infty \tr \Big(T_c^* (\cN_S + |S|) \cN_H T_c \Gamma\Big),
\end{align*}
where we first shifted $r \to r-p$ and then we used \Cref{lm:mono} to replace the $q-$sum by $1$.
For the third term on the right-hand side of \eqref{eq:Q4M2_Case1} we readily find
\begin{align*}
\sum_{\substack{p,q \in S, r\in \Lambda^*: \\ p+r,q+r \in H}} \hat{V}_N(r) \delta_{p,q} \, a_{p+r}^* a_{p+r} \1_{\{\mathcal N_{p+r}=1\}} \1_{\mathcal \{N_{\pm q -(p+r)}=1\}} \1_{\{\mathcal N_{\pm p -(p+r)}=1\}}
\leq \| \hat{V}_N\|_\infty \cN_H |S|.
\end{align*}

\emph{Case $(2)$:} This can be handled like Case $(1)$ but without interchanging $a_q^*$ and $a_p$.

\emph{Case $(3)$:} With Cauchy-Schwarz, we obtain 
\begin{align*}
 \pm \1_{\{\mathcal N_{p+r}=1\}} &\1_{\mathcal \{N_{y-(p+r)}=1\}} a_{p+r}^*a_q^*a_{q+r} a_p \1_{\{\mathcal N_{q+r}=1\}} \1_{\{\mathcal N_{x-(q+r)}=1\}}
\\
&\leq  a_{p+r}^* a_q^*a_q a_{p+r} \1_{\{\mathcal N_{p+r}=1\}}\1_{\mathcal \{N_{y-(p+r)}=1\}} + a_{q+r}^*  a_p^*a_p a_{q+r} \1_{\{\mathcal N_{q+r}=1\}} \1_{\{\mathcal N_{x-(q+r)}=1\}}.
\end{align*}
The contribution of the first term (the second term can be handled analogously) can be estimated, after summing over $p$ and $x$, by
\begin{align*}
\sum_{\substack{q \in S, r\in \Lambda^*: \\ p+r,q+r \in H}} &\sum_{y \in S} \hat{V}_N(r) \tr \, a_{p+r}^* a_q^*a_q a_{p+r} \1_{\{\mathcal N_{p+r}=1\}}\1_{\mathcal \{N_{y-(p+r)}=1\}} \delta_{p,\pm q} T_c \Gamma T_c^* 
\\
&\leq \|\hat{V}_N\|_\infty \tr \, \sum_{q \in S} a_q^*a_q \sum_{r\in H} a_{r}^*a_{r} \1_{\{\mathcal N_{r}=1\}} \sum_{y \in S}\1_{\mathcal \{N_{y-r}=1\}} T_c \Gamma T_c^*
\\
&\leq \|\hat{V}_N\|_\infty \tr \, \cN_S \cN_H T_c \Gamma T_c^* \, .
\end{align*}
\end{proof}

\subsection{Proof of Lemma \ref{lem:Tc_main}}

First of all, we combine the statements of \Cref{lem:Tc_kin}, \Cref{lem:Q3_Tp} and \Cref{lem:Q4_H} to estimate the expectation of kinetic energy, cubic and quartic terms  on the right hand side of  \eqref{eq:claimB1}.  
\begin{lemma} \label{lem:Tc_changed}
Let $\Gamma$ be a density matrix on $\cF (\Lambda)$ satisfying $\Gamma = \mathds{1}_{\{\cN_{S^c} = 0\}} \, \Gamma \mathds{1}_{\{\cN_{S^c} = 0\}}$ as well as $[\Gamma, \mathbb{P}_k ] = 0$ for all $k \in S$. Then we have
\begin{equation} \label{eq:Tc_changed}
\tr \, T_c^* \Big[ \sum_{p \in \Lambda^*} p^2 a_p^*a_p + Q^H_3 + \wt{Q}^H_4 \Big] T_c \Gamma \leq \; \tr \, \sum_{p\in S}p^2a_p^*a_p \Gamma + 2N^\kappa \big( 8\pi\mathfrak{a} - \hat{V}(0)\big) \tr \, \cN_S \Gamma + \delta_1(\Gamma)
\end{equation}
with
\begin{equation} \label{eq:Tc_changed_error}
\delta_1(\Gamma) \lesssim N^{-1+5\kappa/2+3\vep/2} \tr \sum_{p\in S}\cN_p^2 \Gamma
+  N^{-5+17\kappa/2 + 11\eps} \tr \sum_{p \in S} \cN_p^2 (\cN_S + N^{3\kappa/2+3\eps}) \Gamma + N^{\kappa-\eps} \tr \, \cN_S \Gamma.
\end{equation}
Moreover, we have the bounds
\begin{align} \label{eq:apriori_Q3_Q4}
\tr \, T_c^* \sum_{p \in \Lambda^*} p^2 a_p^*a_p T_c \Gamma, \; \pm \tr \, T_c^* Q_3^H T_c \Gamma, \; \tr \, T_c^* Q^H_4 T_c \Gamma \lesssim N^{\kappa+2\vep}\tr \cN_S \Gamma + \delta_1(\Gamma). 
\end{align}

\end{lemma}
\begin{proof} We start with the proof of (\ref{eq:Tc_changed}). As a first step, we derive a rough upper bound on the kinetic operator. Let $k\in S' \subset S$.
From Lemma \ref{lem:Tc_kin} and Lemma \ref{lem:vphi_prop} we obtain, on the range of $\Lambda_k$, 
\[ T_k^* \sum_{p \in \Lambda^*} p^2 a_p^* a_p T_k \lesssim  \sum_{p \in \Lambda^*} p^2 a_p^* a_p + N^\kappa \cN_k + N^{-1+5\kappa/2 +3\eps/2} \cN^2_k \lesssim \sum_{p \in \Lambda^*} p^2 a_p^* a_p + N^\kappa \cN^2_k.  \] 
Since $[\cN_k,T_q] = 0$ for $k\neq q$, conjugating with $T_{c, S' \backslash \{k \}} = \prod_{p \in S' \backslash \{ k \}} T_p$ yields 
\[ T_{c,S'}^* \sum_{p \in \Lambda^*} p^2 a_p^* a_p T_{c,S'} \lesssim T^*_{c, S' \backslash \{k \}} \sum_{p \in \Lambda^*} p^2 a_p^* a_p T_{c, S' \backslash \{k \}} + N^\kappa \cN^2_k.  \]
Iterating, we arrive, on the range of $\mathds{1}_{\{\cN_{S^c} = 0\}}$, at
\begin{equation}\label{eq:kin-rough} T_{c,S'}^* \sum_{p \in \Lambda^*} p^2 a_p^* a_p T_{c,S'} \lesssim \sum_{p \in S} p^2 a_p^* a_p  + N^\kappa \sum_{k \in S} \cN^2_k \lesssim N^{\kappa+2\eps} \sum_{k \in S} \cN^2_k.
\end{equation} 

Next, we show \eqref{eq:Tc_changed}. Let $k \in S$.
By Lemmas \ref{lem:Tc_kin}, \ref{lem:Q3_Tp} and \ref{lem:Q4_H} we have
\begin{align} \label{eq:Tk_conclusion}
& \pm \Lambda_k \Big[ T_k^* \big( \sum_{p \in \Lambda^*} p^2 a_p^*a_p + Q_3^H + \widetilde{Q}_4^H \big) T_k - \sum_{p\in \Lambda^*}p^2a_p^*a_p -  Q_3^H - \widetilde{Q}_4^H 
- Z^{(k)}\Big] \Lambda_k \nn
\\
& \qquad \lesssim  N^{-1+5\kappa/2+3\vep/2} \cN_k^2
+ N^{-5+15\kappa/2+7\vep} \Big(\sum_{\substack{p,r \in H \\ p+r \in S}} N^{2\vep} |p|^2 a_r^*a_p^*a_p a_r + N^\kappa \cN_S \Big) (\cN_k+1), 
\end{align}
where we denoted
\begin{align}\label{eq:Zk} 
Z^{(k)} = \sum_{r \in H_k} N \Big[2\hat{V}_N(r) + 2r^2 \vphi_r + \sum_{t \in H_k}  \hat{V}_N(r-t) \vphi_t \Theta^{(2)}_{k,t}  \Big]
(\vphi_r + \vphi_{r+k}) \cN_k \Theta^{(2)}_{k,r}.
\end{align}
Conjugating with $T_{c,S \backslash \{ k \}}$, we obtain, on the range of $\mathds{1}_{\{\cN_{S^c} = 0\}}$,  
\begin{align}\label{eq:iter-main} 
 \pm \Big[ T_c^*  \big( &\sum_{p \in \Lambda^*} p^2 a_p^*a_p + Q_3^H + \widetilde{Q}_4^H \big) T_c  \nn 
 \\
  &\hspace{1cm}  - T_{c,S\backslash \{ k \}}^* \big( \sum_{p\in \Lambda^*}p^2a_p^*a_p + Q_3^H + \widetilde{Q}_4^H \big) T_{c,S\backslash \{ k \}} - T_{c,S\backslash \{ k \}}^* Z^{(k)} T_{c,S\backslash \{ k \}} \Big] \nn
\\
\lesssim \; &N^{-1+5\kappa/2+3\vep/2} \cN_k^2 \nn
\\ &+ N^{-5+15\kappa/2+7\vep} T_{c,S\backslash \{ k \}}^* \Big[ \sum_{\substack{p,r \in H \\ p+r \in S}}   N^{2\vep} |p|^2 a_r^*a_p^*a_p a_r + N^\kappa \cN_S \Big]  (\cN_k+1) T_{c,S\backslash \{ k \}}. 
\end{align} 
Using Lemma \ref{lm:mono} (note that the $r$-sum runs over the $S$-neighborhood of $-p$) we find 
\begin{align} \label{eq:NH^2_monogamy}
\tr\Big( T_{c,S\backslash \{ k \}}^* \sum_{p,r \in H} &N^{2\varepsilon} |p|^2 \1_S(p+r) a_r^*a_p^*a_pa_r \mathcal N_k T_{c,S\backslash \{ k \}} \Gamma \Big) \nn \\
	&\leq  \tr T_{c,S\backslash \{ k \}}^* N^{2\varepsilon} \sum_{p \in H} |p|^2   a_p^*a_p \mathcal N_k T_{c,S\backslash \{ k \}} \Gamma \nn
	\\
	&\leq N^{\kappa+4\vep} \tr \sum_{p \in S} \cN_p^2 \, \cN_k \Gamma,
\end{align}
where we used \eqref{eq:kin-rough} in the last inequality. 

Let us return to (\ref{eq:iter-main}) and consider the term $T^*_{c,S\backslash \{ k \}} Z^{(k)} T_{c,S\backslash \{ k \}}$. We claim that 
\begin{align}\label{eq:Tc_iteration}
 \pm\Lambda_q \big(  T_q^*  Z^{(k)} T_q   -  Z^{(k)} \big)\Lambda_q 
 \lesssim N^{-5+17\kappa/2+7\vep} \cN_q \cN_k.
\end{align}
for all $q \in S \backslash \{ k \}$. Since $T_\ell^* \cN_q \cN_k T_\ell = \cN_q \cN_k$ for all $\ell \in S \backslash \{ k, q \}$, (\ref{eq:Tc_iteration}) immediately implies that  
\begin{align}\label{eq:Tc_iteration2}
\pm \tr \big(  T_{c,S \backslash \{ k \}}^*  Z^{(k)} T_{c,S \backslash \{k \}}   -  Z^{(k)} \big)\Gamma 
 \lesssim N^{-5+17\kappa/2+7\vep} \tr \mathcal N_S \cN_k \Gamma.
\end{align}

Inserting \eqref{eq:NH^2_monogamy} and \eqref{eq:Tc_iteration2} in \eqref{eq:iter-main} yields 
\[  \begin{split} 
\tr \, T_c^*  \big( \sum_{p \in \Lambda^*} p^2 &a_p^*a_p + Q_3^H + \widetilde{Q}_4^H \big) T_c \Gamma \\ \leq  \; &\tr \, T_{c,S\backslash \{ k \}}^* \big( \sum_{p\in \Lambda^*}p^2a_p^*a_p + Q_3^H + \widetilde{Q}_4^H \big) T_{c,S\backslash \{ k \}} \Gamma +  \tr \, Z^{(k)} \Gamma 
\\ &+ C N^{-1+5\kappa/2+3\vep/2} \tr \cN_k^2 \Gamma  
+ C N^{-5+17\kappa/2+11\vep} \, \tr \sum_{p \in S} \cN_p^2 (\cN_k + 1) \Gamma  \end{split} \]
Iterating (and using the assumption $\Gamma = \mathds{1}_{\{\cN_{S^c}=0\}} \Gamma \mathds{1}_{\{\cN_{S^c} = 0\}}$), we conclude that 
\begin{equation}\label{eq:iter-end} \begin{split} 
\tr \, & T_c^*  \big( \sum_{p \in \Lambda^*} p^2 a_p^*a_p + Q_3^H + \widetilde{Q}_4^H \big) T_c \Gamma  \\ 
&\leq \;  \tr  \Big[ \sum_{p\in S} p^2 a_p^*a_p + \sum_{k \in S} Z^{(k)}  \Big] \Gamma + CN^{-1+5\kappa/2+3\vep/2} \tr \sum_{p\in S}\cN_p^2 \Gamma 
\\
& \qquad + C N^{-5+17\kappa/2+11\vep} \, \tr \sum_{p \in S} \cN_p^2 (\cN_S + N^{3\kappa/2+3\eps}) \Gamma.  
\end{split} \end{equation}

In the definition (\ref{eq:Zk}) of $Z^{(k)}$, we apply the scattering equation (\ref{eq:vphi}). With (\ref{eq:scatlength}), Lemma~\ref{lem:vphi_prop} and $|\hat{V}_N(r-k)-\hat{V}_N(r)| \lesssim |k| N^{-2+2\kappa} \leq N^{-2+5\kappa/2+\vep}$, we find
\begin{align*}
\sum_{r \in H_k} N \Big[2\hat{V}_N(r) + &2r^2 \vphi_r + \sum_{t \in H_k}  \hat{V}_N(r-t) \vphi_t \Big]
(\vphi_r + \vphi_{r+k}) \\ &= \sum_{r \in H_k} N \hat{V}_N(r)
(\vphi_r + \vphi_{r+k}) + \mathcal{O}(N^{\kappa - \vep}) = 2N^\kappa \big( 8\pi\mathfrak{a} - \hat{V}(0)\big) + \mathcal{O}(N^{\kappa - \vep}).
\end{align*}
With (\ref{eq:iter-end}), this implies \eqref{eq:Tc_changed}.

It remains to show \eqref{eq:Tc_iteration}. With \eqref{eq:comm_a_Theta}, we obtain
\begin{align*}
\Lambda_q T_q^* \Theta_{k,r}^{(2)} \Theta_{k,t}^{(2)} T_q \Lambda_q
	&= \Lambda_q \left( \cos(X_q) \Theta_{k,r}^{(2)}\Theta_{k,t}^{(2)}  \cos(X_q) + \frac{\sin(X_q)}{X_q} \mathcal B_q^\circ \Theta_{k,r}^{(2)}\Theta_{k,t}^{(2)}  \mathcal B_q^\sharp\frac{\sin(X_q)}{X_q} \right) \Lambda_q \\
	&=\Theta_{k,r}^{(2)}\Theta_{k,t}^{(2)} \Lambda_q \left( \cos(X_q)^2 + \left(\frac{\sin(X_q)}{X_q}\right)^2 (X_q^{(r,-(k+r),-(k+t),t)})^2\right) \Lambda_q \\
	&=\Theta_{k,r}^{(2)}\Theta_{k,t}^{(2)} \Lambda_q \left( 1 - \left(\frac{\sin(X_q)}{X_q}\right)^2 \delta_q^{(r,-(k+r),-(k+t),t)}\right) \Lambda_q,
\end{align*}
where $\delta_q^{(r,-(k+r),-(k+t),t)}$ and $X_q^{(r,-(k+r),-(k+t),t)}$ were defined in (\ref{eq:def_delta_4}). Using (\ref{eq:est_delta_4}) we find 
\begin{align} \label{eq:est_Yk_Bi}
\pm \Lambda_q \left(T_q^* \Theta_{k,r}^{(2)} \Theta_{k,t}^{(2)} T_q - \Theta_{k,r}^{(2)} \Theta_{k,t}^{(2)}\right) \Lambda_q \lesssim N^{-5 + 15\kappa/2 + 7\varepsilon} \mathcal N_q.
\end{align}
From this and
$$
\sum_{r \in H_k} N \Big[2\hat{V}_N(r) + 2r^2 \vphi_r + \sum_{t \in H_k}  \hat{V}_N(r-t) \vphi_t \Big] (\vphi_r + \vphi_{r+k})
\lesssim N^\kappa,
$$
which follows from \Cref{lem:vphi_prop}, we find \eqref{eq:Tc_iteration}.

Finally, let us show \eqref{eq:apriori_Q3_Q4}. The bound on the kinetic term $\sum p^2 a_p^*a_p$ was shown in \eqref{eq:kin-rough}.
By Cauchy--Schwarz we find
\begin{equation} \label{eq:Q3_bound_Q4}
\pm Q_3^H \leq \frac{1}{2} Q_4^H + C N^\kappa \cN_S,
\end{equation}
which implies
$$
Q_4^H = Q_4^H + 2Q_3^H - 2Q_3^H \leq 2(Q_4^H + Q_3^H) + CN^\kappa \cN_S.
$$
This, together with \eqref{eq:Tc_changed}, \Cref{lem:Q4_HS} and \Cref{lem:Tc_N}, implies the bound on the quartic term $Q_4^H$. The bound on the cubic term then follows again by \eqref{eq:Q3_bound_Q4}.
\end{proof}

We are now set to conclude the proof of \Cref{lem:Tc_main}. Recalling the assumptions on $\Gamma$ and the decomposition \eqref{eq:Q3-deco} and collecting the bounds from \Cref{lem:TLc_main}, \Cref{lem:Tc_nondiag}, \Cref{lm:Q3SM}, \Cref{lem:Q4_HS}, \Cref{lem:Q4S}, \Cref{lem:Q4M1},  \Cref{lem:Q4M2} and \Cref{lem:Tc_changed} , we have
\begin{equation}\label{eq:Tc-main1} \begin{split} 
\tr \, T_c^* e^{-\cB_1} &\weyl^* \cH_N \weyl e^{\cB_1} T_c \Gamma \\ \leq \; &4 \pi \mathfrak{a} N^{1+\kappa} - 8 \pi \mathfrak{a} N^\kappa (N-N_0) + \sum_{p \in S} \frac{(4\pi \mathfrak{a} N^\kappa)^2}{p^2}
\\ &+ \tr \Big[\sum_{p\in S} p^2 a_p^*a_p + 16 \pi \mathfrak{a}N^\kappa \cN_S + 4 \pi \mathfrak{a} N^\kappa \sum_{p \in S} \tr \, \big(a_p^* a_{-p}^* + \text{h.c.} \big)\Big] \Gamma + \delta (\Gamma) 
\end{split} \end{equation} 
where  
\begin{equation}\label{eq:delta1}
\begin{split} 
\delta ( &\Gamma) \lesssim \, \tr \, T_c^* \cE_1 T_c \Gamma + (\sqrt{N_0/N}-1)\tr T_c^* Q_3^H T_c \Gamma +
N^{-2+4\kappa +\eps} \tr \, \big (\cN_S + N^{3\kappa/2+3\eps} \big) \Gamma
 \\ & \qquad + N^{-1+5\kappa/2+3\vep} \tr T_c^* \cN_H^2 T_c \Gamma + N^{-1+\kappa} \tr T_c^* \cN  ( \cN_S + N^{3\kappa/2+3\eps}) T_c \Gamma 
 + \delta_1(\Gamma).
\end{split} \end{equation}
With (\ref{eq:claimB1E}), \eqref{eq:Tc_changed_error} and with Lemma \ref{lem:Tc_N}, we find 
\[ \begin{split}  \delta ( &\Gamma) \\ \lesssim &\, N^{-3\eps} \tr \,  T_c^* \sum_{p \in \Lambda^*} p^2 a_p^* a_p T_c \Gamma + N^{-5\eps/2} \tr \, T_c^*  Q_4 T_c \Gamma + (\sqrt{N_0/N}-1)\tr T_c^* Q_3^H T_c \Gamma
 \\
 & + \big(N^{-\kappa/2-\vep/2} + N^{-7+13\kappa+6\vep}\big) \tr \, \cN_S^2 \Gamma + N^{-1+\kappa+6\vep} \tr \big(\cN_S + N^{3\kappa/2+3\vep}\big) \sum_{p\in S}\cN_p^2 \Gamma + N^{5\kappa/2-\vep/2}, 
\end{split} \] 
if $1/2 < \kappa < 8/15 - 2\vep/3$ (note again that in particular $\vep < 1/20$). For the first three terms on the right hand side we use \eqref{eq:apriori_Q3_Q4} and find that they can be absorbed in the error terms that are already present.

\qed

\subsection{Proof of Proposition \ref{prop:NN2}} 
\label{sec:propN} 

From (\ref{eq:weyl2}) and (\ref{eq:bogo1}), we find
\begin{align}\label{eq:propNN2_1}
e^{-\cB_1} & \weyl^* \cN \weyl e^{\cB_1}
= e^{-\cB_1} \big(N_0 + \cN + \sqrt{N_0}(a_0^*+a_0) \big) e^{\cB_1} \nn
\\
&= N_0 + \sum_{p \in \Lambda^*} \left[(c_p^2+s_p^2)a_p^*a_p + c_p s_p (a_p^*a_{-p}^* + a_pa_{-p}) + s_p^2 \right] + \sqrt{N_0}(a_0^*+a_0). 
\end{align}
From \eqref{eq:sp_def} and (\ref{eq:Q3SM}), we find that on the range of $T_c \Gamma$, the second term in \eqref{eq:propNN2_1} simplifies to
\begin{align}
\sum_{p \in \Lambda^*} \left[(c_p^2+s_p^2)a_p^*a_p + c_p s_p (a_p^*a_{-p}^* + a_pa_{-p})  + s_p^2 \right]  
	&= \mathcal N +2  \sum_{p \in H} N_0^2 \ph_p^2 a_p^*a_p + \hspace{-.3cm} \sum_{|p| > N^{\kappa/2+\vep}} \hspace{-.3cm} N_0^2 \ph_p^2.  \label{eq:propNN2_2}
\end{align}
Here we used \Cref{lm:mono} to argue that $\tr \,  T_c^* a_{p}^*a_{-p}^* T_c  \, \Gamma = 0$ for $p\in H$ as $a_{p}^*a_{-p}^*$ would leave two excitations in $H$ which are no longer connected. Also, from the definition of $\Gamma$ in (\ref{eq:G-def}), we clearly have 
\begin{align*}
\tr \left(T_c^* a_0^*   T_c  \Gamma \right) = \tr \left( a_0^*  \Gamma \right) =  0.
\end{align*}
Recall from \eqref{eq:TN-inv} that on the range of $\Gamma$ we have
\begin{align*}
T_c^* \mathcal N T_c \geq T_c^* (\mathcal N_S + \frac{\mathcal N_H}{2}) T_c = \mathcal N_S + \frac{\mathcal N_H}{2} = \mathcal N_S.
\end{align*}
Then we obtain that
\begin{align*}
\tr \cN \Gamma_N \geq \tr \cN_S  \Gamma + N_0.
\end{align*}
To prove the upper bound, we combine the above identities (\ref{eq:propNN2_1}) and (\ref{eq:propNN2_2}) with Lemma \ref{lem:Tc_N} and $\varphi_p \lesssim p^{-2} N^{-1+\kappa}$ from Lemma \ref{lem:vphi_prop}, as well as $N_0 \lesssim N$, which follows directly from the definition \eqref{eq:N0-def} and \Cref{lm:G-prop}.
We find that on the range of $\Gamma$
\begin{align*}
T_c^*e^{-\cB_1} \weyl^* \mathcal N \weyl e^{\cB_1} T_c - \mathcal N_S - N_0  \lesssim N^{-2+3\kappa+\vep} \cN_S + N^{3\kappa/2 - \vep},
\end{align*}
and from $ \tr \cN_S  \Gamma \lesssim N^{3\kappa/2}$, see Lemma \ref{lm:G-prop}, we obtain \Cref{prop:NN2}.

\qed


\appendix
%
%

\section{Proof of Lemma \ref{lem:vphi_prop}} 
\label{app:scat}

For $N$ large enough, $V_N$ is supported in $\Lambda$ so that we may consider it as a function on the torus, i.e.\ $V_N \in L^2(\Lambda)$. Since $V_N \geq 0$ we may invert $-\Delta + \frac{1}{2}V_N$ on $P^\perp L^2(\Lambda)$, where $P^\perp = 1-|\mathds{1}_\Lambda\rangle \langle \1_\Lambda|$ is the projection onto
the orthogonal complement of the zero mode. 
We define 
\begin{equation}\label{eq:def_vphi_pos}
\check{\vphi} = -\frac{1}{2}\frac{1}{P^\perp (-\Delta + \frac{1}{2}V_N)P^\perp} P^\perp V_N \in P^\perp L^2(\Lambda)
\end{equation}
and readily find that in momentum space $\vphi$ satisfies \eqref{eq:vphi}. Moreover, from \eqref{eq:def_vphi_pos} we obtain
$$
\braket{\check{\vphi},-\Delta \check{\vphi}} = -\frac{1}{2}\braket{\check{\vphi},V_N(1+\check{\vphi})} \leq -\frac{1}{2}\braket{\check{\vphi}, V_N}
\leq \frac{1}{2} \|\nabla \check{\vphi}\|_2 \|(P^\perp \nabla P^\perp)^{-1} V_N\|_2.
$$
In particular $\check{\vphi} \in H^1(\Lambda)$ since $\braket{\check{\vphi},V_N} \leq \|\vphi\|_2 \|V_N\|_2 < \infty$.
We compute
\begin{align*}
\|(P^\perp \nabla P^\perp)^{-1} V_N\|_2^2 &= \sum_{p\neq 0}\frac{\hat{V_N}(p)^2}{p^2} = \sum_{0<|p|<N^{1-\kappa}} \frac{\hat{V_N}(p)^2}{p^2} + \sum_{|p|\geq N^{1-\kappa}} \frac{\hat{V_N}(p)^2}{p^2} \\
&\leq C \|\hat{V}_N\|_\infty N^{1-\kappa} + C N^{-2+\kappa} \|V_N\|_2 \leq C N^{-1+\kappa}
\end{align*}
and conclude the estimate $\|p \vphi\|_2^2 \leq C N^{-1+\kappa}$ from \Cref{lem:vphi_prop}. From this and \eqref{eq:vphi} we deduce the pointwise bound $p^2\vphi_p \leq C N^{-1+\kappa}$. 
This readily implies the bounds $\|\vphi\|_\infty, \|\vphi_2\| \leq CN^{-1+\kappa}$ and the bounds for the cutoff version of $\vphi$.
Next, we estimate the $\ell^1$-norm of $\vphi$. 
For $q\in[1,6)$ we use H\"older respectively Young inequalities, note that we always estimate $\hat{V}_N$ in the $\ell^2-$norm, and find that
\begin{align*}
\|\vphi\|_q &\leq \|\vphi \1_{|p|<N^{-1+\kappa}}\|_q + \|\vphi \1_{|p| \geq N^{-1+\kappa}}\|_q 
\\
&\leq C_N + \frac{1}{2} \|p^{-2} \hat{V}_N \1_{|p| \geq N^{-1+\kappa}}\|_q + \frac{1}{2}\|p^{-2} (\hat{V}_N \ast \vphi) \1_{|p| \geq N^{-1+\kappa}}\|_q \\
&\leq C_N + C_N \|\vphi\|_s,
\end{align*}
if we choose $1 \leq s < \frac{6q}{6-q}$. By a bootstrap argument, choosing first $q=1,$ then $q = 6/5 - \vep_1$ and finally $q = 3/2-\vep_2$ for $0<\vep_1,\vep_2$ small enough we find that 
$$
\|\vphi\|_1 \leq C_N + C_N\|\vphi\|_{6/5-\vep_1} \leq C_N + C_N \|\vphi\|_{3/2-\vep_2} < \infty.
$$
Next we show the exact scaling behavior in $N$. From \eqref{eq:vphi} and $p^2 \vphi_p \leq C N^{-1+\kappa}$ we obtain for any constant $D>0$
\begin{align*}
\|\vphi\|_1 &= \|\vphi \1_{|p|<DN^{-1+\kappa}}\|_1 + \|\vphi \1_{|p|\geq DN^{-1+\kappa}}\|_1
\\
&\leq C_D + \frac{1}{2} \| p^{-2} \hat{V}_N \1_{|p|\geq DN^{-1+\kappa}}\|_1 + \frac{1}{2} \| p^{-2} \hat{V}_N \ast \vphi \1_{|p|\geq DN^{-1+\kappa}}\|_1
\\
&\leq C_D + C D^{-1/2} \|\vphi_1\|.
\end{align*}
Thus, choosing $D$ large enough, we obtain $\|\vphi\|_1 \leq C$.

Next, we want to show \eqref{eq:scat_comparison}. We want to compare the scattering solution $\vphi$ on the torus to the one on $\R^3$. It is well known that the optimization problem
$$
\inf \left\{\int_{\R^3} |\nabla w|^2 + \frac{1}{2}\int_{\R^3} V |1+w|^2 : w\in H^1(\R^3) \right\}
$$
has a unique minimizer $\omega$ and that the minimal value is given by $4\pi\mathfrak{a}$. Moreover, $\omega$ solves the scattering equation
\begin{align*}
-\Delta \omega + \frac{1}{2} V (1+\omega) = 0
\end{align*}
and it holds that $|\omega| \leq \frac{C}{1+|\,\cdot\,|}$, as well as $|\nabla \omega| \leq \frac{C}{1+|\,\cdot\,|^2}$.
As a mean to compare $\vphi$ and $\omega$ we rescale and truncate the full space scattering solution. Let $\omega_N(x) = \omega(N^{1-\kappa} x)\chi(x)$, with $0 \leq \chi \leq 1$ a smooth and radial function such that $\chi \equiv 1$ for $|x|\leq 1/2$ as well as $\chi \equiv 0$ for $|x|\geq 1$. We obtain that $\omega_N$ solves the following equation
\begin{equation} \label{eq:scat_wN}
(-\Delta + \frac{1}{2}V_N) \omega_N = -\frac{1}{2} V_N -2N^{1-\kappa} (\nabla \omega)(N^{1-\kappa}\cdot)\nabla \chi - \omega(N^{1-\kappa} \cdot) \Delta \chi =: -\frac{1}{2}V_N + \frac{1}{2} \epsilon_N.
\end{equation}
From the estimate on $\omega$ we find $\|\omega_N\|_1 \leq C N^{-1+\kappa}$ and $\|\epsilon_N\|_1 \leq C N^{-1+\kappa}$.
Observe that $8\pi\mathfrak{a} = \int_{\R^3}V(1+\omega) = N^{1-\kappa} \int_\Lambda V_N(1+\omega_N)$ and $8\pi\mathfrak{a}_N = N^{1-\kappa} \int_\Lambda V_N(1+\check{\vphi})$. Then with \eqref{eq:scat_wN} and \eqref{eq:def_vphi_pos} we conclude that
\begin{align*}
8\pi|\mathfrak{a} - \mathfrak{a}_N| &= N^{1-\kappa} |\braket{V_N,\omega_N-\check{\vphi}}_{L^2(\Lambda)}| 
\leq N^{1-\kappa} |\braket{P^\perp V_N,\omega_N-\check{\vphi}}| + N^{1-\kappa} \|V_N\|_1 \|\omega_N\|_1
\\
&\leq N^{1-\kappa} |\braket{2 P^\perp (-\Delta+\frac{1}{2}V_N)\check{\vphi},\omega_N-\check{\vphi}}| + C N^{-1+\kappa}
\\
&\leq N^{1-\kappa} |\braket{2 \check{\vphi},(-\Delta+\frac{1}{2}V_N)(\omega_N-\check{\vphi})}| + N^{1-\kappa} |\braket{\check{\vphi},V_N P\omega_N}| + C N^{-1+\kappa}
\\
& \leq N^{1-\kappa} |\braket{\check{\vphi},\epsilon_N}| + N^{1-\kappa} \|\check{\vphi}\|_\infty \|V_N\|_1 \|\omega_N\|_1 + C N^{-1+\kappa}
\\
&\leq N^{1-\kappa} \|\check{\vphi}\|_\infty \|\epsilon_N\|_1 + CN^{-1+\kappa}
\\
&\leq CN^{-1+\kappa}.
\end{align*}


 \section{Proof of Proposition \ref{prop:localization}}
 \label{app:loc}
 
 We follow partly \cite{BCS}.
 
 \subsection{Dirichlet localization}
 Let $\Gamma_L$ be a density matrix over $\cF (\Lambda_L)$, satisfying periodic boundary conditions such that (\ref{eq:cond-N}) holds true. By the spectral theorem, we can decompose $$\Gamma_L = \sum_{j\in J} \lambda_j |\Psi_j \rangle \langle \Psi_j|$$ where $\lambda_j \geq 0$ for all $j \in J$, $\sum_{j \in J} \lambda_j = 1$ and where $\Psi_j \in \cF (\Lambda_L)$ is an orthonormal family on $\cF (\Lambda_L)$. We have $\Psi_j = \{ \Psi_j^{(n)} \}_{n \geq 0}$ where $\Psi_j^{(n)}$ is $L$-periodic in all its coordinates (we think of $\Psi_j^{(n)}$ as a periodic function defined on $\bR^{3n}$). For fixed $j \in J$ and $u \in \Lambda_L$, we define now $\Psi^D_{j,u} = \{ (\Psi_{j,u}^D)^{(n)} \}_{n \in \bN}  \in \cF (\Lambda^u_{L+2\ell})$, where $\Lambda^u_{L+2\ell} = u + \Lambda_{L+2\ell}$ is a box of side length $L+2\ell$, with center at $u$, setting 
 \[ (\Psi_{j,u}^D)^{(n)} (x_1 ,\dots , x_n) = \Psi_j^{(n)} (x_1, \dots , x_n) \prod_{i=1}^n Q_{L,n} (x_i - u) \]
 with $Q_{L,\ell} (x) = \prod_{j=1}^3 q_{L,\ell} (x^{(j)})$ for all $x = (x^{(1)}, x^{(2)}, x^{(3)}) \in \bR^3$, with $q_{L,\ell} : \bR \to [0;1]$ defined through 
 \[ q_{L,\ell }(t) = \left\{ \begin{array}{ll} 
 \cos \big( \frac{\pi (t+L/2-\ell)}{4\ell} \big) \quad &\text{if } |t+L/2| \leq \ell \\ 
 1 \quad &\text{if } |t| < L/2 -\ell \\
 \cos \big(\frac{\pi (t-L/2+\ell)}{4\ell} \big) \quad &\text{if } |t-L/2| \leq \ell \\
 0 \quad &\text{otherwise} \end{array} \right. \]
By construction $\Psi_{j,u}^D$ satisfies Dirichlet boundary conditions on the box $\Lambda^u_{L+2\ell}$. 
\begin{lemma}[Dirichlet localization] \label{lm:diri}  
We have 
\begin{equation} \label{eq:psiipsij} \langle \Psi_{j,u}^D , \Psi_{i, u}^D \rangle_{\cF (\Lambda^u_{L+2\ell})} = \langle \Psi_j , \Psi_i \rangle = \delta_{ij} \end{equation}
Thus, $\Gamma_{L+2\ell, u}^D := \sum_{j\in J} \lambda_j |\Psi_{j,u}^D \rangle \langle \Psi_{j,u}^D|$ is a non-negative operator on $\cF (\Lambda_{L+2\ell}^u)$, with $\tr \Gamma_{L+2\ell,u}^D = 1$, with $S(\Gamma_{L+2\ell ,u}^D) = S (\Gamma_L)$ and  
\begin{equation}\label{eq:NjGamma} \tr \cN^j \Gamma_{L+2\ell,u}^D = \tr \cN^j \Gamma_L \end{equation} 
for all $j \in \bN$ and all $u \in \Lambda_L$. Moreover, there exists $\bar{u} \in \Lambda_L$ such that 
\begin{equation} \label{eq:HGuD} \tr \, \cH \Gamma_{L+2\ell, \bar{u}}^D \leq \tr \cH \Gamma_L + \frac{C}{L \ell} \tr \cN \Gamma_L 
\end{equation} 
\end{lemma} 
\begin{proof}
The proof is an adaptation of \cite[Lemma A.1]{BCS} to the mixed state setting. It is based on the observation that, for arbitrary $L$-periodic functions $\Phi \in L^1_\text{loc} (\bR)$, 
\[ \int_{-L/2-\ell}^{L/2+\ell} \Phi (t) |q_{L,\ell} (t)|^2 dt = \int_{-L/2}^{L/2}  \Phi (t) dt \, .  \]
This shows that in particular that
\begin{align*}
\langle (\Psi_{j,u}^D)^{(n)} , (\Psi_{i,u}^D)^{(n)} \rangle 
	&=  \langle \Psi_{j}^{(n)} , \Psi_{i}^{(n)} \rangle \\
\left\langle (\Psi_{j,u}^D)^{(n)} , \sum_{1\leq i < j \leq n} V(x_i-x_j)  (\Psi_{i,u}^D)^{(n)} \right\rangle 
	&=  \left\langle \Psi_{j}^{(n)} , \sum_{1\leq i < j \leq n} V(x_i-x_j) \Psi_{i}^{(n)} \right\rangle \\
\end{align*}
and thus implies (\ref{eq:psiipsij}) and (\ref{eq:NjGamma}). 
The fact that $S(\Gamma_{L+2\ell,u}^D) = S (\Gamma_L)$ is clear, since the operators $\Gamma_{L+2\ell,u}^D$ and $\Gamma_L$ have the same eigenvalues. As for (\ref{eq:HGuD}), we  proceed as in the proof of Lemma \cite[Lemma A.1]{BCS} to show that, for all $j \in J$ and all $u \in \Lambda_L$,  
\[ \langle \psi_{j,u}^D , \cH \psi_{j,u}^D \rangle \leq \langle \psi_j , \cH \psi_j \rangle + \frac{C}{\ell^2} \sum_{n \in \bN} n \int dx_1 \, \wt{\chi}_{L,\ell} (x_1 - u) \int_{\Lambda_L^{n-1}} dx_2 \dots dx_n \, | \Psi_j^{(n)} (x_1, \dots , x_n) |^2  \]
where we set $\wt{\chi}_{L,\ell} (x) = \sum_{k=1}^3 \chi_{L,\ell} (x^{(k)}) \prod_{j \not = k} \chi_{[-L/2;L/2]}  (x^{(j)})$ and we used the notation \[ \chi_{L,\ell} = \chi_{[-L/2-\ell ; -L/2+\ell]} + \chi_{[L/2 -\ell ; L/2 + \ell]}.\]  Summing over $j \in J$ with the weights $\lambda_j$, we obtain 
\[ \tr \, \cH \Gamma_{L+2\ell,u}^D \leq \tr \cH \Gamma_L + \frac{C}{\ell^2} \sum_{j \in J} \lambda_j \sum_{n \in \bN} n \int dx_1 \, \wt{\chi}_{L,\ell} (x_1 - u) \int_{\Lambda_L^{n-1}} dx_2 \dots dx_n |\Psi_j^{(n)} (x_1, \dots , x_n) |^2 \]
for all $u \in \Lambda_L$. Averaging over $u$, using $\| \wt{\chi}_{L,\ell} \|_1  \leq C L^2 \ell$, we conclude that there exists $\bar{u} \in \Lambda_L$ such that 
\[ \tr \, \cH \Gamma_{L+2\ell,\bar{u}}^D \leq \tr \cH \Gamma_L + \frac{C}{L \ell} \tr \cN \Gamma_L \]
as claimed.
\end{proof} 

\subsection{Patching up the boxes}

For $j \in J$, we define now $\Psi_{j,L+2\ell}^D \in \cF (\Lambda_{L+2\ell})$ setting $(\Psi_{j,L+2\ell}^D)^{(n)} (x_1, \dots , x_n) =  (\Psi^D_{j,\bar{u}})^{(n)} (x_1 - \bar{u}, x_2 - \bar{u}, \dots , x_n - \bar{u})$, with $\bar{u}$ chosen as in Lemma \ref{lm:diri}. Correspondingly, we introduce the density matrix $\Gamma_{L+2\ell}^D \in \cF (\Lambda_{L+2\ell})$ on the box $\Lambda_{L+2\ell} = [-L/2-\ell ; L/2 + \ell]^3$, satisfying Dirichlet boundary conditions.  This allows us to replicate $\Gamma_{L+2\ell}^D$ into several adjacent copies of $\Lambda_{L+2\ell}$, separated by corridors of size $R > 0$ (to avoid interactions between particles in diffierent boxes). For $t \in \bN$, let $$\wt{L} = t (L+2\ell + R).$$ We think of the thermodynamic box $\Lambda_{\wt{L}}$ as the (almost) disjoint union of $t^3$ shifted copies of the small box $\Lambda_{L+2\ell+R}$, centered at 
\[ (-\wt{L}/2 , - \wt{L}/2, -\wt{L}/2) + (L+2\ell+R) (i_1 - 1/2 , i_2 - 1/2 , i_3 - 1/2) \]
for $i_1, i_2, i_3 \in \{ 1, \dots , t \}$. Let $(c_i)_{i=1,\dots , t^3}$ denote an enumeration of the centers of the boxes. Using $L^2 (\Lambda_{\wt{L}}) = \oplus_{i=1}^{t^3} L^2 (\Lambda_{L+2\ell+R}^{c_i})$ and the canonical identification (see for eg. \cite[Theorem 16]{Derezinski}) $$\cF (\Lambda_{\wt{L}}) \simeq \bigotimes_{i=1}^{t^3} \cF (\Lambda_{L+2\ell+R}^{c_i}),$$ we can define on $\cF (\Lambda_{\wt{L}})$ the state
\begin{equation}\label{eq:GwtL}  \Gamma^D_{\wt{L}} \simeq  \Gamma_{L+2\ell+R, c_1}^D \otimes \dots \otimes \Gamma_{L+2\ell + R, c_{t^3}}^D. \end{equation} 
Here the tensor products are symmetric.
Then we have $\tr \Gamma^D_{\wt{L}} = 1$, 
\[ \tr \cN \Gamma^D_{\wt{L}} = t^3 \tr \cN \Gamma_{L+2\ell+R}^D
 \]
as can be seen decomposing $\cN = \sum_{i=1}^{t^3} \cN_i$, where $\cN_i$ measures the number of particles in the box with side length $(L+2\ell+R)$ centered at $c_i$ 
Moreover, 
\[ \tr \cH \Gamma^D_{\wt{L}} = t^3 \tr \cH \Gamma^D_{L+2\ell+R} \]
and $S (\Gamma^D_{\wt{L}}) = t^3 S (\Gamma_{L+2\ell+R}^D)$, as follows by noticing that, if 
$(\lambda_j)_{j\in J}$ are the eigenvalues of $\Gamma_{L+2\ell+R}$, then the products $\lambda_{j_1} \dots \lambda_{j_{t^3}}$ are the eigenvalues of $\Gamma^D_{\wt{L}}$. In particular, we obtain 
\[ \tr \cH \Gamma^D_{\wt{L}} - T S (\Gamma^D_{\wt{L}}) = t^3 \big[  \tr \cH \Gamma_{L+2\ell+R}^D - T S (\Gamma_{L+2\ell+r}^D) \big] \]
 The state $\Gamma^D_{\wt{L}}$ is a good trial state for the free energy in the grand canonical ensemble. 
 
 \subsection{From grand canonical to canonical ensemble}
 
 As a last step, infer a bound on the canonical free energy, with a fixed number of particles, from the energy of a grand-canonical state.
 

\begin{lemma}\label{lm:grand2} Suppose there exists a sequence $\Gamma_{\tilde{L}}^D$ of density matrices on $\cF (\Lambda_{\tilde L})$, parametrized by $\tilde L > 0$ and satisfying Dirichlet boundary conditions and such that 
\begin{equation}\label{eq:assum-G} \lim_{{\tilde L}\to\infty} \frac{1}{\tilde{L}^3}\tr \, \cN \Gamma_{\tilde{L}}^D = {\tilde{\rho}}.
\end{equation} 
Then 
\[ f ({\tilde{\rho}}, T) \leq \liminf_{\tilde{L} \to \infty} \frac{1}{{\tilde{L}}^3} \left[ \tr \cH \Gamma_{\tilde{L}}^D - T S (\Gamma_L^D) \right]. \]
\end{lemma} 
\begin{proof}
We use the equivalence of ensembles: the free energy ${\tilde{\rho}} \mapsto f({\tilde{\rho}},T)$ is convex and is given by the Legendre transform of the pressure, see for example \cite{Ruelle}
\begin{align} \label{eq:legendre}
f({\tilde{\rho}},T) = \sup_{\mu \in \mathbb{R}} \{\mu{\tilde{\rho}} + f_{GC}(\mu,T)\}
\end{align}
with
\begin{align*}
f_{GC}(\mu,T) = \lim_{{\tilde{L}}\to\infty} -\frac{T}{\tilde{L}^3} \log \tr e^{-\frac{1}{T}(\mathcal H - \mu \mathcal N)}.
\end{align*}
Then, for all $\mu \in \mathbb{R}$, using Gibbs' variational principle, we obtain
\begin{align*}
f_{GC}(\mu,T) \leq \liminf_{{\tilde{L}} \to \infty} \frac{1}{\tilde{L}^3} \left[ \tr (\cH-\mu) \Gamma_{\tilde{L}}^D - T S (\Gamma_{\tilde{L}}^D) \right] = \liminf_{{\tilde{L}} \to \infty} \frac{1}{{\tilde{L}}^3} \left[ \tr \cH \Gamma_{\tilde{L}}^D - T S (\Gamma_{\tilde{L}}^D) \right] -\mu {\tilde{\rho}}.
\end{align*}
The above inequality being valid for all $\mu$, we deduce the claim using (\ref{eq:legendre}).
\end{proof}

\subsection{Proof of Proposition \ref{prop:localization}}
We can now conclude the proof of  \cshref{prop:localization}. 

Let $\Gamma_L$ be a density matrix on $\cF (\Lambda_L)$, satisfying periodic boundary conditions and 
\[ \tr \cN \Gamma_L < \infty. \]
With Lemma \ref{lm:diri}, we find a density matrix $\Gamma^D_{L+2\ell}$ on $\cF (\Lambda_{L+2\ell})$ satisfying Dirichlet boundary conditions, with  
\[ \tr \cN \Gamma^D_{L+2\ell} =\tr \cN \Gamma_L\]
and 
\[ \tr \cH \Gamma_{L+2\ell}^D - T S (\Gamma_{L+2\ell}^D)  \leq \tr \cH \Gamma_L - T S (\Gamma_L) + \frac{C}{L\ell} \tr \cN \Gamma_L. \]
From (\ref{eq:GwtL}), we find, for any $t \in \bN$, a density matrix $\Gamma^D_{\wt{L}}$ on $\cF (\Lambda_{\wt{L}})$, with $\wt{L} = t (L+2\ell +R)$, satisfying Dirichlet boundary conditions, with 
\[
\frac{1}{\wt{L}^3} \left( \tr \cH \Gamma_{\wt{L}}^D - T S(\Gamma_{\wt{L}}^D)\right) \leq \frac{1}{(L+2\ell +R)^3} \left(  \tr \cH \Gamma_L - T S (\Gamma_L) \right) + \frac{C}{(L\ell)(L+2\ell +R)^3} \tr \cN \Gamma_L. \]
Denoting
\[ \tilde{\rho} := \frac{1}{\wt{L}^3}\tr \cN \Gamma^D_{\wt{L}} = \frac{1}{{(L+2\ell +R)}^3}\tr \cN \Gamma_{{L}},
 \]
which is independent of $t$, we obtain by Lemma \ref{lm:grand2}, letting $\wt{L} \to \infty$ (i.e.\ $t \to \infty$), that 
\[ \begin{split} f(\tilde \rho, T) &\leq \lim_{\wt{L} \to \infty} \frac{1}{\wt{L}^3}  \left[ \tr \cH \Gamma_{\wt{L}}^D - T S(\Gamma_{\wt{L}}^D) \right] \\  &\leq \frac{1}{(L+2\ell +R)^3}  \left[ \tr \cH \Gamma_L - T S (\Gamma_L) \right] + \frac{C\tilde \rho}{L\ell}.
\end{split}  \]

\qed


\begin{thebibliography}{10}
%

\bibitem{AdhBreSch-21}
{\sc A.~Adhikari, C.~Brennecke, and B.~Schlein}, {\em {B}ose-{E}instein
  condensation beyond the {G}ross-{P}itaevskii regime.} Ann. Henri Poincar\'e,
  22 (2021), pp.~1163--1233.
  
  	\bibitem{Basti-22}
  	{\sc G.~Basti},
	{\em A second order upper bound on the ground state energy of a Bose gas beyond the Gross–Pitaevskii regime.}
	J. Math. Phys, 63 (2022)
  
  	\bibitem{BCS}
	{\sc G.~Basti, S.~Cenatiempo, and B.~Schlein}, {\em A new second order upper bound for the ground state energy of dilute {B}ose gases.} Forum Math. Sigma, 9 (2021), E74.
	
	\bibitem{BCGOPS-22}
	{\sc G.~Basti, S.~Cenatiempo, A.~Giuliani,  A.~Olgiati, G.~Pasqualetti, and B.~Schlein},
	{\em A simple upper bound for the ground state energy of a dilute Bose gas of hard spheres.} arXiv preprint arXiv:2212.04431.
	


	\bibitem{BBCS1}
{\sc C.~Boccato, S.~Brennecke, S.~Cenatiempo and B.~Schlein}, {\em The excitation
  spectrum of Bose gases interacting through singular potentials.} J. Eur.
  Math. Soc., 22 (2020), pp.~2331--2403.
  
	\bibitem{BBCS2}
	{\sc C.~Boccato, S.~Brennecke, S.~Cenatiempo and B.~Schlein}, {\em
	Bogoliubov theory in the Gross--Pitaevskii limit.} Acta Math. 222 (2019), pp. 219--335.
	
	\bibitem{BocDeuSto-23}
	{\sc C.~Boccato, A.~Deuchert, and D.~Stocker},
	{\em Upper bound for the grand canonical free energy of the Bose gas in the Gross--Pitaevskii limit.} arXiv preprint arXiv:2305.19173.
	
	
	  \bibitem{BreCapSch-21}
{\sc C.~Brennecke, M.~Caporaletti, and B.~Schlein}, {\em Excitation spectrum
  for {B}ose gases beyond the {G}ross-{P}itaevskii regime.} Rev. Math. Phys., 34 (2022) p.~2250027.
  
   \bibitem{BreSchSch-22}  {\sc C. Brennecke, B. Schlein, and S. Schraven}, {\em Bogoliubov theory for trapped bosons in the Gross--Pitaevskii regime.} Ann. Henri Poincar\'e 23 (2022), pp. 1583--1658.
	
	
		\bibitem{B} {\sc 
M. Brooks}, {\em Diagonalizing Bose Gases in the Gross--Pitaevskii Regime and Beyond.} Preprint arXiv:2310.11347.

%
%

		\bibitem{BocDeuSto-24}
	{\sc C.~Boccato, A.~Deuchert, and D.~Stocker},
	{\em Upper bound for the grand canonical free energy of the Bose gas in the Gross--Pitaevskii limit}
	SIAM J. Math. Anal., 56(2), 2611-2660.
	
	\bibitem{CapDeu-23}
	{\sc M.~Caporaletti, and A.~Deuchert},
	{\em Upper bound for the grand canonical free energy of the Bose gas in the Gross-Pitaevskii limit for general interaction potentials.}
	(2023) arXiv preprint arXiv:2310.12314.
	
	\bibitem{DeuSei-20}
	{\sc A.~Deuchert and R.~Seiringer},
	{\em Gross–Pitaevskii Limit of a Homogeneous Bose Gas at Positive Temperature.}
	Arch. Ration. Mech. Anal. 236, 1217–1271 (2020).

	\bibitem{DeuSeiYng-19}
	{\sc A.~Deuchert, R.~Seiringer, and J.~Yngvason},
	{\em Bose–Einstein Condensation in a Dilute, Trapped Gas at Positive Temperature.}
	Commun. Math. Phys. 368, 723–776 (2019).
	
	\bibitem{Derezinski}
	{\sc J.~Dereziński},
	{\em Introduction to Representations of the Canonical Commutation and Anticommutation Relations. In: Dereziński, J., Siedentop, H. (eds) Large Coulomb Systems. Lecture Notes in Physics, vol 695. Springer, Berlin, Heidelberg}, (2006)
	
	
	
%
	\bibitem{Dyson-57}
	{\sc F.~J.~Dyson}, {\em Ground state energy of a hard-sphere gas}. Phys. Rev., 106 (1957), pp.~20--26.
%
%
%
	\bibitem{FouSol-20}
	{\sc S.~Fournais and J.~P.~Solovej}, {\em The energy of dilute {B}ose gases}. Ann. of Math., 192 (2020), pp.~893--976.

	\bibitem{FouSol-23}
	{\sc S.~Fournais and J.~P.~Solovej}, {\em The energy of dilute {B}ose gases {II}: The general case}. Invent. Math.232(2023), no.2, 863–994.
%
%

	\bibitem{HabHaiNamSeiTri-23}
	{\sc F.~Haberberger, C.~Hainzl, P.T.~Nam, R.~Seiringer, and  A.~Triay}, 
	{\em The free energy of dilute Bose gases at low temperatures.} arXiv preprint arXiv:2304.02405.
	
	\bibitem{HST} {\sc C. Hainzl, B. Schlein, and A. Triay}, {\em Bogoliubov theory in the Gross--Pitaevskii limit: a simplified approach.}  Forum Math. Sigma, 10(e90) (2022), pp.~1--39. 
	
	\bibitem{LamTri-24}
	{\sc J.~Lampart and A.~Triay},
{\em The excitation spectrum of a dilute {B}ose gas with an impurity.}
arXiv preprint arXiv:2401.14911

	\bibitem{LHY-57}
	{\sc T.~D.~Lee, K.~Huang, and C.~N.~Yang}, {\em Eigenvalues and eigenfunctions of a {B}ose system of hard spheres and its low-temperature properties}. Phys. Rev., 106 (1957), pp.~1135--1145.
	

	\bibitem{LieYng-98}
	{\sc E.~H.~Lieb and J.~Yngvason}, {\em Ground state energy of the low density {B}ose gas}. Phys. Rev. Lett., 80 (1998), pp.~2504--2507.
	
	\bibitem{NamTri-23}
	{\sc P.~T. Nam and A.~Triay}, {\em Bogoliubov excitation spectrum of trapped Bose gases in the
	  Gross--Pitaevskii regime}. J. Math. Pures Appl. (9)176(2023), 18–101

	\bibitem{Ruelle}
	{\sc D.~Ruelle}, {\em Statistical mechanics. Rigorous results}. {Singapore: World Scientific. London: Imperial College Press}, (1999).
	
	\bibitem{Seiringer-08}
{\sc R.~Seiringer}, {\em Free energy of a dilute {B}ose gas: lower bound.}
  Commun. Math. Phys. 279 (2008), pp. 595--636. 

	\bibitem{YauYin-09}
	{\sc H.-T.~Yau and J.~Yin}, {\em The second order upper bound for the ground energy of a {B}ose gas}. J.~Stat. Phys., 136 (2009), pp.~453--503.

	\bibitem{Yin-10} {\sc J. Yin}, {\em Free energies of dilute Bose gases: upper bound.} J. Stat. Phys. 141 (2010), pp. 683--726. 
	\end{thebibliography}
\end{document}